\documentclass[11pt,english]{article}
\usepackage{color}
\usepackage[T1]{fontenc}
\usepackage[latin9]{inputenc}
\usepackage[a4paper]{geometry}
\usepackage{textcomp}
\usepackage{amsthm}
\usepackage{amsmath}
\usepackage{amssymb}
\usepackage{esint}
\usepackage{bbm}
\usepackage[colorlinks=true]{hyperref}
\usepackage{babel}
\usepackage{xcolor}

\usepackage{amscd}

\usepackage{amsfonts}
\usepackage{amsthm}
\usepackage{mathrsfs}
\usepackage{dsfont}
\usepackage{yfonts}

\usepackage{stmaryrd}

\usepackage{mathtools}

\usepackage{enumerate}
\usepackage{bm}
\usepackage{bbm}
\usepackage{verbatim}

\usepackage{enumitem}
\newlist{abbrv}{itemize}{1}
\setlist[abbrv,1]{label=,labelwidth=1.2in,align=parleft,itemsep=0.1\baselineskip,leftmargin=!}

\numberwithin{equation}{section}

\makeatletter

\DeclareFontEncoding{LGR}{}{}
\DeclareTextSymbol{\~}{LGR}{126}

\newcommand{\tr}{\text{Tr}}
\newcommand{\id}{\mathbbm{1}}

\newcommand{\vA}{\boldsymbol{A}}
\newcommand{\vB}{\boldsymbol{B}}
\newcommand{\vE}{\boldsymbol{E}}
\newcommand{\vF}{\boldsymbol{F}}
\newcommand{\vJ}{\boldsymbol{J}}

\newcommand{\vJt}{\widetilde{\boldsymbol{J}}}

\newcommand{\vAk}{\boldsymbol{A}_{\kappa}}
\newcommand{\vAkt}{\widetilde{\boldsymbol{A}}_{\kappa}}

\newcommand{\vep}{\boldsymbol{\epsilon}}

\newcommand{\scp}[2]{\left\langle #1 , #2 \right\rangle}
\newcommand{\abs}[1]{\left| #1 \right|}

\newcommand{\bra}[1]{\langle #1 |}

\newcommand{\ket}[1]{| #1 \rangle}

\newcommand{\norm}[1]{\left\| #1 \right\| }

\DeclareMathOperator*{\esssup}{ess\,sup}
\DeclareMathOperator\supp{supp}

\newtheorem{theorem}{Theorem}[section]
\newtheorem{lemma}{Lemma}[section]
\newtheorem{def:lemma}{Definition and Lemma}[section]

\newtheorem{proposition}{Proposition}[section]
\newtheorem{assumption}{Assumption}[section]

\theoremstyle{definition}

\theoremstyle{remark}
\newtheorem{remark}{Remark}[section]

\usepackage[colorinlistoftodos]{todonotes}

\usepackage{ulem}

\allowdisplaybreaks[1]

\begin{document}

\title{
Derivation of the Vlasov--Maxwell system from the Maxwell--Schr\"odinger equations with extended charges
}

\author{Nikolai Leopold\footnote{nikolai.leopold@unibas.ch} \ and Chiara Saffirio\footnote{chiara.saffirio@unibas.ch}
\\
University of Basel\footnote{Department of Mathematics and Computer Science, Spiegelgasse 1, 4051 Basel, Switzerland}
}

\maketitle

\frenchspacing

\begin{abstract}
We consider the Maxwell--Schr\"odinger equations in the Coulomb gauge describing the interaction of extended fermions with their self-generated electromagnetic field. They heuristically emerge as mean-field equations from non-relativistic quantum electrodynamics in a mean-field limit of many fermions. In the semiclassical regime, we establish the convergence of the Maxwell--Schr\"{o}dinger equations for extended charges towards the non-relativistic Vlasov--Maxwell dynamics and provide explicit estimates on the accuracy of the approximation. To this end, we build a well-posedness and regularity theory for the Maxwell--Schr\"odinger equations and for the Vlasov--Maxwell system for extended charges.
\end{abstract}



\section{Introduction}

The 3D non relativistic Vlasov--Maxwell system is an important model in kinetic theory used to describe the time evolution of the phase space distribution function $f: \mathbb{R}_{+} \times \mathbb{R}^3 \times \mathbb{R}^3 \rightarrow \mathbb{R}$ of a plasma in interaction with its self-generated electromagnetic field $(\vA, \vB): \mathbb{R}_{+} \times \mathbb{R}^3 \rightarrow \mathbb{R}^3 \times \mathbb{R}^3$. The Cauchy problem associated with the Vlasov--Maxwell system is given by
\begin{align}
\label{eq: Vlasov-Maxwell literature form}
\begin{cases}
\partial_t f + v \cdot \nabla_x f - \frac{e}{m} \left( \vE + \frac{v}{c} \times \vB \right) \cdot \nabla_v f = 0 ,
\\
\partial_t \vE  - c \nabla \times \vB = 4 \pi e \int v f dv ,
\\
\partial_t \vB + c \nabla \times \vE = 0
\end{cases}
\end{align}
with initial conditions $(f_0, \vE_0, \vB_0)$ satisfying
\begin{align}
\label{eq: Vlasov-Maxwell literature form initial conditions}
\nabla \cdot \vE_0 = - 4 \pi \int f_0\,dv 
\quad \text{and} \quad
\nabla \cdot \vB_0 = 0 .
\end{align}
Here, $e$ and $m$ are the charge and mass of an electron and $c \geq 1$ is the velocity of light.
The local-in-time well-posedness of \eqref{eq: Vlasov-Maxwell literature form} was proven in \cite{A1986,D1986, W1982,W1984,W1987}. The global existence of weak solutions with large data was shown in \cite{DL1989}.
Furthermore, \cite{AU1986, D1986, BHK2022} show that for an infinite light velocity, i.e.~$c\uparrow\infty$, the solutions of \eqref{eq: Vlasov-Maxwell literature form} converge to solutions of the Vlasov--Poisson systems. 

\noindent In this paper we study a regularized version of \eqref{eq: Vlasov-Maxwell literature form} and prove its emergence in the semiclassical limit from a regularized (fermionic) Maxwell--Schr\"odinger system in the Coulomb gauge, as given in \eqref{eq:Maxwell-Schroedinger equations}. The regularization we adopt can be viewed as a finite-size requirement on the particles, aligning with the assumptions typically employed for the underlying many-body model of non-relativistic quantum electrodynamics, specifically the Pauli--Fierz Hamiltonian (see \cite{S2004}, Section \ref{subsection:heuristic discussion} and Remark~\ref{remark:properties of the cutoff function} below). Our result provides convergence in strong topology and explicit control on the classical limit in terms of the semiclassical parameter $\varepsilon$, that plays the role of the Planck constant $\hbar$. 

\subsection{Heuristic discussion}
\label{subsection:heuristic discussion}
In order to better understand the connection between the Vlasov-Maxwell and the the Maxwell-Schr\"odinger equations and to see how both systems effectively emerge from non-relativistic quantum electrodynamics let us consider the evolution of $N$ electrons (with non-relativistic dispersion and without spin) in interaction with the quantized electromagnetic field in the Coulomb gauge. The functional setting for this system is the Hilbert space 
$L^2_{\rm{as}}(\mathbb{R}^{3N}) \otimes \bigoplus_{n \geq 0} \mathfrak{h}^{\otimes_s^n}$
where ``as'' indicates antisymmetry under exchange of variables, i.e. taking the fermionic nature of the particles into account, and $\mathfrak{h} = L^2(\mathbb{R}^3) \otimes \mathbb{C}^2$ consists of photon states $f(k,\lambda)$ with wave number $k \in \mathbb{R}^3$ and helicity $\lambda = 1,2$. Elements of the Hilbert space evolve in accordance to the Schr\"odinger equation
\begin{align}
\label{eq:Schroedinger equation Pauli-Fierz}
i \hbar \partial_t \Psi_{N,t} = H_N^{\rm PF} \Psi_{N,t}
\end{align}
with 
\begin{align}
H_N^{\rm PF} &= \sum_{j=1}^N \frac{1}{2m} \left( - i \hbar \nabla_j - c^{-1} e \kappa * \widehat{\vA}(x_j) \right)^2
+ \frac{e^2}{8\pi} \sum_{\substack{i,j=1\\ j\neq i}}^N  \kappa * \kappa * | \cdot |^{-1} (x_j - x_k)
+ H_f 
\end{align} 
being the spinless Pauli--Fierz Hamiltonian.
Here,  $\omega(k) = c \abs{k}$ is the dispersion of the photons, $H_f = \sum_{\lambda =1,2} \int_{\mathbb{R}^3} \hbar\,\omega(k) a_{k,\lambda}^* a_{k,\lambda}\,dk$ denotes the energy of the electromagnetic field and 
\begin{align}
\widehat{\vA}(x)
&=  \sum_{\lambda=1,2}  \int_{\mathbb{R}^3} 
  c\,\sqrt{\hbar/2 \omega(k)}\,\vep_{\lambda}(k) 
(2 \pi)^{-3/2} \left(  e^{ikx} a_{k,\lambda}  + e^{- ik x} a^*_{k,\lambda} \right)\,dk
\end{align}
is the quantized transverse vector potential. There are two real polarization vectors $\big\{ \vep_{\lambda}(k) \big\}_{\lambda = 1,2}$ which implement the Coulomb gauge $\nabla\cdot \widehat{\vA}=0$ by satisfying
\begin{align}
\vep_{\lambda}(k) \cdot \vep_{\lambda'}(k) = \delta_{\lambda, \lambda'}
\quad \text{and} \quad
k \cdot \vep_{\lambda}(k) = 0 .
\end{align}
The pointwise annihilation and creation operators $a_{k,\lambda}$ and $a^*_{k,\lambda}$ satisfy the canonical commutation relations
\begin{align}
\left[ a_{k,\lambda} , a^*_{k',\lambda'} \right] = \delta_{\lambda, \lambda'} \delta(k - k') ,
\quad 
\left[ a_{k,\lambda} , a_{k',\lambda'} \right] 
= \left[ a^*_{k,\lambda} , a^*_{k',\lambda'} \right] 
= 0,
\end{align}
where $[A,B]:=AB-BA$ is the standard commutator of the operators $A$ and $B$.\\
The real function $\kappa$ describes the density of the electrons, which are consequently not regarded as point particles but as finite  size particles with charge distribution $e \kappa(x)$ and mass distribution $m \kappa(x)$. The Pauli--Fierz Hamiltonian is obtained by canonical quantization of the Abraham model and known as a mathematical  model of non-relativistic quantum electrodynamics. If the Fourier transform of the charge distribution\footnote{Throughout this article we use the notation $\mathcal{F}[f](k) = (2 \pi)^{-3/2} \int_{\mathbb{R}^3}  e^{-ikx} f(x)\,dx$ to denote the Fourier transform of a function $f$.} satisfies $\big( \abs{\cdot}^{-1} + \abs{\cdot}^{1/2} \big) \mathcal{F}[\kappa] \in L^2(\mathbb{R}^3)$ (an assumption we will also make later for the effective models, see Remark~\ref{remark:properties of the cutoff function}) it holds that the Hamiltonian is self-adjoint on the domain $\mathcal{D} \big( H_{N}^{\rm{PF}} \big) = \mathcal{D} \big( - \sum_{j=1}^N  \Delta_{x_j} + H_f \big)$ \cite{H2002,M2017,S2004}.
For a detailed introduction to the model we refer to \cite{S2004}. Let us consider equation \eqref{eq:Schroedinger equation Pauli-Fierz} in the semiclassical mean-field regime by setting $e = N^{-1/2}$ and $\hbar = N^{-1/3}$, and choose $c=1$ and $m=2$ for notational convenience. We are interested in many-body wave functions of the form $\Psi_N = \bigwedge_{j=1}^N \varphi_j \otimes W(N^{2/3} \alpha) \Omega$ because they describe electron-photon configurations with little correlations. Here, $\bigwedge_{j=1}^N \varphi_j$ is a Slater determinant of $N$ orthonormal one-particle wave functions $( \varphi_j )_{j=1}^N$ and $W(N^{2/3} \alpha) \Omega$ is a coherent state of photons with mean photon number $N^{4/3} \norm{\alpha}_{\mathfrak{h}}^2$, defined by $W(N^{2/3} \alpha) = \exp \Big( N^{2/3} \sum_{\lambda=1,2} \int_{\mathbb{R}^3} \alpha(k,\lambda) a_{k,\lambda}^* - \overline{\alpha(k,\lambda)} a_{k,\lambda}dk \Big)$ and the vacuum $\Omega$ of the bosonic Fock space over $\mathfrak{h}$.  
It is expected that the structure of such states is preserved during the evolution \eqref{eq:Schroedinger equation Pauli-Fierz} and that the time evolved state at time $t$ is in the limit $N \rightarrow \infty$ approximated by
\begin{align}
\label{eq:propagation of chaos introduction Pauli-Fierz}
\Psi_{N,t} \approx \bigwedge_{j=1}^N \varphi_{j,t} \otimes W(N^{2/3} \alpha_t) \Omega.
\end{align} 
Note that this should be understood as a propagation of chaos assumption and that Slater determinants are completely characterized by their one-electron reduced density matrix $\omega_{N,t} = \sum_{j=1}^N \ket{\varphi_{j,t}} \bra{\varphi_{j,t}}$. If one uses that the action of a quantum field on a coherent state is approximately given by a classical field and that the interaction between electrons in a Slater determinant can be approximated by its mean-field potential plus an exchange term one sees that the electron wave function and the mode function of the electromagnetic field in \eqref{eq:propagation of chaos introduction Pauli-Fierz} should be a solution of the regularized (fermionic) Maxwell--Schr\"odinger system in the Coulomb gauge
\begin{align}
\label{eq:Maxwell-Schroedinger equations}
\begin{cases}
i \varepsilon \partial_t \omega_{N,t} &= \left[ \left( - i \varepsilon \nabla - \kappa * \vA_{\alpha_t} \right)^2  + K * \rho_{\omega_t} - X_{\omega_t} ,  \omega_{N,t} \right]  ,
\\
i \partial_t \alpha_t(k,\lambda) &= \abs{k} \alpha_t(k,\lambda) - \sqrt{\frac{4 \pi^3}{\abs{k}}} \mathcal{F}[\kappa](k) \vep_{\lambda}(k) \cdot \mathcal{F}[\vJ_{\omega_t, \alpha_t}](k)
\end{cases}
\end{align}
with
\begin{subequations}
\begin{align}
\label{eq:Coulomb potential with cutoff}
K &= \frac{1}{8\pi}\, \kappa * \kappa * | \cdot |^{-1} ,
\\
\label{eq:definition vector potential}
\vA_{\alpha}(x) &= (2 \pi )^{- 3/2} \sum_{\lambda = 1,2} \int \frac{1}{\sqrt{2 \abs{k}}} \vep_{\lambda}(k) \left( e^{i k x} \alpha(k,\lambda) + e^{- i k x} \overline{\alpha(k, \lambda)} \right)dk, 
\\
\label{eq:definition density}
\rho_{\omega}(x) &= N^{-1} \omega(x;x) ,
\\
\label{eq:definition charge current}
\vJ_{\omega, \alpha}(x) &= - N^{-1} \left\{ i \varepsilon \nabla, \omega \right\}(x;x) - 2 \rho_{\omega}(x) \kappa * \vA_{\alpha}(x)  ,
\\
\label{eq:definition exchange term}
X_{\omega}(x;y) &= N^{-1} K(x-y) \omega_{N}(x;y) ,
\end{align}
\end{subequations}
semiclassical parameter $\varepsilon = N^{-1/3}$
and initial condition $(\omega_{N,t}, \alpha_t) \big|_{t=0} = (\sum_{j=1}^N \ket{\varphi_{j}} \bra{\varphi_{j}}, \alpha)$. In \eqref{eq:definition charge current} $\{A,B\}:=AB+BA$ stands for the anticommutator of the operators $A$ and $B$, and $\{A,B\}(x;y)$ denotes the kernel of the anticommutator.\\
Within this work we will consider a larger class of initial data, corresponding to mixed states of fermionic systems at positive temperature.

In the mathematical literature the Maxwell--Schr\"odinger system usually appears for $\varepsilon=1$, one single particle ($N=1$) and without regularization, i.e. $\kappa(x) = \delta(x)$. The first equation of \eqref{eq:Maxwell-Schroedinger equations} is often expressed as a Schr\"odinger equation for the one-particle wave function and the second equation of \eqref{eq:Maxwell-Schroedinger equations} is usually written as Maxwell's equations in the Coulomb gauge
\begin{equation}
\nabla \cdot \vA_{\alpha_t} = 0,\quad\quad
\left(\partial_t^2 - \Delta\right)\vA_{\alpha_t} =  - \left( 1 - \nabla \text{div} \Delta^{-1} \right) \kappa * \vJ_{\omega_t, \alpha_t} .
\end{equation}
For the global well-posedness of \eqref{eq:Maxwell-Schroedinger equations} with $N = 1$, $\varepsilon = 1$, $\kappa(x) = \delta(x)$ we refer to \cite{B2009, NM2007} and references therein. Throughout this article we will rely on Proposition \ref{proposition:well-posedness of Maxwell-Schroedinger} for the regularized system under consideration. 

Notice that the rigorous derivation of \eqref{eq:Maxwell-Schroedinger equations} is an open problem which will be addressed in a separate work. With this regard we would like to mention that a rigorous derivation of the Maxwell--Schr\"odinger system from the Pauli--Fierz Hamiltonian in a mean-field limit of many charged bosons was obtained in \cite{LP2020} (see also \cite{FL2022}) and that the fermionic mean-field limit from above was studied in case of the regularized Nelson model \cite{LP2019}, leading to an effective equation in which the electrons linearly couple to a classical Klein--Gordon field. Moreover, let us point out the works \cite{K2009,AFH2022} in which the regularized Newton--Maxwell equations are derived from the Pauli--Fierz Hamiltonian in a classical limit. In \cite{BFP2022} and \cite{CFO2020} the existence of a ground state of the Maxwell--Schr\"odinger energy functional is proven and it's derivation from the Pauli--Fierz Hamiltonian in the quasi-classical regime is obtained.  The scenario where magnetic forces are disregarded, leading to the absence of Maxwell equations in the system, was addressed in \cite{BBPPT2016, BJPSS2016, BPS2014, CLS2021, PP2015, PRSS2017}, where the Hartree-Fock equation was rigorously obtained in the mean-field regime from the many-body Schr\"{o}dinger equation.  

Note that \eqref{eq:Maxwell-Schroedinger equations} depends on $N$ because of the semiclassical parameter $\varepsilon = N^{-1/3}$. In order to see the emergence of the Vlasov--Maxwell equations we introduce the Wigner transform of the one-particle reduced density $\omega_{N,t}$ defined as
\begin{align}
\label{eq:definition Wigner transform}
W_{N,t}(x,v) = \left( \frac{\varepsilon}{2 \pi} \right)^3 \int_{\mathbb{R}^3}  \omega_{N,t} \left( x + \frac{\varepsilon y}{2} ; x - \frac{\varepsilon y}{2} \right) e^{- i v \cdot y}dy 
\end{align}
and its inverse, called the Weyl quantization,
\begin{align}
\label{eq:definition Weyl quantization}
\omega_{N,t}(x;y) = N \int_{\mathbb{R}^3}  W_{N,t} \left( \frac{x+y}{2} , v \right) e^{i v \cdot \frac{x-y}{\varepsilon}}dv.
\end{align}
Using \eqref{eq:Maxwell-Schroedinger equations} and a similar reasoning as in \cite[p.7]{BPSS2016} we obtain 
\begin{align}
\partial_t W_{N,t}(x,v)
&= \left( \frac{\varepsilon}{2 \pi} \right)^3 \int_{\mathbb{R}^3}     \partial_t \omega_{N,t} \left( x + \frac{\varepsilon y}{2} ; x - \frac{\varepsilon y}{2} \right) e^{- i v \cdot y} dy
\nonumber \\
&= \bigg[   \Big(  \nabla K * \rho_{\omega_t}(x) + 2 \sum_{j=1}^3 (\nabla \kappa * \vA_{\alpha_t}^{j})(x) \big( \kappa * \vA_{\alpha_t}^{j}(x) - v_j   \Big)
\cdot \nabla_v
\nonumber \\
&\qquad + 2  \Big( \kappa * \vA_{\alpha_t}(x) - v \Big) \cdot \nabla_x \bigg] W_{N,t}(x,v) + \mathcal{O}(\varepsilon).
\end{align}
Together with $N^{-1} \omega_t(x;x) = \int_{\mathbb{R}^3}  W_{N,t}(x,v)\,dv$ and $\left\{ i \varepsilon \nabla, \omega_{N,t} \right\}(x;x) = - 2 N \int_{\mathbb{R}^3} \, v W_{N,t}(x,v)\,dv$ this suggests that the Wigner transform of $\omega_{N,t}$ satisfies in the limit $N \rightarrow \infty$, i.e. $\varepsilon \rightarrow 0$, the following transport equation\footnote{In this case $f$ has to be initially the Wigner transform of $\omega_{N,t} \big|_{t=0}$.}
\begin{align}
\label{eq: Vlasov-Maxwell different style of writing}
\begin{cases}
\partial_t f_t  &= - 2 \left( v - \kappa * \vA_{\alpha_t} \right) \cdot \nabla_x f_t + \vF_{f_t, \alpha_t} \cdot \nabla_v f_t , 
\\
i \partial_t \alpha_t(k,\lambda) &= \abs{k} \alpha_t(k,\lambda) - \sqrt{\frac{4 \pi^3}{\abs{k}}} \mathcal{F}[\kappa](k) \vep_{\lambda}(k) \cdot \mathcal{F}[\vJt_{f_t,\alpha_t}](k)
\end{cases}
\end{align}
with $A_{\alpha_t}$ being defined as in \eqref{eq:definition vector potential} and
\begin{subequations}
\begin{align}
\label{eq:density of Vlasov equation}
\widetilde{\rho}_{f}(x) &= \int f(x,v) \, dv ,
\\
\label{eq:F-field of Vlasov equation}
\vF_{f, \alpha}(x,v) &= \nabla_x \Big[ K  * \widetilde{\rho}_{f}(x) + \left( \kappa * \vA_{\alpha}(x) \right)^2  - 2 v \cdot \kappa * \vA_{\alpha}(x) \Big] ,
\\
\label{eq:current density of Vlasov equation}
\vJt_{f,\alpha}(x) &= 2 \int \left( v - \kappa * \vA_{\alpha} \right) f(x,v) \, dv 
\end{align}
\end{subequations}
and initial datum $(f, \alpha) \in L^1 \left( \mathbb{R}^3 \times \mathbb{R}^3 \right) \times \left( L^2(\mathbb{R}^3) \otimes \mathbb{C}^2 \right)$.
Note that \eqref{eq: Vlasov-Maxwell different style of writing} with $\kappa(x) = -\delta(x)/4$  is formally equivalent to \eqref{eq: Vlasov-Maxwell literature form} in the Coulomb gauge with $c= e= 1$ and $m=2$, as shown in Appendix \ref{section:comparison of the different Vlasov-Maxwell equations}. For this reason hereafter we refer to \eqref{eq: Vlasov-Maxwell different style of writing} as the (regularized) Vlasov--Maxwell system.

\subsection{Notation}

For $\sigma \in \mathbb{N}_0$, $k \geq 0$ and $1 \leq p \leq \infty$ let $W_{k}^{\sigma,p} \left( \mathbb{R}^{d} \right)$ be the Sobolev space equipped with the norm 
\begin{align}
\label{eq:definition LP-spaces}
\norm{f}_{W_{k}^{\sigma, p} \left( \mathbb{R}^{d} \right)} 
&=
\begin{cases}
\left( \sum_{\abs{\alpha} \leq \sigma} \norm{\left< \cdot \right>^{k} D^{\alpha} f}_{L^p\left( \mathbb{R}^{d} \right)}^p \right)^{\frac{1}{p}}
\quad 1 \leq &p < \infty ,
\\
\max_{\abs{\alpha} \leq \sigma}  \norm{\left< \cdot \right>^{k} D^{\alpha} f}_{L^{\infty}\left( \mathbb{R}^{d} \right)}
\quad &p = \infty,
\end{cases}
\end{align}
where $\left< x \right>^2 = 1 + \abs{x}^2$.
In the cases $\sigma=0$ or $p=2$ we use the shorthand notation
$L_k^{p} \left( \mathbb{R}^{d} \right) \coloneqq W_{k}^{(0,p)} \left( \mathbb{R}^{d} \right)$
and  $H_k^{\sigma} \left( \mathbb{R}^{d} \right) \coloneqq W_{k}^{\sigma, 2} \left( \mathbb{R}^{d} \right)$.  Vectors in $\mathbb{R}^{6}$ are written as
 $z = (x,v)$ so that $\left< z \right>^2 = 1 + \abs{x}^2 + \abs{v}^2$ and $D_z^{\alpha} = \left( \partial/ \partial x_1 \right)^{\beta_1} \left( \partial/ \partial v_1 \right)^{\gamma_1} \cdots \left( \partial/ \partial x_3 \right)^{\beta_3} \left( \partial/ \partial v_3 \right)^{\gamma_3}$ with
$\beta = (\beta_i)_{i \in \llbracket 1, 3 \rrbracket} \in \mathbb{N}_0^3$ and $\gamma = (\gamma_i)_{i \in \llbracket 1, 3 \rrbracket} \in \mathbb{N}_0^3$ 
such that $\abs{\alpha} = \sum_{i=1}^3 \beta_i + \gamma_i$. 
For two Banach spaces $A$ and $B$ we denote by $A \cap B$ the Banach space of vectors $f \in A \cap B$ with norm $\norm{f}_{A \cap B} = \norm{f}_A + \norm{f}_B$. For $a \in \mathbb{R}$ we define the weighted $L^2$-space
 $\dot{\mathfrak{h}}_{a} = L^2(\mathbb{R}^3, \abs{k}^{2a} dk) \otimes \mathbb{C}^2$ with norm 
\begin{align}
\norm{\alpha}_{\dot{\mathfrak{h}}_{a}} &= \bigg( \sum_{\lambda=1,2} \int_{\mathbb{R}^3} 
\abs{k}^{2a} \abs{\alpha(k,\lambda)}^2dk \bigg)^{1/2} = \norm{ |\cdot |^{a} \alpha}_{\mathfrak{h}}  .
\end{align}
Note that the notation $\mathfrak{h}$ will be used to denote the Hilbert space $\mathfrak{h}_0$ with scalar product
\begin{align}
\scp{\alpha}{\beta} = \sum_{\lambda=1,2} \int_{\mathbb{R}^3} 
\overline{\alpha(k,\lambda)} \beta(k,\lambda) dk.
\end{align}
It is, moreover, convenient to define for $a \geq 0$ the Banach spaces
\begin{align}
\mathfrak{h}_a &= \left\{ \alpha \in \mathfrak{h} : (1 + \abs{k}^2)^{a/2} \alpha \in \mathfrak{h}  \right\}
\end{align}
with 
\begin{align}
\norm{\alpha}_{\mathfrak{h}_a} &= \bigg( \sum_{\lambda=1,2} \int_{\mathbb{R}^3} 
\left( 1 + \abs{k}^2 \right)^a \abs{\alpha(k,\lambda)}^2dk \bigg)^{1/2}  .
\end{align}
For a reflexive Banach space $(X, \norm{\cdot}_{X})$ and $p \in [1,\infty]$ we denote by $L^p(0,T; X)$ the space of (equivalence classes of) strongly Lebesgue-measurable functions $\alpha: [0,T] \rightarrow X$ with the property that
\begin{align}
\norm{\alpha}_{L_T^p X} 
&= 
\begin{cases}
\left( \int_0^T  \norm{\alpha_t}_X^p dt \right)^{1/p}
\quad \text{if} \; 1 \leq p < \infty \\
\esssup_{t \in [0,T]} \norm{\alpha_t}_{X} \quad \text{if} \; p= \infty
\end{cases}
\end{align}
is finite. Note that $(L^p(0,T; X), \norm{\cdot}_{L_T^p X})$ is a Banach space. The Sobolev space $W^{1,p}(0,T; X)$ consists of all functions $\alpha \in L^p(0,T; X)$ such that $\partial_t \alpha_t$ exists in weak sense and belongs to $L^p(0,T; X)$. Furthermore,
\begin{align}
\norm{\alpha}_{W^{1,p}(0,T;X)}
&= 
\begin{cases}
\left( \int_0^T \norm{\alpha_t}_{X}^p + \norm{\partial_t \alpha_t}_{X}^p dt \right)^{1/p}
\quad \text{if} \; 1 \leq p < \infty \\
\esssup_{t \in [0,T]}
\left( \norm{\alpha_s}_{X} + \norm{\partial_t \alpha_t}_{X} \right)
\quad \text{if} \; p= \infty .
\end{cases}
\end{align}

Let $\textfrak{S}^{\infty} \left( L^2(\mathbb{R}^3) \right)$ be the set of all bounded operators on $L^2 (\mathbb{R}^3)$ and $\textfrak{S}^1 \left( L^2(\mathbb{R}^3) \right)$ be the set of trace class operators on $L^2(\mathbb{R}^3)$. More generally, for $p\in[1,\infty)$, we denote by $\textfrak{S}^p(\mathbb{R}^3)$ the $p$-Schatten space equipped with the norm $\norm{A}_{\textfrak{S}^p}=(\tr \abs{A}^p)^\frac{1}{p}$, where $A$ is an operator, $A^*$ its adjoint and $\abs{A}=\sqrt{A^*A}$.
For $a \geq 0$ let
\begin{align}
\textfrak{S}^{a,1}(L^2(\mathbb{R}^3))
&= \left\{  \omega: \omega \in \textfrak{S}^{\infty}(L^2(\mathbb{R}^3)), \omega^* = \omega 
\; \text{and} \; \left( 1 - \Delta \right)^{a/2} \omega\left( 1 - \Delta \right)^{a/2} \in \textfrak{S}^{1}(L^2(\mathbb{R}^3))
\right\}
\end{align}
with 
\begin{align}
\norm{\omega}_{\textfrak{S}^{a,1}}
&= \norm{\left( 1 - \Delta \right)^{a/2} \omega \left( 1 - \Delta \right)^{a/2}}_{\textfrak{S}^1}.
\end{align}
The positive cone of the latter one is defined as
\begin{align}
\textfrak{S}_{+}^{a,1}(L^2(\mathbb{R}^3)) 
&= \left\{ \omega \in \textfrak{S}^{a,1}(L^2(\mathbb{R}^3)) : \omega \geq 0
\right\} .
\end{align}
We use the symbol $C$ to represent a general positive constant that could vary from one line to another, and it may depend on the parameters $k$ and $\sigma$ in \eqref{eq:definition LP-spaces}, as well as on the cutoff function outlined in \eqref{eq:assumption on the cutoff parameter 3} below. Positive constants depending on $\varepsilon$ are denoted by $C_{\varepsilon}$. Additionally, we adopt the notation $C(f,g)$ for positive constants depending on the enclosed quantities (specifically $f$ and $g$ in this context). 

\subsection{Organisation of the paper}

The paper is organised as follows. In Section~\ref{section:main-results} we present our main result concerning the semiclassical approximation of the Maxwell--Schr\"{o}dinger dynamics with a solution to the Vlasov--Maxwell equation (see Theorem~\ref{theorem:main theorem}). This result builds upon the well-posedness and regularity theory for both Maxwell--Schr\"{o}dinger equations and for the Vlasov--Maxwell system, which are given in Proposition~\ref{proposition:well-posedness of Maxwell-Schroedinger} and Proposition~\ref{lemma:global solutions VM}, respectively.
We also discuss the physical relevance of the assumptions made on the cut-off parameter (see Remark~\ref{remark:properties of the cutoff function} and Remark~\ref{remark:examples of cut-off}) and compare our results with the existing literature. Section~\ref{section:preliminaries} collects preliminary estimates that will be used throughout the text. Sections~\ref{section:proof of main theorem}, \ref{section:Well-posedness of the Maxwell--Schroedinger equations} and \ref{section:Properties of the solutions of the Vlasov-Maxwell equations} are devoted to the proofs of the Theorem~\ref{theorem:main theorem}, Proposition~\ref{proposition:well-posedness of Maxwell-Schroedinger}  and Proposition~\ref{lemma:global solutions VM}, respectively. 
The paper includes two appendices: in Appendix~\ref{section:comparison of the different Vlasov-Maxwell equations}, we establish the equivalence between two distinct formulations of the Vlasov--Maxwell system, while Appendix~\ref{section:auxiliary estimates} shows that certain Schatten norms of commutators between the Weyl quantization of a particle distribution and position or momentum operators can be related to weighted Sobolev norms of the particle distribution.

\section{Main results}\label{section:main-results}

Throughout the paper we assume the cut-off parameter to satisfy the following assumptions.

\begin{assumption}
\label{assumption:cutoff function}
Let $\kappa: \mathbb{R}^3 \rightarrow \mathbb{R}$ be a real and even charge distribution such that 
\begin{align}
\label{eq:assumption on the cutoff parameter 3}
\kappa \in L^1(\mathbb{R}^3) 
\quad \text{and} \quad
\left( - \Delta \right)^{1/2} \kappa \in L^2(\mathbb{R}^3).
\end{align}
\end{assumption}

\begin{remark}
\label{remark:properties of the cutoff function}
We will not study the explicit dependence of the estimates on the charge distribution and for this reason estimate the respective norms \eqref{eq:assumption on the cutoff parameter 3} by a generic not specified constant.
The first condition in \eqref{eq:assumption on the cutoff parameter 3} requires that the negative and positive parts of the charge distribution are summable. If the charge distribution is purely negative or positive it only assumes that the total charge is finite. The second condition requires a sufficient decay of the Fourier modes with high momenta.
Note that our assumptions imply $\mathcal{F}[\kappa] \in L^{\infty}(\mathbb{R}^3)$ and 
\begin{align}
\label{eq:assumption on the cutoff parameter}
\left( \abs{\cdot}^{-1} + \abs{\cdot}^{1/2} \right) \mathcal{F}[\kappa] \in L^2(\mathbb{R}^3) ,
\end{align}
where the latter one is usually required to prove the self-adjointness of the Pauli--Fierz Hamiltonian (see Subsection \ref{subsection:heuristic discussion}).
Since we require $\kappa$ to be even and summable, our assumptions are slightly more restrictive than the assumptions required for the self-adjointness of the Pauli--Fierz Hamiltonian because we exclude non-even charge distributions and those with finite total charge whose positive and negative part are not summable. 
Note that 
\begin{align}
\label{eq:assumption on the cutoff parameter 2}
\mathcal{F}[\kappa](k) = \mathcal{F}[\kappa](-k)
\quad \text{and} \quad
\mathcal{F}[\kappa](k) \in \mathbb{R}
\quad  \forall k \in \mathbb{R}^3
\end{align}
because $\kappa$ is real and even.
\end{remark}

\begin{remark}\label{remark:examples of cut-off}
Typical examples of cut-off parameters which satisfy our assumptions are 
\begin{align}
 e \frac{\mathcal{F}[\id_{\abs{\cdot} \leq \Lambda}]}{(2 \pi)^{3/2}}
\quad \text{and} \quad 
 \frac{e}{\sigma^3 (2 \pi)^{3/2}} e^{- \frac{x^2}{2 \sigma^2}}
\quad \text{with} \quad
\sigma > 0 ,
\end{align}
describing extended particles with total charge $e \in \mathbb{R}$ 
distributed by means of a sharp ultraviolet momentum cutoff 
and in a Gaussian fashion, respectively.
\end{remark}


In this paper we rely on the following well-posedness result for the Maxwell--Schr\"odinger system, which is proved in Section \ref{section:Well-posedness of the Maxwell--Schroedinger equations}.

\begin{proposition}
\label{proposition:well-posedness of Maxwell-Schroedinger}
Let $\kappa$ satisfy Assumption \ref{assumption:cutoff function}.
For all $(\omega_0, \alpha_0) \in \textfrak{S}_{+}^{2,1} (L^2(\mathbb{R}^3)) \times \mathfrak{h}_{1/2} \cap \dot{\mathfrak{h}}_{-1/2} $ the Cauchy problem for the Maxwell--Schr\"odinger system \eqref{eq:Maxwell-Schroedinger equations} associated with $(\omega_0, \alpha_0)$   has a unique  $
 C \big( \mathbb{R}_{+} ; \textfrak{S}_{+}^{2,1} ( L^2(\mathbb{R}^3) ) \big) \cap C^1 \big( \mathbb{R}_{+} ; \textfrak{S}^{1} ( L^2(\mathbb{R}^3) )  \big)   \times  C \big( \mathbb{R}_{+} ; \mathfrak{h}_{1/2} \cap \dot{\mathfrak{h}}_{-1/2} \big) \cap C^1 \big( \mathbb{R}_{+} ; \dot{\mathfrak{h}}_{-1/2} \big) 
$ solution. 
The mass and energy 
\begin{align}
\label{eq:Maxwell-Schroedinger with exchange term energy definition}
\mathcal{E}^{\rm{MS}}[\omega, \alpha] &= 
\tr \left( \omega \left( - i \varepsilon \nabla - \kappa * \vA_{\alpha} \right)^2 \right) + \frac{1}{2} \tr \left( \left( \, K * \rho_{\omega} - X_{\omega} \right) \omega \right)
 + N \norm{\alpha}_{\dot{\mathfrak{h}}_{1/2}}^2
\end{align} 
of the system are conserved, i.e. 
\begin{align}
\tr  \left( \omega_t \right) = \tr  \left( \omega_0 \right)
\quad \text{and} \quad 
\mathcal{E}^{\rm{MS}}[\omega_t, \alpha_t] = \mathcal{E}^{\rm{MS}}[\omega_0, \alpha_0]
\quad \text{for all} \; t \in  \mathbb{R}_{+} .
\end{align}
\end{proposition}

We will also make use of the following result concerning the solution and regularity theory for the Vlasov--Maxwell system. Its proof is provided in Section~\ref{section:Properties of the solutions of the Vlasov-Maxwell equations}.

\begin{proposition}
\label{lemma:global solutions VM}
Let $R>0$ and $a, b \in \mathbb{N}$ satisfying $a \geq 5$ and $b \geq 3$. For all $(f_0, \alpha_0) \in H_a^{b}(\mathbb{R}^6) \times \mathfrak{h}_b \cap \dot{\mathfrak{h}}_{-1/2}$ such that $\supp f_0 \subset A_R$ with $A_R = \{ (x,v) \in \mathbb{R}^6 , \abs{v} \leq R \}$ system
\eqref{eq: Vlasov-Maxwell different style of writing} has a unique
$L^{\infty} \big( \mathbb{R}_{+}; H_{a}^{b}(\mathbb{R}^6) \big) \cap C \big( \mathbb{R}_{+}; H_{a}^{b-1}(\mathbb{R}^6) \big)  \cap C^1 (\mathbb{R}_{+};  H_{a}^{b-2}(\mathbb{R}^6)) \times  C \big( \mathbb{R}_{+} ; \mathfrak{h}_{b} \cap \dot{\mathfrak{h}}_{-1/2} \big) \cap C^1 \big( \mathbb{R}_{+} ; \mathfrak{h}_{b-1} \cap \dot{\mathfrak{h}}_{-1/2} \big)$ solution.
The $L^p$--norms of the particle distribution (with $p \geq 1$) and the energy
\begin{align}
\label{eq:Vlasov-Maxwell energy definition}
\mathcal{E}^{\rm{VM}}[f, \alpha] &=  
\int_{\mathbb{R}^6} dx \, dv \, f(x,v) \left( v - \kappa * \vA_{\alpha}(x) \right)^2 
+  \frac{1}{2} \int_{\mathbb{R}^6} dx \, dv \,
f(x,v) K * \widetilde{\rho}_{f}(x) 
+  \norm{\alpha}_{\dot{\mathfrak{h}}_{1/2}}^2 
\end{align}
are conserved, i.e. $\norm{f_t}_{L^p(\mathbb{R}^6)} = \norm{f_0}_{L^p(\mathbb{R}^6)}$ and $\mathcal{E}^{\rm{VM}}[f_t, \alpha_t] =  \mathcal{E}^{\rm{VM}}[f_0, \alpha_0]$ for all $t \in \mathbb{R}_{+}$.
\end{proposition}

Our main result is the following

\begin{theorem}
\label{theorem:main theorem}
Let $\kappa$ satisfy Assumption \ref{assumption:cutoff function}, $\varepsilon = N^{-1/3}$, $\alpha_0 \in \mathfrak{h}_{1/2} \cap \dot{\mathfrak{h}}_{-1/2}$ and $\omega_{N,0} \in \textfrak{S}^{2,1} \big( L^2(\mathbb{R}^3) \big)$ be a sequence of reduced density matrices on $L^2(\mathbb{R}^3)$ satisfying $\tr \left( \omega_{N,0} \right) = N$ and $0 \leq \omega_{N,0} \leq 1$. Let $(\omega_{N,t}, \alpha_t)$ be the unique solution of \eqref{eq:Maxwell-Schroedinger equations} with initial datum $(\omega_{N,0},\alpha_0)$. Let $\big( \widetilde{W}_{N,t}, \widetilde{\alpha}_t \big)$ be the unique solution to \eqref{eq: Vlasov-Maxwell different style of writing} with initial datum $\big( \widetilde{W}_{N,0}, \widetilde{\alpha}_0 \big)$, such that $\widetilde{W}_{N,0} \geq 0$, $\widetilde{W}_{N,0} \in W_2^{0,1}(\mathbb{R}^6)$ and verifying
\begin{align}
\label{eq:main theorem condition on the Vlasov-Maxwell solution}
\big( \widetilde{W}_{N}, \widetilde{\alpha} \, \big) \in L_{\rm{loc}}^{\infty} \big( \mathbb{R}_{+}, H_7^8(\mathbb{R}^6) \big) \times L_{\rm{loc}}^{\infty} \big( \mathfrak{h}_5 \cap \dot{\mathfrak{h}}_{- 1/2} \big) . 
\end{align}
Moreover, let $\widetilde{\omega}_{N,t}$ be the Weyl quantization of $\widetilde{W}_{N,t}$ and
\begin{align}
\label{eq:main theorem definition of the norm under consideration}
\Xi (t)
&= N^{-1} \norm{\sqrt{1 - \varepsilon^2 \Delta} \left( \omega_{N,t} - \widetilde{\omega}_{N,t} \right) \sqrt{1 - \varepsilon^2 \Delta}}_{\textfrak{S}^{1}}
+ \norm{\alpha_t - \widetilde{\alpha}_t}_{\dot{\mathfrak{h}}_{-1/2} \, \cap \, \dot{\mathfrak{h}}_{1/2}} .
\end{align}
Then, 
\begin{align}
\Xi (t) &\leq \left( \Xi(0) + \varepsilon \widetilde{C}(t) \right)
e^{C(t)} .
\end{align}
Here,
\begin{align}
\widetilde{C}(t) &\leq 
\int_0^t ds \,
\sum_{j=0}^{6} \varepsilon^j \norm{\widetilde{W}_{N,s}}_{H_{7}^{j+2}}
\bigg( 1 +   \sum_{k=0}^3 \varepsilon^{2 + k} \norm{\widetilde{W}_{N,s}}_{H_{4}^{k+1}} 
+ \sum_{k=0}^4 \varepsilon^k  \norm{\widetilde{\alpha}_s}_{\mathfrak{h}_{k+1}}^2
\bigg) ,
\\
C(t) &= C \left( 1 + N^{-1} \mathcal{E}^{\rm{MS}}[\omega_{N,0}, \alpha_0] + \mathcal{E}^{\rm{VM}}[\widetilde{W}_{N,0}, \widetilde{\alpha}_0] + 
\norm{\widetilde{W}_{N,0}}_{L^1(\mathbb{R}^6)} \right)^2
\nonumber \\
&\quad \times
\bigg( \left< t \right> + \int_0^t ds \, \sum_{j=0}^{7} \varepsilon^j \norm{\widetilde{W}_{N,s}}_{H_{7}^{j+1}} \bigg) 
\end{align}
and $C$ is a numerical constant which depends on the specific choice of the cutoff function $\kappa$.
\end{theorem}

\begin{remark}
Notice that \eqref{eq:main theorem condition on the Vlasov-Maxwell solution} is satisfied if $\big( \widetilde{W}_{N,0}, \widetilde{\alpha}_0 \big) \in  H_7^{8}(\mathbb{R}^6) \times \mathfrak{h}_8 \cap \dot{\mathfrak{h}}_{-1/2}$ such that $\widetilde{W}_{N,0} \geq 0$, $\widetilde{W}_{N,0} \in W_2^{0,1}(\mathbb{R}^6)$ and $\supp \widetilde{W}_{N,0} \subset A_R$ with $A_R = \{ (x,v) \in \mathbb{R}^6 , \abs{v} \leq R \}$ for some $R>0$ because of Proposition \ref{lemma:global solutions VM}. If in addition $N^{-1} \mathcal{E}^{\rm{MS}}[\omega_{N,0}, \alpha_0] \leq C$ holds, $\widetilde{C}(t)$ and $C(t)$ can be bounded uniformly in $N$ by a constant depending on time but not on the number of particles.
Furthermore, the solution of \eqref{eq:Maxwell-Schroedinger equations} exists and is unique thanks to Proposition \ref{proposition:well-posedness of Maxwell-Schroedinger}.
\end{remark}

\subsection{Comparison with the literature}


To our knowledge the literature on the subject is rather limited. In \cite{MMY2023} M\"{o}ller, Mauser and Yang considered the classical limit of the Pauli--Poisswell system, in which the coupling with the Maxwell equations is replaced by a simplified equation. For monokinetic initial data and by means of WKB methods, they show that as $\hbar\to 0$ the system is approximated by the Euler-Poisswell equations. See also the related works \cite{MS2007,MM2023,M2023}. In \cite{P2023}  the linear problem with external magnetic field has been considered. More precisely, the magnetic Liouville equation is obtained as classical limit of the Heisenberg equation with non constant magnetic field, whose vector potential is regular and given. \\
To the best of our knowledge, our result is the first the derivation of the Vlasov--Maxwell system from the quantum dynamics described by the Maxwell--Schr\"{o}dinger equations for extended charges, dealing with a self-consistent electromagnetic field which satisfies Maxwell's equations.
Our contribution can be viewed as a step in the derivation of the Vlasov--Maxwell system from non-relativistic quantum electrodynamics. The Vlasov--Maxwell system is indeed expected to be a good mean-field and semiclassical approximation of the spinless Pauli--Fierz Hamiltonian. As outlined in Section~\ref{subsection:heuristic discussion}, the Maxwell--Schr\"{o}dinger equations can be heuristically obtained as mean-field limit of the spinless Pauli--Fierz Hamiltonian (and will be addressed rigorously in a forthcoming work), whereas the semiclassical approximation is the focus of the present article.\\ 
In this spirit, notice that a regularized variant of the relativistic Vlasov--Maxwell system as the mean-field limit of a system of many classical particles was derived in \cite{G2011}. In the context of classical mechanics, we further refer to the works \cite{L2016,CLPY2019}, where the relativistic Vlasov--Maxwell system is obtained in a combined mean-field and point-particle limit of a system of $N$ rigid charges with $N$-dependent radius.\\
Concerning the propagation of moments for the Vlasov--Maxwell system, we mention the works \cite{R2021,R2022}, dealing with a magnetized Vlasov--Poisson equation where the magnetic field is external and uniformly bounded. \\
We would also like to point out the works \cite{S1981,LP1993,APPP2011,AKN2013,GMP2016,BPSS2016,S2019,L2019,Lewin-Sabin_2020,LS2020,CLL2021,CLS2022} in which the related problem of the semiclassical approximation of the Hartree equation with the Vlasov--Poisson system is addressed. 

\subsection{Strategy of the proof}

In order to prove Theorem~\ref{theorem:main theorem} we adopt the approach of \cite{BPSS2016} and compare the solutions of the Maxwell--Schr\"odinger equations~\eqref{eq:Maxwell-Schroedinger equations} with the Weyl quantization \eqref{eq:definition Weyl quantization} of the solutions of the Vlasov--Maxwell system \eqref{eq: Vlasov-Maxwell different style of writing} by means of a Gr\"{o}nwall estimate. The main novelty in comparison to \cite{BPSS2016} is to deal with the additional difficulties arising from the coupling of the electrons to their self-induced electromagnetic field.\\ On the one hand a control on the difference of the vector potentials is required. This is achieved most efficiently by the mode functions $\alpha_t$ and $\widetilde{\alpha}_t$ of the electromagnetic fields, and gives the rationale for writing Maxwell's equations in \eqref{eq:Maxwell-Schroedinger equations} and \eqref{eq: Vlasov-Maxwell different style of writing} in terms of  $\alpha_t$ and $\widetilde{\alpha}_t$. \\
On the other hand the fact that the vector potentials of the respective magnetic fields couple to the charge currents of the electrons forces us to measure the distance of the electron states in a Sobolev trace norm in which the Laplacian is weighted with the semiclassical parameter (see \eqref{eq:main theorem definition of the norm under consideration}). \\
Beside the semiclassical analysis, we deal with the well-posedness of the Maxwell--Schr\"odinger equations and the Vlasov--Maxwell equations and prove that there exist initial data whose evolution leads to the integrability and regularity conditions required by Theorem~\ref{theorem:main theorem}. This is shown in Proposition \ref{proposition:well-posedness of Maxwell-Schroedinger} and Proposition \ref{lemma:global solutions VM} by means of fixed point arguments -- leading to local solutions -- and suitable propagation estimates. Crucial technical ingredients of our approach to obtain the result for the Maxwell--Schr\"odinger equations are the use of propagation estimates for the time evolution of the magnetic Laplacian from \cite{NM2005} and the specific choice of the Banach space \eqref{eq:local solutions MS equations definition Banach space 1}. In the case of the Vlasov--Maxwell system the compact support in the velocity of the initial data is the key assumption to overcome the difficulties arising from the interaction between the  electromagnetic field and the non-relativistic electrons with non-finite speed of propagation.  For point particles this assumption has already been used to prove well-posedness in different functional frameworks, for instance in \cite{D1986}.

\section{Preliminaries}\label{section:preliminaries}
Within this section we collect important estimates that will often appear in the proofs of our results.

\begin{lemma}
\label{lemma:estimates for the vector potential and interaction}
Let the Assumption \ref{assumption:cutoff function} hold. Then
\begin{subequations}
\begin{align}
\label{eq:vector potential L-infty estimate}
\norm{\kappa * \vA_{\alpha}}_{L^{\infty}(\mathbb{R}^3)} 
&\leq C \min \left\{  \norm{\alpha}_{\dot{\mathfrak{h}}_{-1}}, \norm{\alpha}_{\dot{\mathfrak{h}}_{-1/2}},  \norm{\alpha}_{\mathfrak{h}} ,  \norm{\alpha}_{\dot{\mathfrak{h}}_{1/2}}
\right\} ,
\\
\label{eq:vector potential L-2 estimate}
\norm{\kappa * \vA_{\alpha}}_{L^{2}(\mathbb{R}^3)}  
&\leq C \norm{\alpha}_{\dot{\mathfrak{h}}_{-1/2}} .
\end{align}
\end{subequations}
For $\sigma \in \mathbb{N}$ we, moreover, have
\begin{subequations}
\begin{align}
\label{eq:estimates for the vector potential and interaction 0} 
\norm{\kappa * \vA_{\alpha}}_{W_0^{1, \infty}(\mathbb{R}^3)} 
&\leq C \norm{\alpha}_{\dot{\mathfrak{h}}_{1/2}} ,
\\
\label{eq:estimates for the vector potential and interaction 1} 
\norm{\kappa * \vA_{\alpha}}_{W_0^{\sigma, \infty}(\mathbb{R}^3)} 
&\leq C \norm{\alpha}_{\mathfrak{h}_{\sigma -1}}
\end{align}
\end{subequations}
\end{lemma}

\begin{proof}
Note that
\begin{align}
\norm{\kappa * \vA}_{L^{\infty}(\mathbb{R}^3)}
&\leq 2 \sum_{\lambda = 1,2} \int  \abs{k}^{-1/2}  \abs{\mathcal{F}[\kappa](k)} \abs{\alpha(k,\lambda)}dk
\nonumber \\
&\leq C \min \Big\{
\norm{\abs{\cdot}^{1/2} \mathcal{F}[\kappa]}_{L^2(\mathbb{R}^3)}
\norm{\abs{\cdot}^{-1} \alpha}_{\mathfrak{h}}  ,
\norm{ \mathcal{F}[\kappa]}_{L^2(\mathbb{R}^3)}
\norm{\abs{\cdot}^{-1/2} \alpha}_{\mathfrak{h}} ,
\nonumber \\
&\qquad \qquad \qquad 
\norm{\abs{\cdot}^{-1/2} \mathcal{F}[\kappa]}_{L^2(\mathbb{R}^3)}
\norm{ \alpha}_{\mathfrak{h}} ,
\norm{\abs{\cdot}^{-1} \mathcal{F}[\kappa]}_{L^2(\mathbb{R}^3)}
\norm{ \abs{\cdot}^{1/2} \alpha}_{\mathfrak{h}}
\Big\} .
\end{align}
Together with \eqref{eq:assumption on the cutoff parameter} this proves the first inequality.
Similarly, 
\begin{align}
\norm{\kappa * \vA}_{W_0^{1,\infty}(\mathbb{R}^3)}
&\leq \norm{\left( 1 - \Delta \right)^{1/2} \kappa * \vA}_{L^{\infty}(\mathbb{R}^3)}
\nonumber \\
&\leq 2 \sum_{\lambda = 1,2} \int \abs{k}^{-1/2} \left< k \right> \abs{\mathcal{F}[\kappa](k)} \abs{\alpha(k,\lambda)}dk
\nonumber \\
&\leq C \norm{\abs{\cdot}^{-1} \left< \cdot \right> \mathcal{F}[\kappa]}_{L^2(\mathbb{R}^3)}
\norm{\abs{\cdot}^{1/2} \alpha}_{\mathfrak{h}} 
\end{align}
implies  \eqref{eq:estimates for the vector potential and interaction 0} and from
\begin{align}
\norm{\kappa * \vA}_{W_0^{\sigma,\infty}(\mathbb{R}^3)}
&\leq \norm{\left( 1 - \Delta \right)^{\sigma/2} \kappa * \vA}_{L^{\infty}(\mathbb{R}^3)}
\nonumber \\
&\leq 2 \sum_{\lambda = 1,2} \int  \abs{k}^{-1/2} \left< k \right>^{\sigma} \abs{\mathcal{F}[\kappa](k)} \abs{\alpha(k,\lambda)}dk
\nonumber \\
&\leq C \norm{\abs{\cdot}^{-1/2} \left< \cdot \right> \mathcal{F}[\kappa]}_{L^2(\mathbb{R}^3)}
\norm{ \alpha}_{\mathfrak{h}_{\sigma - 1}} 
\end{align}
we get \eqref{eq:estimates for the vector potential and interaction 1}. By means of
\begin{align}
\vA_{\alpha}(x) = (2 \pi )^{- 3/2} \sum_{\lambda = 1,2} \int \frac{1}{\sqrt{2 \abs{k}}} \vep_{\lambda}(k) \left( e^{i k x} \alpha(k,\lambda) + e^{- i k x} \overline{\alpha(k, \lambda)} \right)\,dk\,,
\end{align}
Young's inequality and \eqref{eq:assumption on the cutoff parameter 3}  we have
\begin{align}
\norm{\kappa * \vA_{\alpha}}_{L^{2}(\mathbb{R}^3)} 
&\leq \norm{\kappa}_{L^1(\mathbb{R}^3)} \norm{\vA_{\alpha}}_{L^2(\mathbb{R}^3)} 
\leq C \norm{\alpha}_{\dot{\mathfrak{h}}_{-1/2}} .
\end{align}
\end{proof}

Moreover, the following estimates will be useful
\begin{lemma}
\label{lemma:general estimates for the interaction}
Let $s \in \{ 0,1,2 \}$ and $\sigma \in \mathbb{N}$. Then, 
\begin{subequations}
\begin{align}
\label{eq:W-0-2-infty estimate for the direct interaction potential}
\norm{K * \rho_{\omega}}_{W_0^{3, \infty}(\mathbb{R}^3)}
&\leq C N^{-1} \norm{\omega}_{\textfrak{S}^1} 
\quad \text{for} \quad  \omega \in \textfrak{S}^1(L^2(\mathbb{R}^3)) ,
\\
\label{eq:estimates for the vector potential and interaction 2} 
\norm{K * \widetilde{\rho}_{f}}_{W_0^{\sigma ,\infty}(\mathbb{R}^3)}
&\leq  C
\begin{cases}
\min \left\{ \norm{f}_{W_4^{0,2}(\mathbb{R}^6)} , \norm{f}_{L^1(\mathbb{R}^6)} \right\}
&\text{if} \quad  \sigma \leq 3 , \\ 
\norm{f}_{W_4^{\sigma - 3,2}(\mathbb{R}^6)} &\text{if}   \quad \abs{\sigma} > 3 
\end{cases}  
\end{align}
\end{subequations}
and
\begin{subequations}
\begin{align}
\label{eq:estimate exchange term in operator norm}
\norm{X_{\omega} }_{\textfrak{S}^{\infty}(L^2(\mathbb{R}^3))} 
&\leq C N^{-1} \norm{\omega}_{\textfrak{S}^1}
\quad \text{for} \quad  \omega \in \textfrak{S}^1(L^2(\mathbb{R}^3)) ,
\\
\label{eq:estimate commutator with exchange term in L-s,1 norm}
\norm{\left[ X_{\omega} , \Gamma \right]}_{\textfrak{S}^{s,1}} 
&\leq  C N^{-1} \norm{\omega}_{\textfrak{S}^{s,1}} \norm{\Gamma}_{\textfrak{S}^{s,1}}
\quad \text{for} \quad   \omega, \Gamma \in \textfrak{S}^{s,1}(L^2(\mathbb{R}^3)).
\end{align}
\end{subequations}
\end{lemma}

\begin{proof}
Due to \eqref{eq:Coulomb potential with cutoff} we have
\begin{align}
\label{eq:Fourier transform of the potential}
\mathcal{F}[K](k) 
&=  \mathcal{F}[\kappa * \kappa * \abs{\cdot}^{-1}](k)
= C  \mathcal{F}[\kappa](k)^2   \abs{k}^{-2} .
\end{align}
Together with Young's inequality and \eqref{eq:assumption on the cutoff parameter} this implies
\begin{align}
\label{eq:preliminary estimate to bound the direct interaction term}
\norm{K * \rho_{\omega}}_{W_0^{\sigma, \infty}(\mathbb{R}^3)}
&\leq  \norm{K}_{W_0^{3, \infty}(\mathbb{R}^3)}
\norm{\rho_{\omega}}_{W_0^{\sigma-3, 1}(\mathbb{R}^3)}
\nonumber \\
&\leq  \norm{\left< \cdot \right>^3 \mathcal{F}[K]}_{L^1(\mathbb{R}^3)}
\norm{(1 - \Delta)^{(\sigma - 3)/2} \rho_{\omega}}_{L^1(\mathbb{R}^3)}
\nonumber \\
&\leq  C \norm{\left< \cdot \right>^{3/2} \abs{\cdot}^{-1} \mathcal{F}[\kappa]}_{L^2(\mathbb{R}^3)}^2
\norm{(1 - \Delta)^{(\sigma - 3)/2} \rho_{\omega}}_{L^1(\mathbb{R}^3)} 
\nonumber \\
&\leq  C 
\norm{(1 - \Delta)^{(\sigma - 3)/2} \rho_{\omega}}_{L^1(\mathbb{R}^3)} .
\end{align}
Inequality  \eqref{eq:W-0-2-infty estimate for the direct interaction potential} is then obtained by setting $\sigma = 3$ and
\begin{align}
\label{eq: estimate to bound the density by the L1-norm of omega}
\norm{\rho_{\omega}}_{L^1(\mathbb{R}^3)}
&= \sup_{O \in L^{\infty}(\mathbb{R}^3), \norm{O}_{L^{\infty}(\mathbb{R}^3)} \leq 1} \abs{\int dx \, O(x) \rho_{\omega}(x)} \leq N^{-1} \norm{\omega}_{\textfrak{S}^1},
\end{align}
which holds because the space of bounded operators is the dual space of trace-class operators and every function $x \rightarrow O(x)$ defines a multiplication operator. Next, we replace $\rho_{\omega}$  in \eqref{eq:preliminary estimate to bound the direct interaction term} by $\widetilde{\rho}_f$. The case $\sigma \leq 3$ in \eqref{eq:W-0-2-infty estimate for the direct interaction potential} is an immediate consequence of \eqref{eq:density of Vlasov equation}. In order to obtain the estimate for $\sigma > 3$ we bound the last line of \eqref{eq:preliminary estimate to bound the direct interaction term} with  $\rho_{\omega} = \widetilde{\rho}_f$ by
\begin{align}
\norm{(1 - \Delta)^{(\sigma - 3)/2} \widetilde{\rho}_{f}}_{L^{1}(\mathbb{R}^3)} 
&\leq  \int_{\mathbb{R}^6} dz \,  \left< z \right>^{-4} \left< z \right>^4 \abs{(1 - \Delta)^{(\sigma - 3)/2} f(z)}
\nonumber \\
&\leq C \norm{\left< \cdot \right>^4 (1 - \Delta)^{(\sigma - 3)/2} f}_{L^2(\mathbb{R}^6)} 
\nonumber \\
&\leq \norm{f}_{W_4^{\sigma-3,2}(\mathbb{R}^6)} .
\end{align}

Now, let $\left\{ \lambda_j, \varphi_j \right\}_{j \in \mathbb{N}}$ be the spectral decomposition of $\omega$. If we use Young's inequality, the Cauchy Schwarz inequality, $\norm{\mathcal{F}[K]}_{L^1(\mathbb{R}^3)} \leq C \norm{\abs{\cdot}^{-1} \mathcal{F}[K]}_{L^2(\mathbb{R}^3)}^2 \leq C$ (which follows from \eqref{eq:Fourier transform of the potential} and \eqref{eq:assumption on the cutoff parameter}) and \eqref{eq:assumption on the cutoff parameter} we get
\begin{align}
\norm{X_{\omega} }_{\textfrak{S}^{\infty}(L^2(\mathbb{R}^3))}
&= \sup_{\psi \in L^2(\mathbb{R}^3), \norm{\psi}_{L^2(\mathbb{R}^3)} = 1} \norm{X_{\omega} \psi}_{L^2(\mathbb{R}^3)}
\nonumber \\
&\leq \sup_{\psi \in L^2(\mathbb{R}^3), \norm{\psi}_{L^2(\mathbb{R}^3)} = 1} N^{-1} \sum_{j \in \mathbb{N}} \abs{\lambda_j}
\norm{\varphi_j K * \{\overline{\varphi_j} \psi\}}_{L^2(\mathbb{R}^3)}
\nonumber \\
&\leq \sup_{\psi \in L^2(\mathbb{R}^3), \norm{\psi}_{L^2(\mathbb{R}^3)} = 1} N^{-1} \sum_{j \in \mathbb{N}} \abs{\lambda_j} \norm{\varphi}_{L^2(\mathbb{R}^3)}
\norm{K}_{L^{\infty}(\mathbb{R}^3)}
\norm{{\overline{\varphi_j} \psi}}_{L^1(\mathbb{R}^3)}
\nonumber \\
&\leq \sup_{\psi \in L^2(\mathbb{R}^3), \norm{\psi}_{L^2(\mathbb{R}^3)} = 1} N^{-1} \sum_{j \in \mathbb{N}} \abs{\lambda_j} \norm{\varphi}_{L^2(\mathbb{R}^3)}^2 \norm{\psi}_{L^2(\mathbb{R}^3)}
\norm{\mathcal{F}[K]}_{L^1(\mathbb{R}^3)}
\nonumber \\
&\leq C N^{-1} \sum_{j \in \mathbb{N}} \abs{\lambda_j}
\nonumber \\
&\leq C N^{-1} \norm{\omega}_{\textfrak{S}^1} 
\end{align}
and
\begin{align}
\label{eq:estimate commutator with exchange term in L-1 norm}
\norm{\left[ X_{\omega} , \Gamma \right]}_{\textfrak{S}^1}
&\leq 2 \norm{X_{\omega} }_{\textfrak{S}^{\infty}(L^2(\mathbb{R}^3))} \norm{\Gamma}_{\textfrak{S}^1}
\leq  C N^{-1} \norm{\omega}_{\textfrak{S}^1} \norm{\Gamma}_{\textfrak{S}^1} .
\end{align}
In the following let $s \in \{1,2 \}$.  If $\omega \in \textfrak{S}^{s,1}(L^2(\mathbb{R}^3))$ we can split the compact and self-adjoint operator $(1 - \Delta)^{s/2} \omega (1 - \Delta)^{s/2} = \left(  (1 - \Delta)^{s/2} \omega (1 - \Delta)^{s/2} \right)_{+} - \left(  (1 - \Delta)^{s/2} \omega (1 - \Delta)^{s/2} \right)_{-}$ into its positive and negative part. Then 
\begin{align}
\left(  (1 - \Delta)^{s/2} \omega (1 - \Delta)^{s/2} \right)_{+} \left(  (1 - \Delta)^{s/2} \omega (1 - \Delta)^{s/2} \right)_{-} &= 0
\end{align}
and 
\begin{align}
\abs{(1 - \Delta)^{s/2} \omega (1 - \Delta)^{s/2}}
&= \left(  (1 - \Delta)^{s/2} \omega (1 - \Delta)^{s/2} \right)_{+} + \left(  (1 - \Delta)^{s/2} \omega (1 - \Delta)^{s/2} \right)_{-}
\end{align}
We define the positive operators 
\begin{align}
\label{eq:splitting of an operator to obtain Sobolev trace estimates}
\omega_{+} &= (1 - \Delta)^{- s/2} \left(  (1 - \Delta)^{s/2} \omega (1 - \Delta)^{s/2} \right)_{+} (1 - \Delta)^{-s/2} ,
\nonumber \\
\omega_{-} &= (1 - \Delta)^{-s/2} \left(  (1 - \Delta)^{s/2} \omega (1 - \Delta)^{s/2} \right)_{-} (1 - \Delta)^{-s/2}
\end{align}
which satisfy
\begin{align}
\omega_{+} - \omega_{-} = \omega
\quad \text{and} \quad
\norm{\omega}_{\textfrak{S}^{s,1}}
= \norm{(1 - \Delta)^{s/2} \omega (1 - \Delta)^{s/2} }_{\textfrak{S}^{1}}
= \norm{\omega_{+} }_{\textfrak{S}^{s,1}} 
+ \norm{\omega_{-} }_{\textfrak{S}^{s,1}} .
\end{align}
Note that $\omega_{+/-} \in \textfrak{S}^{s,1}(L^2(\mathbb{R}^3))$ because $\omega \in \textfrak{S}^{s,1}(L^2(\mathbb{R}^3))$ and
\begin{align}
\norm{\omega_{+/-}}_{\textfrak{S}^{s,1}} 
&= \norm{\left(  (1 - \Delta)^{s/2} \omega (1 - \Delta)^{s/2} \right)_{+/-}}_{\textfrak{S}^1}
\leq \norm{(1 - \Delta)^{s/2} \omega (1 - \Delta)^{s/2}}_{\textfrak{S}^1} 
= \norm{\omega}_{\textfrak{S}^{s,1}} .
\end{align}
This splitting and the linearity of the mapping $\omega \mapsto X_{\omega}$ yield
\begin{align}
\norm{\left[ X_{\omega} , \Gamma \right]}_{\textfrak{S}^{s,1}}
&\leq 2 \norm{\left( 1 - \Delta \right)^{s/2} X_{\omega_+} (1 - \Delta)^{-s/2}}_{\textfrak{S}^{\infty}(L^2(\mathbb{R}^3))} 
\norm{\Gamma}_{\textfrak{S}^{s,1}}
\nonumber \\
&\quad + 2 \norm{\left( 1 - \Delta \right)^{s/2} X_{\omega_-} (1 - \Delta)^{-s/2}}_{\textfrak{S}^{\infty}(L^2(\mathbb{R}^3))} 
\norm{\Gamma}_{\textfrak{S}^{s,1}} .
\end{align}
Since $\omega_{+} \in \textfrak{S}^{s,1}(L^2(\mathbb{R}^3))$ there exists a spectral set $\{ \lambda_j, \varphi_j \}_{j \in \mathbb{N}}$ with $\lambda_j > 0$ 
for all $j \in \mathbb{N}$ such that $\omega_{+} = \sum_{j \in \mathbb{N}} \lambda_j \ket{\varphi_j} \bra{\varphi_j}$. Because $(1 - \Delta)^{s/2} \varphi_j = \lambda_j^{-1} (1 - \Delta)^{s/2} \omega_{+} \varphi_j$ we get
\begin{align}
\norm{\varphi_j}_{H^s(\mathbb{R}^3)} 
&\leq  \lambda_j^{-1} 
\norm{(1 - \Delta)^{s/2}  \omega_{+} (1 - \Delta)^{s/2}}_{\textfrak{S}^{\infty}(L^2(\mathbb{R}^3))}
\norm{(1 - \Delta)^{-s/2} \varphi_j}_{L^2(\mathbb{R}^3)}
\nonumber \\
&\leq \lambda_j^{-1}  \norm{(1 - \Delta)^{s/2}  \omega_{+} (1 - \Delta)^{s/2}}_{\textfrak{S}^1(L^2(\mathbb{R}^3))} \norm{\varphi_j}_{L^2(\mathbb{R}^3)}.
\end{align}
We consequently have $\varphi_j \in H^s(\mathbb{R}^3)$,
\begin{align}
\left( 1 - \Delta \right)^{s/2} \omega_{+} \left( 1 - \Delta \right)^{s/2} = \sum_{j \in \mathbb{N}} \lambda_j \ket{\left( 1 - \Delta \right)^{s/2} \varphi_j} \bra{\left( 1 - \Delta \right)^{s/2} \varphi_j}
\end{align}
and
\begin{align}
\norm{\omega_{+}}_{\textfrak{S}^{s,1}}
&= \tr \left(  \left( 1 - \Delta \right)^{s/2} \omega_{+} \left( 1 - \Delta \right)^{s/2} \right)
= \sum_{j \in \mathbb{N}} \lambda_j \norm{\varphi_j}_{H^s(\mathbb{R}^3)}^2  .
\end{align}
Using this spectral decomposition, the product rule of differentiation, Young's inequality and the Cauchy--Schwarz inequality, we obtain ($\chi = (1 - \Delta)^{- s/2} \psi$)
\begin{align}
\norm{\left( 1 - \Delta \right)^{s/2} X_{\omega_{+}} \left( 1 - \Delta \right)^{-s/2} \psi}_{L^2(\mathbb{R}^3)}
&\quad = N^{-1} \Big\| \sum_{j \in \mathbb{N}} \lambda_j \varphi_j K * \{ \overline{\varphi_j} \chi \} \Big\|_{H^s(\mathbb{R}^3)}
\nonumber \\
&\leq C N^{-1} \sum_{j \in \mathbb{N}} \lambda_j \norm{\varphi_j}_{H^s(\mathbb{R}^3)} \norm{\overline{\varphi_j} \chi}_{L^1(\mathbb{R}^3)}
\norm{K  }_{W_0^{s,\infty}(\mathbb{R}^3)} 
\nonumber \\
&\leq C N^{-1} \norm{\chi}_{L^2(\mathbb{R}^3)} 
\norm{\left< \cdot \right>^s \mathcal{F}[K]  }_{L^1(\mathbb{R}^3)}
\sum_{j \in \mathbb{N}} \lambda_j \norm{\varphi_j}_{H^s(\mathbb{R}^3)}^2 
\nonumber \\
&\leq C N^{-1} \norm{\psi}_{L^2(\mathbb{R}^3)} 
\norm{\left< \cdot \right>^s \mathcal{F}[K]  }_{L^1(\mathbb{R}^3)} 
\norm{\omega_{+}}_{\textfrak{S}^{s,1}} .
\end{align}
Similarly as above, \eqref{eq:Fourier transform of the potential} and \eqref{eq:assumption on the cutoff parameter} yield $\norm{\left< \cdot \right>^s \mathcal{F}[K]  }_{L^1(\mathbb{R}^3)} \leq  \norm{\left< \cdot \right>^2 \mathcal{F}[K]  }_{L^1(\mathbb{R}^3)} \leq C$. This shows
\begin{align}
&\norm{\left( 1 - \Delta \right)^{s/2} X_{\omega_+} (1 - \Delta)^{- s/2}}_{\textfrak{S}^{\infty}(L^2(\mathbb{R}^3))} 
\nonumber \\
&\quad = \sup_{\psi \in L^2(\mathbb{R}^3), \norm{\psi}_{L^2(\mathbb{R}^3)} = 1} 
\norm{\left( 1 - \Delta \right)^{s/2} X_{\omega_{+}} \left( 1 - \Delta \right)^{- s/2} \psi}_{L^2(\mathbb{R}^3)}
\nonumber \\
&\quad \leq C N^{-1} \norm{\omega_{+}}_{\textfrak{S}^{s,1}} .
\end{align}
Analogously we obtain the same estimate for $X_{\omega_{-}}$, i.e.
\begin{align}
\norm{\left( 1 - \Delta \right)^{s/2} X_{\omega_{-}} (1 - \Delta)^{- s/2}}_{\textfrak{S}^{\infty}(L^2(\mathbb{R}^3))} 
&\leq C N^{-1} \norm{\omega_{-}}_{\textfrak{S}^{s,1}} .
\end{align} 
In total, this shows $\norm{\left[ X_{\omega}, \Gamma \right]}_{\textfrak{S}^{s,1}} \leq C N^{-1}  \norm{\omega}_{\textfrak{S}^{s,1}}
\norm{\Gamma}_{\textfrak{S}^{s,1}} $ for $s \in \{ 1,2 \}$.

\end{proof}

\begin{lemma}
\label{lemma:Charge current L-infty estimate of the difference}
Let $\vJ$ be defined as in \eqref{eq:definition charge current}
and $\norm{\omega}_{\textfrak{S}_{\varepsilon}^{1,1}} = \norm{(1 - \varepsilon^2 \Delta)^{1/2} \omega  (1 - \varepsilon^2 \Delta)^{1/2}}_{\textfrak{S}^1}$. Then,
\begin{align}
\label{eq:Charge current L-infty estimate of the difference}
&\norm{\mathcal{F}[\vJ_{\omega,\alpha}] - \mathcal{F}[\vJ_{\omega',\alpha'}]}_{L^{\infty}(\mathbb{R}^3)}
\nonumber \\
&\quad\leq C  \left( 1 + \min \left\{ \norm{\alpha}_{\dot{\mathfrak{h}}_{-1/2}} , \norm{\alpha}_{\mathfrak{h}}, \norm{\alpha}_{\dot{\mathfrak{h}}_{1/2}}  \right\} \right) 
N^{-1} \norm{\omega - \omega'}_{\textfrak{S}_{\varepsilon}^{1,1}}
\nonumber \\
&\qquad + C N^{-1} \norm{\omega'}_{\textfrak{S}^1}
\min \left\{ \norm{\alpha - \alpha'}_{\dot{\mathfrak{h}}_{-1/2}} , \norm{\alpha - \alpha'}_{\mathfrak{h}}, \norm{\alpha - \alpha'}_{\dot{\mathfrak{h}}_{1/2}}  \right\} .
\end{align}
For $\mathcal{E}^{\rm{MS}}$ being defined as in \eqref{eq:Maxwell-Schroedinger with exchange term energy definition} and $(\omega,\alpha) \in \textfrak{S}_{+}^{1,1} \times \dot{\mathfrak{h}}_{1/2}$ we have
\begin{subequations}
\begin{align}
\label{eq:estimate for the MS energy functional 1}
\norm{\alpha}_{\dot{\mathfrak{h}}_{1/2}}^2
&\leq N^{-1}  \mathcal{E}^{\rm{MS}}[\omega, \alpha] 
+ C N^{-2} \norm{\omega}_{\textfrak{S}^1}^2 ,
\\
\label{eq:estimate for the MS energy functional 2}
\tr \left( - \varepsilon^2 \Delta \omega \right)
&\leq  C \left( \mathcal{E}^{\rm{MS}}[\omega, \alpha] + N^{-1} \norm{\omega}_{\textfrak{S}^1}^2 \right)
\left( 1 +  N^{-1} \norm{\omega}_{\textfrak{S}^1} \right) ,
\\
\label{eq:estimate for the MS energy functional 3}
\mathcal{E}^{\rm{MS}}[\omega, \alpha] 
&\leq \tr \left( - \varepsilon^2 \Delta \omega \right)
+ C N^{-1} \norm{\omega}_{\textfrak{S}^1}^2
+ \norm{\alpha}_{\dot{\mathfrak{h}}_{1/2}}^2 \left( C \norm{\omega}_{\textfrak{S}^1} + N \right) .
\end{align}
\end{subequations}
\end{lemma}

\begin{proof}
We write the difference of the currents as 
\begin{align}
\vJ_{\omega, \alpha}(x) - \vJ_{\omega', \alpha'}(x)
&= - N^{-1} \left\{ i \varepsilon \nabla, \left(  \omega - \omega' \right) \right\}(x;x)  - 2 \rho_{\omega - \omega'}(x)   \kappa * \vA_{\alpha}(x)
\nonumber \\
&\quad 
- 2 \rho_{\omega'}(x) \kappa * \vA_{\alpha - \alpha'}(x)  .
\end{align}
The Fourier transform of $N$ times the first expression on the right hand side can be estimated by
\begin{align}
\norm{\mathcal{F} [\left\{ i \varepsilon \nabla, \left(  \omega_{N,s} - \widetilde{\omega}_{N,s} \right) \right\}(\cdot; \cdot)] }_{L^{\infty}(\mathbb{R}^3)}
&\leq \sup_{k \in \mathbb{R}^3}  \abs{\int dx \, e^{- i k x} \left\{ i \varepsilon \nabla, \left(  \omega_{N,s} - \widetilde{\omega}_{N,s} \right) \right\}(x;x) }
\nonumber \\
&= \sup_{k \in \mathbb{R}^3} \abs{\tr \left(
e^{- i k \hat{x}} \left\{ i \varepsilon \nabla, \left(  \omega_{N,s} - \widetilde{\omega}_{N,s} \right) \right\}
\right) }
\nonumber \\
&\leq \norm{i \varepsilon \nabla \left(  \omega_{N,s} - \widetilde{\omega}_{N,s} \right)}_{\textfrak{S}^1}
+
\norm{\left(  \omega_{N,s} - \widetilde{\omega}_{N,s} \right) i \varepsilon \nabla }_{\textfrak{S}^1} 
\nonumber \\
&\leq  2 \norm{\left( \omega_{N,s} - \widetilde{\omega}_{N,s} \right)}_{\textfrak{S}_{\varepsilon}^{1,1}} .
\end{align}
To obtain the ultimate expression we have used that $\norm{i \varepsilon \nabla \left( 1 - \varepsilon \Delta \right)^{-1/2}}_{\textfrak{S}^{\infty}} \leq 1$. Using \eqref{eq:vector potential L-infty estimate} and \eqref{eq: estimate to bound the density by the L1-norm of omega} we bound the second and third term by
\begin{align}
2 \norm{\mathcal{F}[ \rho_{\omega - \omega'}  \, \kappa * \vA_{\alpha} ]}_{L^{\infty}(\mathbb{R}^3)}
&\leq  \norm{\rho_{\omega - \omega'} \,  \kappa * \vA_{\alpha} }_{L^1(\mathbb{R}^3)}
\nonumber \\
&\leq \norm{\kappa * \vA_{\alpha} }_{L^{\infty}(\mathbb{R}^3)} \norm{\rho_{\omega - \omega'}}_{L^1(\mathbb{R}^3)}  
\nonumber \\
&\leq C N^{-1} \min \left\{ \norm{\alpha}_{\dot{\mathfrak{h}}_{-1/2}} , \norm{\alpha}_{\mathfrak{h}}, \norm{\alpha}_{\dot{\mathfrak{h}}_{1/2}}  \right\} \norm{\omega - \omega'}_{\textfrak{S}^1} .
\end{align}
and
\begin{align}
2 \norm{\mathcal{F}[ \rho_{\omega'} \, \kappa * \vA_{\alpha - \alpha'}]}_{L^{\infty}(\mathbb{R}^3)}
&\leq C N^{-1}  \min \left\{ \norm{\alpha - \alpha' }_{\dot{\mathfrak{h}}_{-1/2}} , \norm{\alpha - \alpha' }_{\mathfrak{h}}, \norm{\alpha - \alpha'}_{\dot{\mathfrak{h}}_{1/2}}  \right\} \norm{\omega' }_{\textfrak{S}^1}.
\end{align}
Collecting the estimates proves \eqref{eq:Charge current L-infty estimate of the difference}.
Inequality \eqref{eq:estimate for the MS energy functional 1} follows from the positivity of the magnetic Laplacian and Lemma~\ref{lemma:general estimates for the interaction}. Using $(- i \varepsilon \nabla - \kappa * \vA_{\alpha})^2 \geq - \varepsilon^2 \Delta + 2 i \varepsilon \nabla \cdot \kappa * \vA_{\alpha}$, $\norm{i \varepsilon \nabla \omega}_{\textfrak{S}^1}^2 \leq \tr \left( - \varepsilon^2 \Delta \omega \right) \norm{\omega}_{\textfrak{S}^1}$ if $\omega \geq 0$, \eqref{eq:vector potential L-infty estimate} and Lemma~\ref{lemma:general estimates for the interaction}
we get
\begin{align}
\tr \left( - \varepsilon^2 \Delta \omega \right)
&\leq 2 \tr \left( (- i \varepsilon \nabla - \kappa * \vA_{\alpha})^2 \omega \right) + C \norm{\omega}_{\textfrak{S}^1} \norm{\alpha}_{\dot{\mathfrak{h}}_{1/2}}^2 
\nonumber \\
&\leq 2 \mathcal{E}^{\rm{MS}}[\omega, \alpha] + C  \norm{\omega}_{\textfrak{S}^1} \left( \norm{\alpha}_{\dot{\mathfrak{h}}_{1/2}}^2 + N^{-1} \norm{\omega}_{\textfrak{S}^1} \right) .
\end{align}
Together with \eqref{eq:estimate for the MS energy functional 1} this leads to \eqref{eq:estimate for the MS energy functional 2}. By similar estimates we obtain inequality \eqref{eq:estimate for the MS energy functional 3} concluding the proof.
\end{proof}

\begin{lemma}
\label{lemma:energy Vlasov-Maxwell estimate}
Let $(f, \alpha) \in W_{2}^{0,1}(\mathbb{R}^6) \times \dot{\mathfrak{h}}_{1/2}$ such that $f \geq 0$. For $\mathcal{E}^{\rm{VM}}$, $\vF_{f,\alpha}$ and $\vJ_{f, \alpha}$ defined as in \eqref{eq:Vlasov-Maxwell energy definition}, \eqref{eq:F-field of Vlasov equation}  and \eqref{eq:current density of Vlasov equation} we have
\begin{subequations}
\begin{align}
\label{eq:estimate for the VM energy functional 1}
\norm{\alpha}_{\dot{\mathfrak{h}}_{1/2}}^2 
&\leq \mathcal{E}^{\rm{VM}}[f,\alpha] + C \norm{f}_{L^1(\mathbb{R}^6)}^2 ,
\\
\label{eq:estimate for the VM energy functional 2}
\int_{\mathbb{R}^6}  f(x,v) v^2dx\,dv
&\leq \mathcal{E}^{\rm{VM}}[f,\alpha] 
+ C  \norm{f}_{L^1(\mathbb{R}^6)} \left( \norm{f}_{L^1(\mathbb{R}^6)} + 
\norm{\alpha}_{\dot{\mathfrak{h}}_{1/2}}^2 \right) ,
\\
\label{eq:estimate for the VM energy functional 3}
\mathcal{E}^{\rm{VM}}[f,\alpha]  
&\leq 2  \int_{\mathbb{R}^6}  f(x,v) v^2dx\,dv
+ C \norm{f}_{L^1(\mathbb{R}^6)}^2
+ \norm{\alpha}_{\dot{\mathfrak{h}}_{1/2}}^2 \left( C \norm{f}_{L^1(\mathbb{R}^6)} + 1 \right) 
\end{align}
\end{subequations}
and
\begin{subequations}
\begin{align}
\label{eq:estimate F function supremum in x} 
\sup_{x \in \mathbb{R}^3}  \abs{\vF_{f, \alpha}(x,v)}
&\leq C \left< v \right> \left( 1 + \mathcal{E}^{\rm{VM}}[f, \alpha] 
+  C \norm{f}_{L^1(\mathbb{R}^6)}^2  \right)   ,
\\
\label{eq:estimate current density Vlasov Maxwell L-1 norm}
\norm{\widetilde{\vJ}_{f,\alpha}}_{L^1(\mathbb{R}^3)}
&\leq C \left( 1 + \mathcal{E}^{\rm{VM}}[f, \alpha] 
+  C \norm{f}_{L^1(\mathbb{R}^6)}^2  \right)   ,
\\
\label{eq:estimate current density Vlasov Maxwell W-0-sigma-1 norm}
\norm{\widetilde{\vJ}_{f,\alpha}}_{W_0^{\sigma,1}(\mathbb{R}^3)}
&\leq C \left( 1 + \norm{\alpha}_{\mathfrak{h}_{\sigma -1}} \right) \norm{f}_{W_5^{\sigma,2}(\mathbb{R}^6)}
\quad \text{with} \; \sigma \in \mathbb{N}_0 .
\end{align}
\end{subequations}

\end{lemma}

\begin{proof}
Note that
\begin{align}
\abs{\int_{\mathbb{R}^6} f(x,v) \left( \kappa * \vA_{\alpha}(x) \right)^2dx\,dv}
&\leq C \norm{f}_{L^1(\mathbb{R}^6)} \norm{\alpha}_{\dot{\mathfrak{h}}_{1/2}}^2
\end{align}
and
\begin{align}
\abs{\int_{\mathbb{R}^6} f(x,v) K * \widetilde{\rho}_f(x)\,dx\,dv}
&\leq C \norm{f}_{L^1(\mathbb{R}^6)}^2
\end{align}
because of \eqref{eq:vector potential L-infty estimate} and \eqref{eq:estimates for the vector potential and interaction 2}. The inequalities \eqref{eq:estimate for the VM energy functional 1}--\eqref{eq:estimate for the VM energy functional 3} then follow from the positivity of $f$, $(v - \kappa * \vA_{\alpha}(x))^2 \geq v^2 - 2 v \cdot \kappa \vA_{\alpha}(x)$ and the Cauchy--Schwarz inequality. By means of  \eqref{eq:estimates for the vector potential and interaction 0}  and \eqref{eq:estimates for the vector potential and interaction 2} and the Cauchy--Schwarz inequality we obtain
\begin{align}
\sup_{x \in \mathbb{R}^3}  \abs{\vF_{f, \alpha}(x,v)}
&\leq C \left< v \right> \left( \norm{f}_{L^1(\mathbb{R}^6)} + \norm{\alpha}_{\dot{\mathfrak{h}}_{1/2}}^2 \right) 
\end{align}
and
\begin{align}
\norm{\widetilde{\vJ}_{f,\alpha}}_{L^1(\mathbb{R}^3)}
&\leq  C \norm{\alpha}_{\dot{\mathfrak{h}}_{1/2}}  \norm{f}_{L^1(\mathbb{R}^6)}
+ 2 \left( \norm{f}_{L^1(\mathbb{R}^6)} \int_{\mathbb{R}^6}  f(x,v) v^2dx\,dv \right)^{1/2}
\end{align}
Together with \eqref{eq:estimate for the VM energy functional 1} and \eqref{eq:estimate for the VM energy functional 2} this leads to 
\eqref{eq:estimate F function supremum in x} and \eqref{eq:estimate current density Vlasov Maxwell L-1 norm}. Note that
\begin{align}
\norm{\widetilde{\vJ}_{f,\alpha}}_{W_0^{\sigma,1}(\mathbb{R}^3)}
&\leq 2 \left( 1 + \norm{\kappa * \vA_{\alpha}}_{W_0^{\sigma, \infty}(\mathbb{R}^3)} \right) \norm{f}_{W_1^{\sigma,1}} .
\end{align}
By applying \eqref{eq:estimates for the vector potential and interaction 1} and the Cauchy-Schwarz inequality we then obtain \eqref{eq:estimate current density Vlasov Maxwell W-0-sigma-1 norm}.
\end{proof}

\begin{lemma}
Let $A$ be a trace class operator, $V: \mathbb{R}^3 \rightarrow \mathbb{C}^3$ be a regular enough function and $D_{\leq n} = \frac{\left( 1 - \varepsilon^2 \Delta \right)^{1/2}}{\left( 1 - \varepsilon^2 \Delta / n^2 \right)^{1/2}} $ with $n \in \mathbb{N}$
as defined in \eqref{eq:regularized fractional Laplacian}. Then,
\begin{align}
\label{estimate:commutator function with operator in trace norm}
\norm{\left[ V, A \right]}_{\textfrak{S}^1}
&\leq  \norm{\left[ \hat{x}, A \right]}_{\textfrak{S}^1}
\int \abs{k}  \abs{\widehat{V}(k)}\,dk  ,
\\
\label{estimate:commutator fractional laplacian with potential}
\norm{\left[ \left( 1 - \varepsilon^2 \Delta \right)^{1/2} , V \right]}_{\textfrak{S}^{\infty}}
&\leq 2 \varepsilon \int  \abs{k} \abs{\widehat{V}(k)}\,dk  ,
\\
\label{estimate:commutator regularized fractional laplacian with potential}
\norm{\left[ D_{\leq n} , V \right]}_{\textfrak{S}^{\infty}}
&\leq 2 \varepsilon \int \abs{k} \abs{\widehat{V}(k)}\,dk
\end{align}

\end{lemma}

\begin{proof}
Using the identity 
\begin{align}
\label{eq:commutator of e-i-k-x and operator Duhamel expansion}
\left[ e^{i k \hat{x}} , A \right]
&= \int_0^1  \, e^{i \lambda k \hat{x}}
\left[ i k \cdot \hat{x} , A \right] e^{i (1 - \lambda) k \hat{x}}\,d \lambda,
\end{align} 
we estimate
\begin{align}
\norm{\left[ V, A \right]}_{\textfrak{S}^1}
&\leq \int \abs{\widehat{V}(k)} \norm{\left[ e^{ik \hat{x}} , A \right]}_{\textfrak{S}^1}dk
\nonumber \\
&\leq \int\abs{k} \abs{\widehat{V}(k)} \norm{\left[ \hat{x} , A \right]}_{\textfrak{S}^1}dk .
\end{align}
Note that
\begin{align}
\left[ \left( 1 - \varepsilon^2 \Delta \right)^{1/2} , V(x) \right]
&= \int \mathcal{F}[V](k)  \left[ \left( 1 - \varepsilon^2 \Delta \right)^{1/2} , e^{-ikx} \right]\,dk
\nonumber \\
&= \int  \mathcal{F}[V](k)  e^{-ikx}
\left( \left( 1 + \varepsilon^2 (i \nabla + k)^2 \right)^{1/2}
- \left( 1 - \varepsilon^2 \Delta \right)^{1/2} \right)\,dk
\nonumber \\
&= \int  \mathcal{F}[V](k)  e^{-ikx}
\frac{1 + \varepsilon^2 (i \nabla + k)^2 -  ( 1 - \varepsilon^2 \Delta)}{\left( \left( 1 + \varepsilon^2 (i \nabla + k)^2 \right)^{1/2}
+ \left( 1 - \varepsilon^2 \Delta \right)^{1/2} \right)}\,dk
\nonumber \\
&= \int \mathcal{F}[V](k)  e^{-ikx}
\frac{\varepsilon^2 k \cdot i \nabla + \varepsilon^2 k \cdot (i \nabla + k)}{\left( \left( 1 + \varepsilon^2 (i \nabla + k)^2 \right)^{1/2}
+ \left( 1 - \varepsilon^2 \Delta \right)^{1/2} \right)}\,dk
\end{align}
By means of the operator inequalities $\varepsilon (i \nabla + k) \leq \left( 1 + \varepsilon^2 (i \nabla + k)^2 \right)^{1/2}$ and $\varepsilon i \nabla \leq \left( 1 - \varepsilon^2 \Delta \right)^{1/2}$, which hold due to the spectral theorem, we obtain
\begin{align}
\norm{\left[ \left( 1 - \varepsilon^2 \Delta \right)^{1/2} , V(x) \right]}_{\textfrak{S}^{\infty}}
&\leq 2 \varepsilon \int   \abs{k} \abs{\mathcal{F}[V](k)}\,dk .
\end{align}
Inequality \eqref{estimate:commutator regularized fractional laplacian with potential} can be proven in a similar way. 

\end{proof}

\section{Proof of Theorem \ref{theorem:main theorem}}\label{section:proof of main theorem}

\subsection{Quantization of the Vlasov--Maxwell equations}

\begin{lemma}
Let $(\widetilde{W}_{N,t}, \widetilde{\alpha}_t )$ be the solution of \eqref{eq: Vlasov-Maxwell different style of writing} and $\widetilde{\omega}_{N,t}$ be the Weyl quantization of $\widetilde{W}_{N,t}$. Then $(\widetilde{\omega}_{N,t} , \alpha_t )$ satisfies
\begin{align}
\label{eq:Vlasov-Maxwell quantum version}
\begin{cases}
i \varepsilon \partial_t \widetilde{\omega}_{N,t}  &= \left[ - \varepsilon^2 \Delta, \widetilde{\omega}_{N,t} \right] + B_{t} + C_t, 
\\
\partial_t \widetilde{\alpha}_t(k,\lambda) &= \abs{k} \alpha_t(k,\lambda) - \sqrt{\frac{4 \pi^3}{\abs{k}}} \mathcal{F}[\kappa](k) \vep_{\lambda}(k) \cdot \mathcal{F}[\vJ_{\widetilde{\omega}_t, \widetilde{\alpha}_t}](k) .
\end{cases}
\end{align}
with
\begin{align}
B_t(x;y) &=  \nabla (K * \rho_{\widetilde{\omega}_{N,t}}) \left( \frac{x + y}{2} \right) \cdot (x - y) \, \widetilde{\omega}_{N,t}(x;y) ,
\\
\label{eq:definition C-t}
C_t(x;y) &= 2 \kappa *  \vA_{\widetilde{\alpha}_t} \left( \frac{x + y}{2} \right) \cdot \left[ i \varepsilon \nabla, \widetilde{\omega}_{N,t} \right](x;y)
+ \nabla (\kappa * \vA_{\widetilde{\alpha}_t})^2 \left( \frac{x+y}{2} \right) \cdot (x - y) \, \widetilde{\omega}_{N,t}(x;y)
\nonumber \\
&\qquad + \sum_{j=1}^3 \nabla (\kappa * \vA_{\widetilde{\alpha}_t}^{(j)}) \left( \frac{x + y}{2} \right) \cdot (x - y) \left\{ i \varepsilon \nabla^{(j)} , \widetilde{\omega}_{N,t} \right\}(x;y) .
\end{align}
Here, $K$, $\vA$, $\rho$ and $\vJ$ are defined as in \eqref{eq:Coulomb potential with cutoff}--\eqref{eq:definition charge current}.
\end{lemma}

\begin{proof}
To simplify the notation we refrain from writing $\sim$ on top of the symbols and write $W_{N,t}, \omega_{N,t}, \alpha_t, \ldots$ instead of $\widetilde{W}_{N,t}, \widetilde{\omega}_{N,t}, \widetilde{\alpha}_t,  \ldots$. Using that $W_{N}(t)$ satisfies the first equation in \eqref{eq: Vlasov-Maxwell different style of writing} and 
\begin{align}
\label{eq:relation between density for Wigner transform and its Weyl quantization}
\tilde{\rho}_{W_{N,t}}(x) =  \int  W_{N,t}(x,v)\,dv =  N^{-1} \omega_{N,t}(x,x) = \rho_{\omega_{N,t}}(x) .
\end{align}
we obtain 
\begin{subequations}
\begin{align}
&\partial_t \omega_{N,t}(x;y)
\nonumber \\
\label{eq:Weyl quantization time derivative derivation 1}
&\quad = - N \int e^{i v \cdot \frac{x-y}{\varepsilon}} 
2 \left( v - \kappa * \vA_{\alpha_t} \left( \frac{x + y}{2} \right) \right) \cdot \nabla_2 W_{N,t} \left( \frac{x+y}{2}, v \right)\,dv
\\
\label{eq:Weyl quantization time derivative derivation 2}
&\qquad + 
N \int  e^{i v \cdot \frac{x-y}{\varepsilon}} \,
(\nabla K * \rho_{\omega_t}) \left( \frac{x+y}{2} \right) \cdot \nabla_v 
W_{N,t} \left(\frac{x+y}{2}, v \right)\,dv
\\
\label{eq:Weyl quantization time derivative derivation 3}
&\qquad + 
N \int  e^{i v \cdot \frac{x-y}{\varepsilon}} 
\left( \nabla ( \kappa * \vA_{\alpha_t})^2 \right) \left( \frac{x + y}{2} \right)  \cdot \nabla_v 
W_{N,t} \left( \frac{x+y}{2}, v \right)\,dv
\\
\label{eq:Weyl quantization time derivative derivation 4}
&\qquad - 2 \sum_{j=1}^3
N \int  e^{i v \cdot \frac{x-y}{\varepsilon}} 
  v^{(j)}  \left( \nabla \kappa * \vA_{\alpha_t}^{(j)} \right) \left( \frac{x+y}{2} \right)  \cdot \nabla_v 
W_{N,t} \left( \frac{x+y}{2}, v \right)\,dv .
\end{align}
\end{subequations}
By straightforward manipulations the right hand side can be written in terms of the Weyl quantization of $W_{N,t}$. More explicitly, we obtain
\begin{align}
\eqref{eq:Weyl quantization time derivative derivation 1}
&= - i \varepsilon^{-1} 
\left(  \left[ - \varepsilon^2 \Delta, \omega_{N,t} \right](x;y) 
+ 2  \kappa * \vA_{\alpha_t} \left( \frac{x + y}{2} \right)  \cdot \left[ i \varepsilon \nabla,  \omega_{N,t} \right](x;y)
\right) ,
\nonumber \\
\eqref{eq:Weyl quantization time derivative derivation 2}
& = - i \varepsilon^{-1}  (\nabla K * \rho_{\omega_t}) \left( \frac{x+y}{2} \right) \cdot (x-y) \,  \omega_{N,t}(x;y),
\nonumber \\
\eqref{eq:Weyl quantization time derivative derivation 3}
& = - i \varepsilon^{-1}  \left( \nabla ( \kappa * \vA_{\alpha_t})^2 \right) \left( \frac{x + y}{2} \right)  \cdot (x-y) \, \omega_{N,t}(x;y) ,
\nonumber \\
\eqref{eq:Weyl quantization time derivative derivation 4}
&= - i \sum_{j=1}^3 \varepsilon^{-1} \left(  \nabla  \kappa * \vA_{\alpha_t}^{(j)} \right) \left( \frac{x+y}{2} \right)  (x - y)  \{ i \varepsilon \nabla^{(j)}, \omega_{N,t} \}(x;y) .
\end{align}
Plugging these expressions into the equality from above and multiplying by $i \varepsilon$ lead to the first equation in \eqref{eq:Vlasov-Maxwell quantum version}. 
Because of
\begin{align}
\{ i \varepsilon \nabla , \omega_{N,t} \}(x;y)
&= - 2 N  \int v W_{N,t} \left( \frac{x+y}{2} , v \right)  e^{i v \cdot \frac{x-y}{\varepsilon}}\,dv 
\end{align} 
we get
\begin{align}
\{ i \varepsilon \nabla , \omega_{N,t} \}(x;x)
&= - 2 N  \int v W_{N,t} \left(  x , v \right)\,dv.
\end{align} 
Together with \eqref{eq:relation between density for Wigner transform and its Weyl quantization} this enables us to write the charge current in the Vlasov--Maxwell equations \eqref{eq: Vlasov-Maxwell different style of writing} as
\begin{align}
\widetilde{\vJ}_{W_{N,t}, \alpha_t}(x) &= 2 \int  \left( v - (\kappa * \vA_{\alpha_t})(x) \right) W_{N,t}(x,v)\,dv
\nonumber \\
&= - N^{-1}  \{ i \varepsilon \nabla , \omega_{N,t} \}(x;x) - 2 \rho_{\omega_t}(x) (\kappa * \vA_{\alpha_t})(x)  ,
\end{align}
leading to the second equation in \eqref{eq:Vlasov-Maxwell quantum version}.

\end{proof}

 \subsection{Gr\"onwall estimate}

In this section, we prove Theorem \ref{theorem:main theorem}. This will be obtained by means of a Gr\"{o}nwall type estimate. Throughout this section, let $(\omega_{N,t}, \alpha_t)$ and $(\widetilde{\omega}_{N,t}, \widetilde{\alpha}_t)$ denote the solutions of \eqref{eq:Maxwell-Schroedinger equations} and \eqref{eq:Vlasov-Maxwell quantum version}
with initial data $(\omega_N, \alpha)$ and $(\widetilde{\omega}_{N}, \widetilde{\alpha})$, respectively. 

Using Duhamel's formula we write the difference between the mode functions of \eqref{eq:Maxwell-Schroedinger equations} and \eqref{eq:Vlasov-Maxwell quantum version} as
\begin{align}
\label{eq:difference alpha and alpha-tilde Duhamel expansion}
\alpha_t(k,\lambda) - \widetilde{\alpha}_t(k,\lambda)
&= e^{- i \abs{k} t} 
\left(  \alpha(k,\lambda)
- \widetilde{\alpha}(k,\lambda)
\right) 
\nonumber \\
&\quad 
+ i \int_0^t  e^{- i \abs{k} (t-s)}
\sqrt{\frac{4 \pi^3}{\abs{k}}} \mathcal{F}[\kappa](k) \vep_{\lambda}(k) \cdot 
\left( \mathcal{F}[\vJ_{\omega_{N,s}, \alpha_s}](k) - \mathcal{F}[\vJ_{\widetilde{\omega}_{N,s}, \widetilde{\alpha}_s}](k) \right)ds.
\end{align}
By means of \eqref{eq:assumption on the cutoff parameter} and \eqref{eq:Charge current L-infty estimate of the difference} with $(\omega, \alpha) = (\widetilde{\omega}_{N,s}, \widetilde{\alpha}_s)$ and $(\omega' , \alpha') = (\omega_{N,s}, \alpha_s)$ we estimate
\begin{align}
\label{eq:proof main theorem estimate for the difference of the fields}
&\norm{\alpha_t - \widetilde{\alpha}_t}_{\dot{\mathfrak{h}}_{-1/2} \, \cap \, \dot{\mathfrak{h}}_{1/2}}
-  \norm{\alpha - \widetilde{\alpha}}_{\dot{\mathfrak{h}}_{-1/2} \, \cap \, \dot{\mathfrak{h}}_{1/2}}
\nonumber \\
&\quad \leq C \int_0^t 
\norm{ \left( \abs{\cdot}^{-1} + 1 \right) \mathcal{F}[\kappa](k) \left(  \mathcal{F}[\vJ_{\omega_{N,s}, \alpha_s}] - \mathcal{F}[\vJ_{\widetilde{\omega}_{N,s}, \widetilde{\alpha}_s}] \right) }_{L^2(\mathbb{R}^3)} ds
\nonumber \\
&\quad \leq C \int_0^t 
\norm{\left( \abs{\cdot}^{-1} + 1 \right)  \mathcal{F}[\kappa]}_{L^2(\mathbb{R}^3)}
\norm{\left(  \mathcal{F}[\vJ_{\omega_{N,s}, \alpha_s}] - \mathcal{F}[\vJ_{\widetilde{\omega}_{N,s}, \widetilde{\alpha}_s}] \right) }_{L^{\infty}(\mathbb{R}^3)}ds
\nonumber \\
&\quad \leq C  N^{-1} \int_0^t  \Big( \big( 1 +  \norm{\widetilde{\alpha}_s}_{\dot{\mathfrak{h}}_{1/2}} \big) 
 \norm{\omega_{N,s} - \widetilde{\omega}_{N,s}}_{\textfrak{S}_{\varepsilon}^{1,1}}
+ \norm{\omega_{N,s}}_{\textfrak{S}^1}
\norm{\alpha_s - \widetilde{\alpha}_s}_{\dot{\mathfrak{h}}_{-1/2} \, \cap \, \dot{\mathfrak{h}}_{1/2}} \Big)ds.
\end{align}

Next, we are going to estimate the difference between particle operators $\omega_{N,t}$ and $\widetilde{\omega}_{N,t}$. To this end, it is convenient to define the operator 
\begin{align}
D &= \left( 1 - \varepsilon^2 \Delta \right)^{1/2} 
\end{align}
and the $\varepsilon$-dependent Sobolev norm
\begin{align}
\norm{\omega}_{\mathfrak{S}_{\varepsilon}^{1,1}} 
&= \norm{\left( 1 - \varepsilon^2 \Delta \right)^{1/2} \omega \left( 1 - \varepsilon^2 \Delta \right)^{1/2}}_{\mathfrak{S}^1}
= \norm{D \omega D}_{\mathfrak{S}^1} .
\end{align}
Note that for notational convenience we refrain from indicating the dependence of the operator on $\varepsilon$. For a rigorous treatment of the Duhamel formula below we introduce a regularized version of $D$ which is defined by
\begin{align}
\label{eq:regularized fractional Laplacian}
D_{\leq n} &= \frac{\left( 1 - \varepsilon^2 \Delta \right)^{1/2}}{\left( 1 - \varepsilon^2 \Delta / n^2 \right)^{1/2}}
\quad \text{with} \quad n \in \mathbb{N} .
\end{align}
Note that $D_{\leq n}$ is a bounded operator with operator norm $\norm{D_{\leq n}}_{\textfrak{S}^{\infty}} \leq n$
which allows us to write the $\varepsilon$-dependent Sobolev norm of $\omega \in \textfrak{S}^{1,1}$ as
\begin{align}
\label{eq:Sobolev tace norm as regularized limit}
\norm{\omega}_{\mathfrak{S}_{\varepsilon}^{1,1}} 
&= \lim_{n \rightarrow \infty} \norm{D_{\leq n} \omega D_{\leq n}}_{\textfrak{S}^{1}} .
\end{align}
In addition, we define the time-dependent Hamiltonian
\begin{align}
\label{eq:definition Hamiltonian with vector potential and interaction potential}
h(t) = \left( - i \varepsilon \nabla - \kappa * \vA_{\alpha_t} \right)^2 + K * \rho_{\omega_{N,t}} 
\end{align}
and the two parameter group $U(t;s)$ satisfying
\begin{align}
\label{eq:definition unitary with vector potential and interaction potential}
i \varepsilon \partial_t U(t;s) = h(t) U(t;s)
\quad \text{and} \quad
U(s;s) = 1 .
\end{align}
Moreover, let $U(t) = U(t;0)$. Using \eqref{eq:Maxwell-Schroedinger equations},  \eqref{eq:Vlasov-Maxwell quantum version} and the linearity of the mappings $\alpha \mapsto \vA_{\alpha}$ and $\omega \mapsto \rho_{\omega}$
we compute 
\begin{align}
&i \varepsilon \partial_t 
\left( U^*(t) D_{\leq n} \left( \omega_{N,t} - \widetilde{\omega}_{N,t} \right) D_{\leq n} U(t) \right)
\nonumber \\
&\quad = - U^*(t) \left[ h(t) , D_{\leq n} \left( \omega_{N,t} - \widetilde{\omega}_{N,t} \right) D_{\leq n} \right] U(t)
\nonumber \\
&\qquad +  U^*(t) D_{\leq n} \left( i \varepsilon \partial_t  \left( \omega_{N,t} - \widetilde{\omega}_{N,t} \right) \right) D_{\leq n} U(t)
\nonumber \\
&\quad = U^*(t) D_{\leq n} \left[ \left( (\kappa * \vA_{\alpha_t})^2 -  (\kappa * \vA_{\widetilde{\alpha}_t})^2 + 2 \kappa * \vA_{\alpha_t - \widetilde{\alpha}_t} \cdot i \varepsilon \nabla \right) , \widetilde{\omega}_{N,t} \right]  D_{\leq n} U(t)
\nonumber \\
&\qquad
+ U^*(t) D_{\leq n} \left[ K * \rho_{\omega_{N,t} - \widetilde{\omega}_{N,t}} , \widetilde{\omega}_{N,t} \right] D_{\leq n} U(t)
\nonumber \\
&\qquad
- U^*(t) D_{\leq n} \left[ X_{\omega_{N,t}} , \omega_{N,t} \right] D_{\leq n} U(t)
\nonumber \\
&\qquad 
+ U^*(t) D_{\leq n} \left(  \left[ K* \rho_{\widetilde{\omega}_{N,t}}, \widetilde{\omega}_{N,t} \right] - B_t \right) D_{\leq n} U(t)
\nonumber \\
&\qquad 
+ U^*(t) D_{\leq n} \left(  \left[ (\kappa * \vA_{\widetilde{\alpha}_t})^2 + 2 \kappa * \vA_{\widetilde{\alpha}_t} \cdot i \varepsilon \nabla , \widetilde{\omega}_{N,t} \right]  - C_t \right) D_{\leq n} U(t)
\nonumber \\
&\qquad - U^*(t) D_{\leq n} \left( \omega_{N,t} - \widetilde{\omega}_{N,t} \right) \left[ h(t) , D_{\leq n} \right] U(t)
\nonumber \\
&\qquad - U^*(t) \left[ h(t) , D_{\leq n} \right] \left( \omega_{N,t} - \widetilde{\omega}_{N,t} \right) D_{\leq n}  U(t) .
\end{align}
If we then apply Duhamel's formula, take the trace norm, take the limit $n \rightarrow \infty$ and use \eqref{eq:Sobolev tace norm as regularized limit}, as well as $\norm{D_{\leq n} D^{-1}}_{\textfrak{S}^{\infty}} \leq 1$, we obtain
\begin{subequations}
\begin{align}
\label{eq: Sobolev norm estimate omega 1}
\norm{\omega_{N,t} - \widetilde{\omega}_{N,t} }_{\textfrak{S}_{\varepsilon}^{1,1}}
&\leq  \norm{\omega_{N} - \widetilde{\omega}_{N} }_{\textfrak{S}_{\varepsilon}^{1,1}}
\\
\label{eq: Sobolev norm estimate omega 2}
&\quad + \varepsilon^{-1} \int_0^t  
\norm{\left[ \left(  (\kappa * \vA_{\alpha_s})^2 - (\kappa * \vA_{\widetilde{\alpha}_s})^2  \right) , \widetilde{\omega}_{N,s} \right]}_{\textfrak{S}_{\varepsilon}^{1,1}}ds
\\
\label{eq: Sobolev norm estimate omega 3}
&\quad + 2 \varepsilon^{-1} \int_0^t  
\norm{\left[  \kappa * \vA_{\alpha_s - \widetilde{\alpha}_s}  \cdot i \varepsilon \nabla  , \widetilde{\omega}_{N,s} \right]  }_{\textfrak{S}_{\varepsilon}^{1,1}}ds
\\
\label{eq: Sobolev norm estimate omega 4}
&\quad + \varepsilon^{-1} \int_0^t 
\norm{\left[ K * \rho_{\omega_{N,s} - \widetilde{\omega}_{N,s}} , \widetilde{\omega}_{N,s} \right]}_{\textfrak{S}_{\varepsilon}^{1,1}}ds
\\
\label{eq: Sobolev norm estimate omega 5}
&\quad + \varepsilon^{-1} \int_0^t  
\norm{\left[ X_{\omega_{N,s}} , \omega_{N,s} \right] }_{\textfrak{S}_{\varepsilon}^{1,1}}ds
\\
\label{eq: Sobolev norm estimate omega 6}
&\quad + \varepsilon^{-1} \int_0^t 
\norm{ \left[ K* \rho_{\widetilde{\omega}_{N,s}}, \widetilde{\omega}_{N,s} \right] - B_s }_{\textfrak{S}_{\varepsilon}^{1,1}}ds
\\
\label{eq: Sobolev norm estimate omega 7}
&\quad + \varepsilon^{-1} \int_0^t 
\norm{\left[ (\kappa * \vA_{\widetilde{\alpha}_s})^2 + 2 \kappa * \vA_{\widetilde{\alpha}_s} \cdot i \varepsilon \nabla , \widetilde{\omega}_{N,s} \right]  - C_s }_{\textfrak{S}_{\varepsilon}^{1,1}}ds
\\
\label{eq: Sobolev norm estimate omega 8}
&\quad + \varepsilon^{-1} \lim_{n \rightarrow \infty} \int_0^t  
\norm{D \left( \omega_{N,s} - \widetilde{\omega}_{N,s} \right) \left[ h(s) , D_{\leq n} \right] }_{\textfrak{S}^1}ds
\\
\label{eq: Sobolev norm estimate omega 9}
&\quad + \varepsilon^{-1} \lim_{n \rightarrow \infty} \int_0^t 
\norm{\left[ h(s) , D_{\leq n} \right] \left( \omega_{N,s} - \widetilde{\omega}_{N,s} \right) D}_{\textfrak{S}^1}ds .
\end{align}
\end{subequations}
In the following, we estimate each line separately.

\paragraph{The term \eqref{eq: Sobolev norm estimate omega 2}:}
Using
\begin{align}
\label{eq:commutator between frational Laplacians}
D [A, B] D &= [A, DBD] +  [D,A] BD + DB[D, A] 
\end{align}
and writing the difference of the vector potentials as 
\begin{align}
(\kappa * \vA_{\alpha_s})^2 - (\kappa * \vA_{\widetilde{\alpha}_s})^2 
&= \kappa * \vA_{\alpha_s + \widetilde{\alpha}_s} \cdot
\kappa * \vA_{\alpha_s - \widetilde{\alpha}_s}
\end{align}
we obtain
\begin{subequations}
\begin{align}
&\norm{\left[ \left( (\kappa * \vA_{\alpha_s})^2 - (\kappa * \vA_{\widetilde{\alpha}_s})^2  \right) , \widetilde{\omega}_{N,s} \right]}_{\textfrak{S}_{\varepsilon}^{1,1}}
\nonumber \\
&\quad \leq 
\norm{\left[ \left( (\kappa * \vA_{\alpha_s})^2 - (\kappa * \vA_{\widetilde{\alpha}_s})^2  \right) , D \widetilde{\omega}_{N,s} D \right]}_{\textfrak{S}^{1}}
\nonumber \\
&\qquad +
\norm{\left[  D, \left( (\kappa * \vA_{\alpha_s})^2 - (\kappa * \vA_{\widetilde{\alpha}_s})^2  \right) \right]  \widetilde{\omega}_{N,s} D }_{\textfrak{S}^{1}}
\nonumber \\
&\qquad +
\norm{D \widetilde{\omega}_{N,s} \left[  D, \left( (\kappa * \vA_{\alpha_s})^2 - (\kappa * \vA_{\widetilde{\alpha}_s})^2  \right) \right] }_{\textfrak{S}^{1}}
\nonumber \\
\label{eq: Sobolev norm estimate omega 2a}
&\quad \leq
\norm{\kappa * \vA_{\alpha_s + \widetilde{\alpha}_s}}_{\textfrak{S}^{\infty}}
\norm{\left[ \kappa * \vA_{\alpha_s - \widetilde{\alpha}_s} , D \widetilde{\omega}_{N,s} D \right]}_{\textfrak{S}^{1}}
\\
\label{eq: Sobolev norm estimate omega 2b}
&\qquad +
\norm{\kappa * \vA_{\alpha_s - \widetilde{\alpha}_s}}_{\textfrak{S}^{\infty}}
\norm{\left[ \kappa * \vA_{\alpha_s + \widetilde{\alpha}_s} , D \widetilde{\omega}_{N,s} D \right]}_{\textfrak{S}^{1}}
\\
\label{eq: Sobolev norm estimate omega 2c}
&\qquad +
\norm{\left[  D, \left( (\kappa * \vA_{\alpha_s})^2 - (\kappa * \vA_{\widetilde{\alpha}_s})^2 \right) \right] }_{\textfrak{S}^{\infty}}
\left(
\norm{\widetilde{\omega}_{N,s} D}_{\textfrak{S}^{1}}
+
\norm{D \widetilde{\omega}_{N,s} }_{\textfrak{S}^{1}}
\right) .
\end{align}
\end{subequations}
Using the Fourier expansion of the regularized vector potential
\begin{align}
\label{eq:Fourier expansion regularized vector potential}
\kappa * \vA_{\alpha}(x) =  2 \text{Re} \sum_{\lambda = 1,2} \int  \frac{\mathcal{F}[\kappa](k)}{\sqrt{2 \abs{k}}} \vep_{\lambda}(k) e^{i k x} \alpha(k,\lambda)\,dk,
\end{align}
\eqref{estimate:commutator function with operator in trace norm},
the Cauchy-Schwarz inequality and \eqref{eq:assumption on the cutoff parameter} we estimate
\begin{align}
\label{eq:Sobolev trace norm estimate commuator vector potential with omega}
\norm{ \left[ \kappa * \vA_{\alpha} , D \widetilde{\omega}_{N,s} D \right] }_{\textfrak{S}^1}
&\leq 2
\sum_{\lambda = 1,2} \int  \abs{k}^{1/2} \abs{\mathcal{F}[\kappa](k)} \abs{\alpha(k,\lambda)}   
\norm{  \left[ \hat{x} , D \widetilde{\omega}_{N,s} D \right] }_{\textfrak{S}^1}dk
\nonumber \\
&\leq  C  \norm{\mathcal{F}[\kappa](k)}_{L^2(\mathbb{R}^3)} \min \left\{ \norm{\alpha}_{\mathfrak{h}} , \norm{\alpha}_{\dot{\mathfrak{h}}_{1/2}} \right\}
\norm{  \left[ \hat{x} , D \widetilde{\omega}_{N,s} D \right] }_{\textfrak{S}^1}
\nonumber \\
&\leq  C  
\min \left\{ \norm{\alpha}_{\mathfrak{h}} , \norm{\alpha}_{\dot{\mathfrak{h}}_{1/2}} \right\}
\norm{  \left[ \hat{x} , D \widetilde{\omega}_{N,s} D \right] }_{\textfrak{S}^1} .
\end{align}
Together with \eqref{eq:vector potential L-infty estimate} this leads to
\begin{align}
\eqref{eq: Sobolev norm estimate omega 2a} +   \eqref{eq: Sobolev norm estimate omega 2b}
&\leq  C  \left( \norm{\alpha_s}_{\dot{\mathfrak{h}}_{1/2}} + \norm{\widetilde{\alpha}_s}_{\dot{\mathfrak{h}}_{1/2}} \right) \norm{\alpha_s - \widetilde{\alpha}_s}_{\mathfrak{h}}
\norm{  \left[ \hat{x} , D \widetilde{\omega}_{N,s} D \right] }_{\textfrak{S}^1}
\end{align}
By means of \eqref{eq:Fourier expansion regularized vector potential}, \eqref{estimate:commutator fractional laplacian with potential}, \eqref{eq:vector potential L-infty estimate} and \eqref{eq:assumption on the cutoff parameter} we estimate 
\begin{align}
&\norm{\left[  D, \left( (\kappa * \vA_{\alpha_s})^2 - (\kappa * \vA_{\widetilde{\alpha}_s})^2 \right) \right] }_{\textfrak{S}^{\infty}}
\nonumber \\
&\quad \leq 
\norm{\kappa * \vA_{\alpha_s+ \tilde{\alpha}_s}}_{\textfrak{S}^{\infty}}
\norm{\left[  D, \kappa * \vA_{\alpha_s - \tilde{\alpha}_s} \right] }_{\textfrak{S}^{\infty}} 
+
\norm{\kappa * \vA_{\alpha_s - \tilde{\alpha}_s} }_{\textfrak{S}^{\infty}}
\norm{\left[  D, \kappa * \vA_{\alpha_s + \tilde{\alpha}_s} \right] }_{\textfrak{S}^{\infty}}
\nonumber \\
&\quad \leq C \varepsilon \norm{\alpha_s + \widetilde{\alpha}_s}_{\mathfrak{h}} \sum_{\lambda = 1,2} \int  \abs{k}^{1/2} \abs{\mathcal{F}[\kappa](k)} \abs{\alpha_s(k,\lambda) - \widetilde{\alpha}_s(k,\lambda)}\,dk
\nonumber \\
&\qquad + C \varepsilon \norm{\alpha_s - \widetilde{\alpha}_s}_{\mathfrak{h}} \sum_{\lambda = 1,2} \int  \abs{k}^{1/2} \abs{\mathcal{F}[\kappa](k)} \abs{\alpha_s(k,\lambda) + \widetilde{\alpha}_s(k,\lambda)}\,dk
\nonumber \\
&\quad \leq  C \varepsilon \left(  \norm{\alpha_s}_{\dot{\mathfrak{h}}_{1/2}}
+ \norm{\widetilde{\alpha}_s}_{\dot{\mathfrak{h}}_{1/2}} \right) \norm{\alpha_s - \widetilde{\alpha}_s}_{\mathfrak{h}} .
\end{align}
In total, we get
\begin{align}
\norm{\left[ \left( \vAk^2(s) - \vAkt^2(s) \right) , \widetilde{\omega}_{N,s} \right]}_{\textfrak{S}_{\varepsilon}^{1,1}}
&\leq  C  \left( \norm{\alpha_s}_{\dot{\mathfrak{h}}_{1/2}} + \norm{\widetilde{\alpha}_s}_{\dot{\mathfrak{h}}_{1/2}} \right) \norm{\alpha_s - \widetilde{\alpha}_s}_{\mathfrak{h}}
\bigg(
\norm{  \left[ \hat{x} , D \widetilde{\omega}_{N,s} D \right] }_{\textfrak{S}^1}
\nonumber \\
&\qquad +  \varepsilon \norm{\widetilde{\omega}_{N,s} D}_{\textfrak{S}^{1}}
+ \varepsilon \norm{D \widetilde{\omega}_{N,s} }_{\textfrak{S}^{1}}
\bigg) .
\end{align}
Using Lemma \ref{lemma:trace norm estimates with gradients and commutators} we, finally, get
\begin{align}
\eqref{eq: Sobolev norm estimate omega 2}
&\leq 
C N \int_0^t  \sum_{j=0}^6 \varepsilon^{j} \norm{\widetilde{W}_{N,s}}_{H_6^{j+1}}  \left( \norm{\alpha_s}_{\dot{\mathfrak{h}}_{1/2}} + \norm{\widetilde{\alpha}_s}_{\dot{\mathfrak{h}}_{1/2}} \right) \norm{\alpha_s - \widetilde{\alpha}_s}_{\mathfrak{h}}ds .
\end{align}

\paragraph{The term \eqref{eq: Sobolev norm estimate omega 3}:}
Using \eqref{eq:commutator between frational Laplacians} and $\left( \nabla \cdot  \kappa * \vA_{\alpha_s - \widetilde{\alpha}_s} \right) = 0$ we obtain
\begin{align}
D \left[  \kappa * \vA_{\alpha_s - \widetilde{\alpha}_s} \cdot i \varepsilon \nabla  , \widetilde{\omega}_{N,s} \right]  D
&=
\left[  \kappa * \vA_{\alpha_s - \widetilde{\alpha}_s}  , D \widetilde{\omega}_{N,s} D \cdot i \varepsilon \nabla\right]   
\nonumber \\
&\quad +
\kappa * \vA_{\alpha_s - \widetilde{\alpha}_s} \cdot 
\left[  i \varepsilon \nabla  , D \widetilde{\omega}_{N,s} D \right]  
\nonumber \\
&\quad +
\left[ D , \kappa * \vA_{\alpha_s - \widetilde{\alpha}_s} \right]
\cdot i \varepsilon \nabla  \widetilde{\omega}_{N,s}  D
\nonumber \\
&\quad + 
D \widetilde{\omega}_{N,s}  i \varepsilon \nabla \cdot \left[ D , \kappa * \vA_{\alpha_s - \widetilde{\alpha}_s} \right]
\end{align}
and
\begin{align}
\norm{\left[  \kappa * \vA_{\alpha_s - \widetilde{\alpha}_s} \cdot i \varepsilon \nabla  , \widetilde{\omega}_{N,s} \right]  }_{\textfrak{S}_{\varepsilon}^{1,1}}
&\leq \varepsilon \norm{\left[ D , \kappa * \vA_{\alpha_s - \widetilde{\alpha}_s} \right]}_{\textfrak{S}^{\infty}}
\left(
\norm{ \nabla  \widetilde{\omega}_{N,s}  D}_{\textfrak{S}^1}
+ \norm{D \widetilde{\omega}_{N,s}  \nabla}_{\textfrak{S}^1}
\right)
\nonumber \\
&\quad + \varepsilon
\norm{\kappa * \vA_{\alpha_s - \widetilde{\alpha}_s}}_{\textfrak{S}^{\infty}}
\norm{\left[  \nabla  , D \widetilde{\omega}_{N,s} D \right]  }_{\textfrak{S}^1}
\nonumber \\
&\quad + \varepsilon
\norm{\left[  \kappa * \vA_{\alpha_s - \widetilde{\alpha}_s}  , D \widetilde{\omega}_{N,s} D  \nabla \right] }_{\textfrak{S}^1} .
\end{align}
By \eqref{estimate:commutator fractional laplacian with potential}, \eqref{eq:estimates for the vector potential and interaction 0} and \eqref{eq:Sobolev trace norm estimate commuator vector potential with omega} with $D \widetilde{\omega}_{N,s} D$ being replaced by $D\widetilde{\omega}_{N,s} D \nabla$ we obtain in analogy to the estimates above
\begin{align}
\norm{\left[  \kappa * \vA_{\alpha_s - \widetilde{\alpha}_s} \cdot i \varepsilon \nabla  , \widetilde{\omega}_{N,s} \right]  }_{\textfrak{S}_{\varepsilon}^{1,1}}
&\leq C \varepsilon \norm{\alpha_s - \widetilde{\alpha}_s}_{\mathfrak{h}} 
\Big(
\norm{ \varepsilon \nabla  \widetilde{\omega}_{N,s}  D}_{\textfrak{S}^1}
+ \norm{D \widetilde{\omega}_{N,s}  \varepsilon \nabla}_{\textfrak{S}^1}
\nonumber \\
&\qquad \quad 
+ \norm{\left[  \nabla  , D \widetilde{\omega}_{N,s} D \right]  }_{\textfrak{S}^1}
+ \norm{  \left[ \hat{x} , D \widetilde{\omega}_{N,s} D   \nabla \right] }_{\textfrak{S}^1}
\Big) .
\end{align}
Together with Lemma \ref{lemma:trace norm estimates with gradients and commutators} this leads to
\begin{align}
\eqref{eq: Sobolev norm estimate omega 3}
&\leq 
C  N \int_0^t 
\sum_{j=0}^7 \varepsilon^{j} \norm{\widetilde{W}_{N,s}}_{H_7^{j+1}}
\norm{\alpha_s - \widetilde{\alpha}_s}_{\mathfrak{h}} ds .
\end{align}

\paragraph{The term \eqref{eq: Sobolev norm estimate omega 4}:}

Using \eqref{eq:commutator between frational Laplacians}
we obtain
\begin{align}
D  \left[ K * \rho_{\omega_{N,s} - \widetilde{\omega}_{N,s}} , \widetilde{\omega}_{N,s} \right]  D
&= 
D \widetilde{\omega}_{N,s} 
\left[ D , K * \rho_{\omega_{N,s} - \widetilde{\omega}_{N,s}} \right] 
\nonumber \\
&\quad +
\left[ D , K * \rho_{\omega_{N,s} - \widetilde{\omega}_{N,s}} \right] \widetilde{\omega}_{N,s} D
\nonumber \\
&\quad +
\left[ K * \rho_{\omega_{N,s} - \widetilde{\omega}_{N,s}} , D \, \widetilde{\omega}_{N,s} D \right]  
\end{align}
and estimate 
\begin{align}
\eqref{eq: Sobolev norm estimate omega 4}
&\leq  \varepsilon^{-1} \int_0^t 
\norm{\left[ D , K * \rho_{\omega_{N,s} - \widetilde{\omega}_{N,s}} \right]}_{\textfrak{S}^{\infty}}
\left( \norm{D \widetilde{\omega}_{N,s}}_{\textfrak{S}^1}
+
\norm{\widetilde{\omega}_{N,s} D}_{\textfrak{S}^1}
\right)\,ds
\nonumber \\
&\quad +  \varepsilon^{-1} \int_0^t 
\norm{\left[ K * \rho_{\omega_{N,s} - \widetilde{\omega}_{N,s}} , D \, \widetilde{\omega}_{N,s} D \right]  }_{\textfrak{S}^1}ds .
\end{align}

By means of \eqref{estimate:commutator fractional laplacian with potential}, \eqref{eq:Fourier transform of the potential}, \eqref{eq: estimate to bound the density by the L1-norm of omega} and  \eqref{eq:assumption on the cutoff parameter}
we get 
\begin{align}
\norm{\left[ D , K * \rho_{\omega_{N,s} - \widetilde{\omega}_{N,s}} \right]}_{\textfrak{S}^{\infty}}
&\leq C \varepsilon \int  \abs{k}  \abs{\mathcal{F}[K](k)}
\abs{ \mathcal{F}[\rho_{\omega_{N,s} - \widetilde{\omega}_{N,s}}](k) }\,dk 
\nonumber \\
&\leq C \varepsilon \norm{\rho_{\omega_{N,s} - \widetilde{\omega}_{N,s}}}_{L^1(\mathbb{R}^3)}  \norm{\abs{\cdot}^{-1/2} \mathcal{F}[\kappa]}_{L^2(\mathbb{R}^3)}^2
\nonumber \\
&\leq C N^{-1} \varepsilon \norm{\omega_{N,s} - \widetilde{\omega}_{N,s}}_{\textfrak{S}^1}
\end{align}
Similarly, we get 
\begin{align}
\norm{\left[ K * \left( \rho_s - \widetilde{\rho}_s \right) , D \, \widetilde{\omega}_{N,s} D \right]  }_{\textfrak{S}^1}
&\leq C  \norm{\left[ \hat{x} , D \, \widetilde{\omega}_{N,s} D \right]  }_{\textfrak{S}^1}
\int  \abs{k}  \abs{\mathcal{F}[K](k)}
\abs{ \mathcal{F}[\rho_{\omega_{N,s} - \widetilde{\omega}_{N,s}}](k) }\,dk
\nonumber \\
&\leq C N^{-1} \varepsilon \norm{\omega_{N,s} - \widetilde{\omega}_{N,s}}_{\textfrak{S}^1}
\norm{\left[ \hat{x} , D \, \widetilde{\omega}_{N,s} D \right]  }_{\textfrak{S}^1}
\end{align}
by using \eqref{estimate:commutator function with operator in trace norm} and the Fourier decomposition
\begin{align}
\mathcal{F}[K *  \rho_{\omega_{N,s} - \widetilde{\omega}_{N,s}}](k)
&= (2 \pi)^{3/2} \mathcal{F}[K](k)
\mathcal{F}[\rho_{\omega_{N,s} - \widetilde{\omega}_{N,s}}](k) .
\end{align}
Together with Lemma \ref{lemma:trace norm estimates with gradients and commutators} this shows
\begin{align}
\eqref{eq: Sobolev norm estimate omega 4}
&\leq C N^{-1} \int_0^t  \norm{\omega_{N,s} - \widetilde{\omega}_{N,s}}_{\textfrak{S}^1}
\left( \norm{D \widetilde{\omega}_{N,s}}_{\textfrak{S}^1}
+
\norm{\widetilde{\omega}_{N,s} D}_{\textfrak{S}^1}
+ \varepsilon^{-1}
\norm{\left[ \hat{x} , D \, \widetilde{\omega}_{N,s} D \right]  }_{\textfrak{S}^1}
\right)\,ds
\nonumber \\
&\leq C  \int_0^t  \norm{\omega_{N,s} - \widetilde{\omega}_{N,s}}_{\textfrak{S}^1}
\sum_{j=0}^6 \varepsilon^j \norm{\widetilde{W}_{N,s}}_{H_6^{j+1}}ds .
\end{align}

\paragraph{The term \eqref{eq: Sobolev norm estimate omega 5}:}
Using the Fourier decomposition
\begin{align}
K(x-y) &= \frac{1}{(2 \pi)^{3/2}} \int e^{i k (x-y)} \widehat{K}(k)\,dk
\end{align}
we write the exchange term as 
\begin{align}
X_s &= \frac{1}{(2 \pi)^{3/2} N} \int  \widehat{K}(k) e^{i k \hat{x}} \omega_{N,s} e^{- i k \hat{x}}dk
\end{align}
and estimate
\begin{align}
\eqref{eq: Sobolev norm estimate omega 5}
&= \varepsilon^{-1} \int_0^t  
\norm{ D \left[ X_s , \omega_{N,s} \right] D}_{\textfrak{S}^{1}}ds
\nonumber \\
&\leq \varepsilon^{-1} N^{-1}
\int_0^t   \int  \abs{\widehat{K}(k)}
\norm{ D \left[ e^{i k \hat{x}} \omega_{N,s} e^{- i k \hat{x}} , \omega_{N,s} \right]  D}_{\textfrak{S}^{1}} dk\,ds
\nonumber \\
&\leq \varepsilon^{-1} N^{-1}
\int_0^t   \int  \abs{\widehat{K}(k)}
\bigg(
\norm{ D e^{i k \hat{x}} \omega_{N,s} e^{- i k \hat{x}}  \omega_{N,s}   D}_{\textfrak{S}^{1}}
\nonumber \\
&\qquad \qquad \qquad \qquad \qquad \qquad \qquad  +
\norm{ D \omega_{N,s}  e^{i k \hat{x}} \omega_{N,s} e^{- i k \hat{x}}   D}_{\textfrak{S}^{1}}
\bigg)\,dk\,ds
\nonumber \\
&\leq \varepsilon^{-1} N^{-1}
\int_0^t   \norm{\omega_{N,s}^{1/2}}_{\textfrak{S}^{\infty}}^2  \int  \abs{\widehat{K}(k)}
\bigg(
\norm{D e^{i k \hat{x}} \omega_{N,s}^{1/2}}_{\textfrak{S}^{2}}
\norm{\omega_{N,s}^{1/2} D}_{\textfrak{S}^{2}}
\nonumber \\
&\qquad \qquad \qquad \qquad \qquad \qquad \qquad  +
 \norm{D \omega_{N,s}^{1/2} }_{\textfrak{S}^{2}}
\norm{\omega_{N,s}^{1/2} e^{i k \hat{x}} D}_{\textfrak{S}^{2}}
\bigg)\,dk\,ds .
\end{align}
By means of $\norm{ D   \omega_{N,s}^{1/2} }_{\textfrak{S}^{2}}^2
= \tr \left( D^2 \omega_{N,s}  \right)$ and
\begin{align}
\norm{ D  e^{i k \hat{x}} \omega_{N,s}^{1/2} }_{\textfrak{S}^{2}}^2 
&=  \tr \left( \omega_{N,s} e^{- i k \hat{x} } D^2 e^{i k \hat{x}} \right) 
\nonumber \\
&=  \tr \left(  \left( 1 - \varepsilon^2 \Delta + \varepsilon^2 k^2 - 2 \varepsilon k \cdot i \varepsilon \nabla \right)  \omega_{N,s} \right)
\nonumber \\
&\leq 2  \tr \left(  \left( 1 - \varepsilon^2 \Delta + \varepsilon^2 k^2  \right)  \omega_{N,s} \right)
\nonumber \\
&= 2  \tr \left(  \left( D^2 + \varepsilon^2 k^2  \right)  \omega_{N,s} \right),
\end{align}
together with \eqref{eq:Fourier transform of the potential} and \eqref{eq:assumption on the cutoff parameter} we obtain
\begin{align}
\eqref{eq: Sobolev norm estimate omega 5}
&\leq C \varepsilon^{-1} N^{-1}
\int_0^t   \norm{\omega_{N,s}^{1/2}}_{\textfrak{S}^{\infty}}^2  \int  \abs{\widehat{K}(k)}
\sqrt{\tr \left(  \left( D^2 + \varepsilon^2 k^2  \right)  \omega_{N,s} \right)}
\sqrt{\tr \left( D^2 \omega_{N,s}  \right)}\,dk\,ds
\nonumber \\
&\leq C \varepsilon^{-1} N^{-1}
\int_0^t   \norm{\omega_{N,s}^{1/2}}_{\textfrak{S}^{\infty}}^2
\Big( \tr \left( \omega_{N,s} \right) + \tr \left( - \varepsilon^2 \Delta \omega_{N,s} \right) \Big)\,ds .
\end{align}
Note that $\omega_{N,s}$ is a trace class operator, ensuring the existence of a spectral set $\{ \lambda_j, \varphi_j \}_{j \in \mathbb{N}}$ with $\lambda_j \geq 0$ for all $j \in \mathbb{N}$ such that $\omega_{N,s} = \sum_{j \in \mathbb{N}} \ket{\varphi_j} \bra{\varphi_j}$. This allows us to estimate the operator norm of $\omega_{N,s}$ by
\begin{align}
\norm{\omega_{N,s}^{1/2}}_{\textfrak{S}^{\infty}}^2
&\leq \sup_{j \in \mathbb{N}} \{ \lambda_j  \}
\leq \Big( \sum_{j \in \mathbb{N}} \lambda_j^3 \Big)^{1/3}
= \norm{\omega_{N,s}}_{\textfrak{S}^3} 
\end{align}
and leads to
\begin{align}
\eqref{eq: Sobolev norm estimate omega 5}
&\leq  C \varepsilon^{-1} N^{-1}
\int_0^t   \norm{\omega_{N,s}}_{\textfrak{S}^3} 
\Big( \norm{\omega_{N,s}}_{\textfrak{S}^1}  + \tr \left( - \varepsilon^2 \Delta \omega_{N,s} \right) \Big)\,ds
\end{align}
With this regard note that $\norm{\omega_{N,s}}_{\textfrak{S}^3} 
= \norm{\omega_{N,0}}_{\textfrak{S}^3}$ because of 
\begin{align}
i \varepsilon \frac{d}{dt} \norm{\omega_{N,s}}_{\textfrak{S}^3}^3 &= 
 \tr \left( \left[ \left( - i \varepsilon \nabla - \kappa * \vA_{\alpha_s} \right)^2  + K * \rho_{\omega_s} - X_{\omega_s}  , \omega_{N,s}^3 \right] \right)
= 0
\end{align}
and that $\norm{\omega_{N,0}}_{\textfrak{S}^3}
= \tr \left( \omega_{N,0}^3 \right)^{1/3} \leq \norm{\omega_{N,0}}_{\textfrak{S}^{\infty}}^{2/3} \norm{\omega_{N,0}}_{\textfrak{S}^{1}}^{1/3}$. In combination with \eqref{eq:estimate for the MS energy functional 2}, the conservation of mass and energy, and the assumptions on the initial data we get
\begin{align}
\eqref{eq: Sobolev norm estimate omega 5}
&\leq  C \varepsilon^{-1} N^{-2/3} \left( \mathcal{E}^{\rm{MS}}[\omega_0,\alpha_0] + N \right) t .
\end{align}

\paragraph{The term \eqref{eq: Sobolev norm estimate omega 6}:}


The next term, will be estimated by similar means as in \cite[Chapter 3]{BPSS2016}. Using
\begin{align}
\norm{\left( 1 - \varepsilon^2 \Delta \right)^{-1} \left( 1 + x^2 \right)^{-1}}_{\textfrak{S}^2}
&\leq C \sqrt{N} 
\end{align}
we obtain
\begin{align}
\eqref{eq: Sobolev norm estimate omega 6}
&\leq C N^{1/2} \varepsilon^{-1} \int_0^t 
\norm{ \left( 1 + x^2 \right) D^3 \left( \left[ K* \rho_{\widetilde{\omega}_{N,s}}, \widetilde{\omega}_{N,s} \right] - B_s \right) D }_{\textfrak{S}^{2}}ds
\end{align}
Since
\begin{align}
\norm{AD}_{\textfrak{S}^2}^2 
&= \norm{A}_{\textfrak{S}^2}^2 
+ \norm{A i \varepsilon \nabla}_{\textfrak{S}^2}^2 ,
\end{align}
\begin{align}
\norm{DA}_{\textfrak{S}^2}^2 
&= \norm{A}_{\textfrak{S}^2}^2 
+ \norm{i \varepsilon \nabla A}_{\textfrak{S}^2}^2 ,
\end{align}
\begin{align}
\norm{\left[D , (1+ x^2) \right] f}_{L^2(\mathbb{R}^3)} &\leq C \norm{\left( 1 + x^2 \right) f}_{L^2(\mathbb{R}^3)}
\end{align}
and
\begin{align}
\norm{i \varepsilon \nabla \left( 1 + x^2 \right) f}_{L^2(\mathbb{R}^3)} 
&\leq \norm{\left( 1 + x^2 \right) f}_{L^2(\mathbb{R}^3)}
+ \norm{\left( 1 + x^2 \right) i \varepsilon \nabla f}_{L^2(\mathbb{R}^3)}
\end{align}
we have
\begin{align}
&\norm{ \left( 1 + x^2 \right) D^3 \left( \left[ K* \rho_{\widetilde{\omega}_{N,s}}, \widetilde{\omega}_{N,s} \right] - B_s \right) D }_{\textfrak{S}^{2}}
\nonumber \\
&\quad \leq 
C \norm{ \left( 1 + x^2 \right) D^2 \left( \left[ K* \rho_{\widetilde{\omega}_{N,s}}, \widetilde{\omega}_{N,s} \right] - B_s \right) }_{\textfrak{S}^{2}}
\nonumber \\
&\qquad + C
\norm{ \left( 1 + x^2 \right) i \varepsilon \nabla D^2 \left( \left[ K* \rho_{\widetilde{\omega}_{N,s}}, \widetilde{\omega}_{N,s} \right] - B_s \right) }_{\textfrak{S}^{2}}
\nonumber \\
&\qquad +
C \norm{ \left( 1 + x^2 \right) D^2 \left( \left[ K* \rho_{\widetilde{\omega}_{N,s}}, \widetilde{\omega}_{N,s} \right] - B_s \right) i \varepsilon \nabla }_{\textfrak{S}^{2}}
\nonumber \\
&\qquad + C
\norm{ \left( 1 + x^2 \right) i \varepsilon \nabla D^2 \left( \left[ K* \rho_{\widetilde{\omega}_{N,s}}, \widetilde{\omega}_{N,s} \right] - B_s \right) i \varepsilon \nabla}_{\textfrak{S}^{2}} .
\end{align}
If we use (see \cite[Chapter 3]{BPSS2016})
\begin{align}
\label{eq:second order Duhamel expansion of the potential}
&K * \rho_{\widetilde{\omega}_{N,s}}(x)  - K * \rho_{\widetilde{\omega}_{N,s}}(y)
- \nabla (K * \widetilde{\rho}_s) \left( \frac{x +y}{2} \right) \cdot (x-y)
\nonumber \\
&\quad =
\sum_{i,j= 1}^3 \int_0^1 \int_0^1  
\left( \partial_i \partial_j K * \rho_{\widetilde{\omega}_{N,s}} \right) \big( v_{\lambda,\mu}(x,y) \big) (x-y)_i (x-y)_j \left( \lambda - 1/2 \right)\,d\mu\,d\lambda  ,
\end{align}
where $v_{\lambda,\mu}(x,y) = \mu (\lambda x + (1 - \lambda) y) + (1 - \mu) (x+y)/2$, we get
\begin{align}
&\left( \left[ K* \rho_{\widetilde{\omega}_{N,s}}, \widetilde{\omega}_{N,s} \right] - B_s \right)(x; y)
\nonumber \\
&\quad =
\sum_{i,j= 1}^3 \int_0^1  \int_0^1   \left( \lambda - 1/2 \right)
\left( \partial_i \partial_j K * \rho_{\widetilde{\omega}_{N,s}} \right) \big( v_{\lambda,\mu}(x,y) \big) 
\left[ \hat{x}_i , \left[ \hat{x}_j , \widetilde{\omega}_{N,s} \right] \right](x;y)\,d\mu\,d\lambda .
\end{align}
Using $\norm{O}_{\textfrak{S}^2}^2 = \iint  \abs{O(x;y)}^2 dx\,dy$ for any Hilbert-Schmidt operator $O$, the product rule for the gradient and estimating the term involving the potential by the $L^{\infty}(\mathbb{R}^3)$-norm we get
\begin{align}
\eqref{eq: Sobolev norm estimate omega 6}
&\leq C N^{1/2} \varepsilon^{-1} \int_0^t 
\sum_{2 \leq \abs{\beta} \leq 6} \varepsilon^{\abs{\beta} - 2} \norm{\nabla^{\beta} K * \rho_{\widetilde{\omega}_{N,s}}}_{L^{\infty}(\mathbb{R}^3)}
\nonumber \\
&\qquad \times
\sum_{\abs{\gamma} \leq 3} \varepsilon^{\abs{\gamma}}
\left( 
\norm{ \left( 1 + x^2 \right) \nabla^{\gamma} \left[ \hat{x} , \left[ \hat{x} , \widetilde{\omega}_{N,s} \right] \right]}_{\textfrak{S}^{2}}
+
\norm{ \left( 1 + x^2 \right) \nabla^{\gamma} \left[ \hat{x} , \left[ \hat{x} , \widetilde{\omega}_{N,s} \right] \right] i \varepsilon \nabla}_{\textfrak{S}^{2}}
\right)\,ds .
\end{align}
Together with \eqref{eq:relation between density for Wigner transform and its Weyl quantization}, \eqref{eq:estimates for the vector potential and interaction 2} Lemma \ref{lemma:trace norm estimates with gradients and commutators} this leads to 
\begin{align}
\eqref{eq: Sobolev norm estimate omega 6}
&\leq C N \varepsilon^{-1} \int_0^t  
\sum_{2 \leq \abs{\beta} \leq 6} \varepsilon^{\abs{\beta} - 2} \norm{\nabla^{\beta} K * \widetilde{\rho}_{\widetilde{W}_{N,s}}}_{L^{\infty}(\mathbb{R}^3)}
\sum_{j=2}^{8} \varepsilon^j \norm{\widetilde{W}_{N,s}}_{H_{6}^{j}}ds 
\nonumber \\
&\leq C N \varepsilon \int_0^t 
\left( \left< \varepsilon \right> \norm{\widetilde{W}_{N,s}}_{L^1(\mathbb{R}^6)}
+ \sum_{i=0}^3 \varepsilon^{2 + i} \norm{\widetilde{W}_{N,s}}_{H_{4}^{j+1}} \right)
\sum_{j=0}^{6} \varepsilon^j \norm{\widetilde{W}_{N,s}}_{H_{6}^{j+2}}ds  .
\end{align}

\paragraph{The term \eqref{eq: Sobolev norm estimate omega 7}:}

In the following it will be convenient to split the operator $C_s$, defined in \eqref{eq:definition C-t}, into the following three parts
\begin{align}
C_{1,s}(x;y) &= \left(\nabla (\kappa * \vA_{\widetilde{\alpha}_s})^2 \right) \left(\frac{x+y}{2} \right) \cdot (x - y) \, \widetilde{\omega}_{N,s}(x;y) ,
\nonumber \\
C_{2,s}(x;y)
&= \sum_{j=1}^3 \left( \nabla^{(j)} \kappa * \vA_{\widetilde{\alpha}_s} \right) \left( \frac{x + y}{2} \right) \cdot (x - y)^{(j)} \left\{ i \varepsilon \nabla , \widetilde{\omega}_{N,s} \right\}(x;y) ,
\nonumber \\
C_{3,s}(x;y) &= 2 \kappa * \vA_{\widetilde{\alpha}_s} \left( \frac{x + y}{2} \right) \cdot \left[ i \varepsilon \nabla, \widetilde{\omega}_{N,s} \right](x;y) .
\end{align}
Moreover, note that
\begin{align}
\left[  2 \kappa * \vA_{\widetilde{\alpha}_s} \cdot i \varepsilon \nabla , \widetilde{\omega}_{N,s} \right]
&= \left[ \kappa * \vA_{\widetilde{\alpha}_s} , \left\{ i \varepsilon \nabla, \widetilde{\omega}_{N,s} \right\} \right]
+ \left\{ \kappa * \vA_{\widetilde{\alpha}_s} , \left[ i \varepsilon \nabla , \widetilde{\omega}_{N,s} \right]  \right\} 
\end{align}
leading to
\begin{subequations}
\begin{align}
\label{eq:Sobolev norm estimate omega 7a}
\eqref{eq: Sobolev norm estimate omega 7}
&\leq 
\varepsilon^{-1} \int_0^t 
\norm{\left[ (\kappa * \vA_{\widetilde{\alpha}_s})^2  , \widetilde{\omega}_{N,s} \right]  - C_{1,s} }_{\textfrak{S}_{\varepsilon}^{1,1}} ds
\\
\label{eq:Sobolev norm estimate omega 7b}
&\quad +
\varepsilon^{-1} \int_0^t 
\norm{\left[ \kappa * \vA_{\widetilde{\alpha}_s} , \left\{ i \varepsilon \nabla, \widetilde{\omega}_{N,s} \right\} \right]  - C_{2,s} }_{\textfrak{S}_{\varepsilon}^{1,1}} ds
\\
\label{eq:Sobolev norm estimate omega 7c}
&\quad +
\varepsilon^{-1} \int_0^t 
\norm{\left\{ \kappa * \vA_{\widetilde{\alpha}_s} , \left[ i \varepsilon \nabla , \widetilde{\omega}_{N,s} \right]  \right\}   - C_{3,s} }_{\textfrak{S}_{\varepsilon}^{1,1}} ds
\end{align}
\end{subequations}
The first term can be estimated in complete analogy to \eqref{eq: Sobolev norm estimate omega 6}. Replacing $K * \rho_{\widetilde{\omega}_{N,s}}$ by $(\kappa * \vA_{\widetilde{\alpha}_s})^2$ in \eqref{eq: Sobolev norm estimate omega 6}, performing exactly the same estimates and applying \eqref{eq:estimates for the vector potential and interaction 1}  lead to  
\begin{align}
\eqref{eq:Sobolev norm estimate omega 7a}
&\leq C N \varepsilon \int_0^t  
\sum_{2 \leq \abs{\beta} \leq 6} \varepsilon^{\abs{\beta} - 2} \norm{\nabla^{\beta} (\kappa * \vA_{\widetilde{\alpha}_s})^2}_{L^{\infty}(\mathbb{R}^3)}
\sum_{j=0}^{6} \varepsilon^j \norm{\widetilde{W}_{N,s}}_{H_{6}^{j+2}}ds
\nonumber \\
&\leq C N \varepsilon \int_0^t 
\sum_{i=0}^4 \varepsilon^i \norm{\widetilde{\alpha}_s}_{\mathfrak{h}_{i+1}}^2
\sum_{j=0}^{6} \varepsilon^j \norm{\widetilde{W}_{N,s}}_{H_{6}^{j+2}}ds
\end{align}
If we replace $K * \rho_{\widetilde{\omega}_{N,s}}$ by $\kappa * \vA_{\widetilde{\alpha}_s}$ in \eqref{eq:second order Duhamel expansion of the potential} we get
\begin{align}
&\left( \left[ \kappa * \vA_{\widetilde{\alpha}_s} , \left\{ i \varepsilon \nabla, \widetilde{\omega}_{N,s} \right\} \right]  - C_{2,s} \right)(x; y)
\nonumber \\
&\quad =
\sum_{i,j= 1}^3 \int_0^1  \int_0^1  \left( \lambda - \frac{1}{2} \right)
\left( \partial_i \partial_j \kappa * \vA_{\widetilde{\alpha}_s} \right) \big( v_{\lambda,\mu}(x,y) \big) 
\left[ \hat{x}_i , \left[ \hat{x}_j , \left\{ i \varepsilon \nabla , \widetilde{\omega}_{N,s} \right\} \right] \right](x;y)\,d\mu\,d\lambda,
\end{align}
with $v_{\lambda,\mu}(x,y) = \mu (\lambda x + (1 - \lambda) y) + (1 - \mu) (x+y)/2$.
By the same estimates as for \eqref{eq: Sobolev norm estimate omega 6}, Lemma \ref{lemma:trace norm estimates with gradients and commutators} and \eqref{eq:estimates for the vector potential and interaction 1} we conclude
\begin{align}
\eqref{eq:Sobolev norm estimate omega 7b}
&\leq C N^{1/2} \varepsilon^{-1} \int_0^t  
\sum_{2 \leq \abs{\beta} \leq 6} \varepsilon^{\abs{\beta} - 2} \norm{\nabla^{\beta} \kappa * \vA_{\widetilde{\alpha}_s}}_{L^{\infty}(\mathbb{R}^3)}
\nonumber \\
&\qquad \times
\sum_{\abs{\gamma} \leq 3} \varepsilon^{\abs{\gamma}}
\bigg( 
\norm{ \left( 1 + x^2 \right) \nabla^{\gamma} \left[ \hat{x} , \left[ \hat{x} , \left\{ i \varepsilon \nabla , \widetilde{\omega}_{N,s} \right\} \right] \right]}_{\textfrak{S}^{2}}
\nonumber \\
&\qquad \qquad \qquad \qquad
+
\norm{ \left( 1 + x^2 \right) \nabla^{\gamma} \left[ \hat{x} , \left[ \hat{x} , \left\{ i \varepsilon \nabla , \widetilde{\omega}_{N,s} \right\} \right] \right] i \varepsilon \nabla}_{\textfrak{S}^{2}}
\bigg) ds
\nonumber \\
&\leq C N \varepsilon \int_0^t  
\sum_{2 \leq \abs{\beta} \leq 6} \varepsilon^{\abs{\beta} - 2} \norm{\nabla^{\beta} \kappa * \vA_{\widetilde{\alpha}_s}}_{L^{\infty}(\mathbb{R}^3)}
\sum_{j=0}^{6} \varepsilon^j \norm{\widetilde{W}_{N,s}}_{H_{7}^{j+2}} ds 
\nonumber \\
&\leq C N \varepsilon \int_0^t  
\sum_{i=0}^4 \varepsilon^i \norm{\widetilde{\alpha}_s}_{\mathfrak{h}_{i+1}}
\sum_{j=0}^{6} \varepsilon^j \norm{\widetilde{W}_{N,s}}_{H_{7}^{j+2}}ds  .
\end{align}
Now let 
$u_{\lambda}(x;y) = \lambda x + (1 - \lambda) (x+y)/2$. Then,
\begin{align}
\kappa * \vA_{\widetilde{\alpha}_s}(x) - \kappa * \vA_{\widetilde{\alpha}_s} \left( \frac{x+y}{2} \right)
&= \frac{1}{2} \sum_{j=1}^3 (x-y)_j \int_0^1 \left( \partial_j \kappa * \vA_{\widetilde{\alpha}_s} \right) \big( u_{\lambda}(x;y) \big)d\lambda
\end{align}
and 
\begin{align}
&\left( \left\{ \kappa * \vA_{\widetilde{\alpha}_s} , \left[ i \varepsilon \nabla , \widetilde{\omega}_{N,s} \right]  \right\}   - C_{3,s}
\right)(x;y)
\nonumber \\
&\quad = \frac{1}{2} \sum_{j=1}^3 \int_0^1 
\left( \left( \partial_j \kappa * \vA_{\widetilde{\alpha}_s} \right) \big( u_{\lambda}(x;y) \big)
- \left( \partial_j \kappa * \vA_{\widetilde{\alpha}_s} \right) \big(u_{\lambda}(y;x) \big)
\right)
\left[ \hat{x}_j , \left[ i \varepsilon \nabla , \widetilde{\omega}_{N,s} \right] \right](x;y)\,d\lambda .
\end{align}
By the same arguments as before we get
\begin{align}
\eqref{eq:Sobolev norm estimate omega 7c}
&\leq C N^{1/2} \varepsilon^{-1} \int_0^t  
\sum_{1 \leq \abs{\beta} \leq 5} \varepsilon^{\abs{\beta} - 1} \norm{\nabla^{\beta} \kappa * \vA_{\widetilde{\alpha}_s}}_{L^{\infty}(\mathbb{R}^3)}
\nonumber \\
&\qquad \times
\sum_{\abs{\gamma} \leq 3} \varepsilon^{\abs{\gamma}}
\bigg( 
\norm{ \left( 1 + x^2 \right) \nabla^{\gamma} \left[ \hat{x} , \left[ i \varepsilon \nabla , \widetilde{\omega}_{N,s} \right] \right]}_{\textfrak{S}^{2}}
\nonumber \\
&\qquad \qquad \qquad \qquad
+
\norm{ \left( 1 + x^2 \right) \nabla^{\gamma} \left[ \hat{x} , \left[ i \varepsilon \nabla , \widetilde{\omega}_{N,s} \right] \right] i \varepsilon \nabla}_{\textfrak{S}^{2}}
\bigg) ds
\nonumber \\ 
&\leq C N \varepsilon \int_0^t 
\sum_{1 \leq \abs{\beta} \leq 5} \varepsilon^{\abs{\beta} - 1} \norm{\nabla^{\beta} \kappa * \vA_{\widetilde{\alpha}_s}}_{L^{\infty}(\mathbb{R}^3)}
\sum_{j=0}^{6} \varepsilon^j \norm{\widetilde{W}_{N,s}}_{H_{6}^{j+2}} ds
\nonumber \\ 
&\leq C N \varepsilon \int_0^t  
\sum_{i=0}^4 \varepsilon^i \norm{\widetilde{\alpha}_s}_{\mathfrak{h}_i}
\sum_{j=0}^{6} \varepsilon^j \norm{\widetilde{W}_{N,s}}_{H_{6}^{j+2}}ds .
\end{align}
In total this shows
\begin{align}
\eqref{eq: Sobolev norm estimate omega 7}
&\leq C N \varepsilon \int_0^t 
\sum_{i=0}^4 \varepsilon^i \norm{\widetilde{\alpha}_s}_{\mathfrak{h}_{i+1}} \left( 1 + \norm{\widetilde{\alpha}_s}_{\mathfrak{h}_{i+1}} \right)
\sum_{j=0}^{6} \varepsilon^j \norm{\widetilde{W}_{N,s}}_{H_{7}^{j+2}}ds .
\end{align}

\paragraph{The terms \eqref{eq: Sobolev norm estimate omega 8} and \eqref{eq: Sobolev norm estimate omega 9}:}

Using
\begin{align}
\left[ D_{\leq n}, h(s) \right]
&= 2 i \varepsilon \nabla \cdot \left[ D_{\leq n} , \kappa * \vA_{\alpha_s} \right]
+ \left[ D_{\leq n} , ( \kappa * \vA_{\alpha_s})^2 \right]
+ \left[ D_{\leq n}, K * \rho_{\omega_{N,s}} \right] ,
\end{align}
$\norm{D^{-1} i \varepsilon \nabla }_{\textfrak{S}^{\infty}}
= \norm{\left( 1 - \varepsilon^2 \Delta \right)^{-1/2} i \varepsilon \nabla }_{\textfrak{S}^{\infty}} \leq 1$ and $\norm{D^{-1}}_{\textfrak{S}^{\infty}} \leq 1$ 
we get the bound
\begin{align}
\eqref{eq: Sobolev norm estimate omega 8}
+ \eqref{eq: Sobolev norm estimate omega 9}
&\leq C \varepsilon^{-1} 
\lim_{n \rightarrow \infty}
\int_0^t 
\norm{\omega_{N,s} - \widetilde{\omega}_{N,s}}_{\textfrak{S}_{\varepsilon}^{1,1}}
\Big(  \norm{\left[ D_{\leq n} , \kappa * \vA_{\alpha_s} \right]}_{\textfrak{S}^{\infty}}
\nonumber \\
&\qquad  \qquad  
+ \norm{\left[ D_{\leq n} , (\kappa * \vA_{\alpha_s})^2 \right]}_{\textfrak{S}^{\infty}}
+ \norm{\left[ D_{\leq n} , K * \rho_{\omega_{N,s}} \right]}_{\textfrak{S}^{\infty}}
\Big)\,ds.
\end{align}
Thus if we use  \eqref{estimate:commutator regularized fractional laplacian with potential}, 
$\mathcal{F}[K *  \rho_{\omega_{N,s}}](k)= \mathcal{F}[K](k)
\mathcal{F}[\rho_{\omega_{N,s} }](k)$ and the
Fourier expansion of the regularized vector potential \eqref{eq:Fourier expansion regularized vector potential} we obtain
\begin{align}
\eqref{eq: Sobolev norm estimate omega 8}
+ \eqref{eq: Sobolev norm estimate omega 9}
&\leq C  \int_0^t 
\norm{\omega_{N,s} - \widetilde{\omega}_{N,s}}_{\textfrak{S}_{\varepsilon}^{1,1}}
\nonumber \\
&\qquad \times
\bigg[ \left( 1 +  \norm{\kappa * \vA_{\alpha_s}}_{L^{\infty}(\mathbb{R}^3)}
\right)  \sum_{\lambda=1,2} \int  \abs{k}^{1/2} \abs{\mathcal{F}[\kappa](k)} \abs{\alpha_s(k,\lambda)} dk
\nonumber \\
&\qquad \quad
+\int  \abs{k} \abs{\mathcal{F}[K](k)} \abs{\mathcal{F}[\rho_{\omega_{N,s} }](k)}
dk \bigg]\,ds.
\end{align}
By means of \eqref{eq:vector potential L-infty estimate}, \eqref{eq:Fourier transform of the potential}, \eqref{eq: estimate to bound the density by the L1-norm of omega}, and \eqref{eq:assumption on the cutoff parameter} the Cauchy--Schwarz inequality the right hand side can be bounded by
\begin{align}
\eqref{eq: Sobolev norm estimate omega 8}
+ \eqref{eq: Sobolev norm estimate omega 9}
&\leq C  \int_0^t 
\norm{\omega_{N,s} - \widetilde{\omega}_{N,s}}_{\textfrak{S}_{\varepsilon}^{1,1}}
\Big[ \left( 1 + \norm{\alpha_s}_{\dot{\mathfrak{h}}_{1/2}} \right) \norm{\alpha_s}_{\dot{\mathfrak{h}}_{1/2}} + N^{-1} \norm{\omega_{N,s}}_{\textfrak{S}^1} \Big]ds .
\end{align}
Collecting the estimates and combining them with \eqref{eq:proof main theorem estimate for the difference of the fields} lead to
\begin{align}
&N^{-1} \norm{\omega_{N,t} - \widetilde{\omega}_{N,t}}_{\textfrak{S}_{\varepsilon}^{1,1}}
+  \norm{\alpha_t - \widetilde{\alpha}_t}_{\dot{\mathfrak{h}}_{-1/2} \, \cap \, \dot{\mathfrak{h}}_{1/2}}
\nonumber \\
&\quad\leq C \int_0^t  C(s) \left( N^{-1} \norm{\omega_{N,s} - \widetilde{\omega}_{N,s}}_{\textfrak{S}_{\varepsilon}^{1,1}}
+ \norm{\alpha_s - \widetilde{\alpha}_s}_{\dot{\mathfrak{h}}_{-1/2} \, \cap \, \dot{\mathfrak{h}}_{1/2}} \right)ds
\nonumber \\
&\qquad + N^{-1}  \norm{\omega_{N,0} - \widetilde{\omega}_{N,0}}_{\textfrak{S}_{\varepsilon}^{1,1}} 
+ \norm{\alpha - \widetilde{\alpha}}_{\dot{\mathfrak{h}}_{-1/2} \, \cap \, \dot{\mathfrak{h}}_{1/2}}
+  C \varepsilon \int_0^t  \widetilde{C}(s)ds ,
\end{align}
with
\begin{align}
C(s) &=  \bigg( 1 + N^{-1} \norm{\omega_{N,s}}_{\textfrak{S}^1} 
+ \sum_{j=0}^{7} \varepsilon^j \norm{\widetilde{W}_{N,s}}_{H_{7}^{j+1}} \bigg) 
\left( 1 + \norm{\alpha_s}_{\dot{\mathfrak{h}}_{1/2}} +\norm{\widetilde{\alpha}_s}_{\dot{\mathfrak{h}}_{1/2}} \right)^2
\end{align}
and
\begin{align}
\widetilde{C}(s) &=  
N^{-2/3} \varepsilon^{-2} \left( 1 + N^{-1} \mathcal{E}^{\rm{MS}}[\omega_0, \alpha_0]  \right) +
\sum_{j=0}^{6} \varepsilon^j \norm{\widetilde{W}_{N,s}}_{H_{7}^{j+2}}
\bigg( \left< \varepsilon \right> \norm{\widetilde{W}_{N,s}}_{L^1(\mathbb{R}^6)}
\nonumber \\
&\quad + \sum_{k=0}^3 \varepsilon^{2 + k} \norm{\widetilde{W}_{N,s}}_{H_{4}^{k+1}} 
+ \sum_{k=0}^4 \varepsilon^k \left( 1 + \norm{\widetilde{\alpha}_s}_{\mathfrak{h}_{k+1}} \right)^2
\bigg) .
\end{align}
If we use $\varepsilon = N^{-1/3},$ $\norm{\omega_{N,0}}_{\textfrak{S}^1} = N$, \eqref{eq:estimate for the MS energy functional 1}, \eqref{eq:estimate for the VM energy functional 1} and the conservation of mass and energy both of the Maxwell--Schr\"odinger and the Vlasov--Maxwell system, we can bound the above quantities by
\begin{align}
C(s) &\leq  C \bigg( 1 + \sum_{j=0}^{7} \varepsilon^j \norm{\widetilde{W}_{N,s}}_{H_{7}^{j+1}} \bigg) 
\left( 1 + N^{-1} \mathcal{E}^{\rm{MS}}[\omega_{N,0}, \alpha_0] + \mathcal{E}^{\rm{VM}}[\widetilde{W}_{N,0}, \widetilde{\alpha}_0] + 
\norm{\widetilde{W}_{N,0}}_{L^1(\mathbb{R}^6)} \right)^2
\end{align}
and
\begin{align}
\widetilde{C}(s) &\leq 
1 + N^{-1} \mathcal{E}^{\rm{MS}}[\omega_0, \alpha_0]  +
\sum_{j=0}^{6} \varepsilon^j \norm{\widetilde{W}_{N,s}}_{H_{7}^{j+2}}
\bigg( \left< \varepsilon \right> \norm{\widetilde{W}_{N,0}}_{L^1(\mathbb{R}^6)}
\nonumber \\
&\quad + \sum_{k=0}^3 \varepsilon^{2 + k} \norm{\widetilde{W}_{N,s}}_{H_{4}^{k+1}} 
+ \sum_{k=0}^4 \varepsilon^k \left( 1 + \norm{\widetilde{\alpha}_s}_{\mathfrak{h}_{k+1}} \right)^2
\bigg) .
\end{align}

\section{Well-posedness of the Maxwell--Schr\"odinger equations}
\label{section:Well-posedness of the Maxwell--Schroedinger equations}

In this section we prove Proposition \ref{proposition:well-posedness of Maxwell-Schroedinger}. To this end, we recall propagation estimates for the time evolution of the magnetic Laplacian from \cite{NM2005}. We then consider a linearized version of \eqref{eq:Maxwell-Schroedinger equations}  and show the existence as well as regularity properties of the respective solutions. These will be used to prove the existence of local solutions to \eqref{eq:Maxwell-Schroedinger equations}.
Finally, we provide propagation estimates allowing us to conclude that under suitable initial conditions the solutions of the Maxwell--Schr\"odinger equations exist globally. The strategy of the proof is inspired by \cite{NM2005} and \cite{P2014,ADM2017,AMS2022}.

\begin{lemma}
\label{lemma:evolution operator for magnetic laplacian}
Let $\alpha \in L^{\infty}(0,T ; \mathfrak{h}_{1/2} \cap \dot{\mathfrak{h}}_{-1/2}) \cap W^{1,1}(0,T; \dot{\mathfrak{h}}_{-1/2} )$. The equation
\begin{align}
\label{eq:mangetic Schroediger equation}
\begin{cases}
i \varepsilon \partial_t \psi(t) &= \left( - i \varepsilon \nabla - \kappa * \vA_{\alpha} \right)^2 \psi(t) \\
\psi(0) = \psi_0
\end{cases}
\end{align}
has a unique $C(0,T; H^2(\mathbb{R}^3)) \cap C^1(0,T; L^2(\mathbb{R}^3))$ solution. There exists a constant $K_{\varepsilon}$ (depending on $\varepsilon$) and a strongly continuous two-parameter family $U_{\alpha}$ of operators on $H^s(\mathbb{R}^3)$, $0 \leq s \leq 2$, with 
\begin{align}
\label{eq:evolution operator for magnetic laplacian estimate}
K_{\varepsilon, \alpha,T} &= \sup_{t,\tau \in [0,T]} \norm{U_{\alpha}(t,\tau)}_{\textfrak{S}^{\infty}(H^2(\mathbb{R}^3))}
\leq K_{\varepsilon} \left( 1 + \norm{\alpha}_{L_T^{\infty} \mathfrak{h}_{1/2} \cap L_T^{\infty} \dot{\mathfrak{h}}_{-1/2}}^4 \right) e^{K_{\varepsilon} \norm{ \partial_t  \alpha}_{L_T^1 \dot{\mathfrak{h}}_{-1/2}}} 
\end{align}
and
$\sup_{t,\tau \in [0,T]} \norm{U_{\alpha}(t,\tau)}_{\textfrak{S}^{\infty}(H^s(\mathbb{R}^3))} \leq K_{\varepsilon, \alpha, T}^{s/2}$ .
In particular, $U_\alpha(t,\tau)$ is a unitary group on $L^2(\mathbb{R}^3)$ and $U_\alpha^*(t, \tau) = U_\alpha(\tau, t)$.
Moreover, for any $\psi_0 \in H^s(\mathbb{R}^3)$, $U_\alpha(t,t_0) \psi_0$ is a unique $H^s$-solution to \eqref{eq:mangetic Schroediger equation}.
\end{lemma}

\begin{proof}
Using \eqref{eq:vector potential L-2 estimate} we estimate
\begin{align}
\norm{\kappa * \vA_{\alpha}}_{L_T^{\infty} H^1(\mathbb{R}^3)} 
&= \norm{\kappa * \vA_{\left< \cdot \right> \alpha}}_{L_T^{\infty} L^2(\mathbb{R}^3)}
\leq \norm{\alpha}_{L_T^{\infty} \mathfrak{h}_{1/2} \cap L_T^{\infty} \dot{\mathfrak{h}}_{-1/2}} 
\end{align}
By interpolation we get  $\norm{\kappa * \vA_{\alpha_t}}_{L^3(\mathbb{R}^3)} \leq \norm{\alpha_t}_{\dot{\mathfrak{h}}_{-1/2}}$ from \eqref{eq:vector potential L-infty estimate} and \eqref{eq:vector potential L-2 estimate}, leading to
\begin{align}
\norm{ \kappa * \vA_{\alpha}}_{L^1_T L^3(\mathbb{R}^3)} \leq C \norm{ \alpha}_{L_T^1 \dot{\mathfrak{h}}_{-1/2}}
\quad \text{and} \quad 
\norm{\partial_t \kappa * \vA_{\alpha}}_{L^1_T L^3(\mathbb{R}^3)} \leq C \norm{\partial_t  \alpha}_{L_T^1 \dot{\mathfrak{h}}_{-1/2}} .
\end{align}
Since the vector potential is divergence free by construction, this allows us to conclude that Assumption (A1) on \cite[p. 574]{NM2005} is satisfied for $\kappa * \vA_{\alpha}$ and $u =0$. Application of \cite[Lemma 3.1 and Lemma 3.2]{NM2005} then shows the claim.
\end{proof}

\subsection{Linear equations}

\begin{lemma}
\label{lemma:MS well-posedness linearized Schroedinger equation}
Let $T >0$,  $(\omega, \alpha) \in \left(  C \left( 0,T; \textfrak{S}^{1} \left( L^2(\mathbb{R}^3) \right)  \right)  \cap L^{\infty}(0,T; \textfrak{S}^{2,1}(L^2(\mathbb{R}^3)) \big) \right)  \times \break \left( L^{\infty}(0,T ; \mathfrak{h}_{1/2} \cap \dot{\mathfrak{h}}_{-1/2}) \cap W^{1,1}(0,T; \dot{\mathfrak{h}}_{-1/2} ) \right)$ and $\Gamma_0 \in \textfrak{S}^{2,1}\left( L^2(\mathbb{R}^3) \right)$. The linear equation
\begin{align}
\label{eq: Well-posedness MS linearized Schroedinger equation}
i \varepsilon \partial_t \Gamma_t &= \left[ \left( - i \varepsilon \nabla - \kappa * \vA_{\alpha_t} \right)^2   + K * \rho_{\omega_t} - X_{\omega_t} , \Gamma_t \right]
\end{align}
with initial datum $\Gamma_0$ has a unique $C^1 \left( 0,T; \textfrak{S}^{1} \left( L^2(\mathbb{R}^3) \right) \right) \cap C \left( 0,T; \textfrak{S}^{2,1} \left( L^2(\mathbb{R}^3) \right) \right)$ solution. 
\end{lemma}

\begin{proof}
Let $(\mathcal{Z}_T, d_T)$ be the Banach space $\mathcal{Z}_T = \left\{ \Gamma \in C \left( 0,T; \textfrak{S}^{2,1} \left( L^2(\mathbb{R}^3) \right) \right) : \Gamma(0) = \Gamma_0 \right\}$ with metric $d_T(\Gamma, \Gamma') = \sup_{t \in [0,T]} \norm{\Gamma(t) - \Gamma'(t)}_{\textfrak{S}^{2,1}}$. We define the mapping $\Phi: \mathcal{Z}_T \rightarrow \mathcal{Z}_T$ by
\begin{align}
\label{eq:linear Schroedinger equation fixed point map}
\Phi(\Gamma)(t) 
&= U_{\alpha}(t,0) \bigg( \Gamma_0
- i \varepsilon^{-1} \int_0^t U_{\alpha}(0,s)
\left[ K * \rho_{\omega_s} - X_{\omega_s} , \Gamma_s \right] U_{\alpha}(s,0)\,ds \bigg) U_{\alpha}(0,t),
\end{align}
where $U_{\alpha}$ is defined as in Lemma \ref{lemma:evolution operator for magnetic laplacian}.
By means of \eqref{eq:W-0-2-infty estimate for the direct interaction potential}, \eqref{eq:estimate commutator with exchange term in L-s,1 norm} and the properties of $U_{\alpha}$ it is straightforward to see that the integrand on the right hand side defines a strongly continuous function w.r.t. $s$ on $\textfrak{S}^{1}$ and a summable function w.r.t. $s$ on $\textfrak{S}^{2,1}$, enabling us to define its integral on $\textfrak{S}^{1}$ as a Riemann integral and on $\textfrak{S}^{2,1}$ as a Bochner integral. Together with the strong differentiability of $U_{\alpha}$ this proves that the right hand side of \eqref{eq:linear Schroedinger equation fixed point map} is an element of $C^1 \left( 0,T; \textfrak{S}^{1} \left( L^2(\mathbb{R}^3) \right) \right) \cap C \left( 0,T; \textfrak{S}^{2,1} \left( L^2(\mathbb{R}^3) \right) \right)$, ensuring that the mapping $\Phi$ is well defined. By means of \eqref{eq:W-0-2-infty estimate for the direct interaction potential},  \eqref{eq:estimate commutator with exchange term in L-s,1 norm} and Lemma \ref{lemma:evolution operator for magnetic laplacian}  we obtain
\begin{align}
\label{eq:MS well-posedness linearized equation contraction argument}
d_T \left( \Phi(\Gamma),  \Phi(\Gamma') \right) 
&\leq
\varepsilon^{-1} \sup_{t \in [0,T]} \left\{  \int_0^t  \norm{U_{\alpha}(t,s)
\left[ K * \rho_{\omega_s} - X_{\omega_s}  , \Gamma_s - \Gamma_s' \right] U_{\alpha}(s,t) }_{\textfrak{S}^{2,1}}ds \right\}
\nonumber \\
&\leq  d_T (\Gamma, \Gamma') T  C_{\varepsilon}  N^{-1}  K_{\varepsilon, \alpha, T} \norm{\omega}_{L_T^{\infty} \textfrak{S}^{2,1}} 
\end{align}
for  $\Gamma, \Gamma' \in \mathcal{Z}_T$ where $C_{\varepsilon}$ is a generic constant depending on $\varepsilon$. Thus if we choose $T^* \in (0,T]$ sufficiently small we obtain that  $\Phi$ is a contraction on $(\mathcal{Z}_{T^*}, d_{T^*})$. This proves the existence of a unique fixed point $\Gamma \in \mathcal{Z}_{T^*}$ satisfying $d_{T^*}(\Phi(\Gamma), \Gamma) = 0$. By differentiation of \eqref{eq:linear Schroedinger equation fixed point map} we obtain the existence of a unique solution to \eqref{eq: Well-posedness MS linearized Schroedinger equation} until time $T^*$. The solution can be extended to the interval $[0,T]$ because
$T^*$ can be chosen independently of $\Gamma$. Its regularity can be inferred from the regularity of the right hand side of \eqref{eq:linear Schroedinger equation fixed point map} as discussed above.
\end{proof}

\begin{lemma}
\label{lemma:MS well-posedness linearized field equation}
Let $T >0$, $(\omega, \alpha) \in  C \left( 0,T; \textfrak{S}^{1,1} \left( L^2(\mathbb{R}^3 \right) \right) \times C(0,T; \mathfrak{h}_{1/2} \cap \dot{\mathfrak{h}}_{-1/2})$, $\xi_0 \in \mathfrak{h}_{1/2} \cap \dot{\mathfrak{h}}_{-1/2}$ and $\vJ_{\omega_t,  \alpha_t}$ be defined as in \eqref{eq:definition charge current}. The linear equation
\begin{align}
\label{eq:MS well-posedness linearized field equation}
i \partial_t \xi_t(k,\lambda) &=  \abs{k}  \xi_t(k,\lambda) -  \sqrt{ \frac{4 \pi^3}{\abs{k}}}  \mathcal{F}[\kappa](k) \vep_{\lambda}(k) \cdot \mathcal{F}[\vJ_{\omega_t, \alpha_t}](k)  
\end{align}
with initial datum $\xi_0$ has a unique $ C \big( 0,T; \mathfrak{h}_{1/2} \cap \dot{\mathfrak{h}}_{-1/2} \big) \cap C^1 \big( 0,T;  \dot{\mathfrak{h}}_{-1/2} \big)$ solution.
\end{lemma}

\begin{proof}
We consider
\begin{align}
\label{eq:MS wellposedness linearized field equation expanded solution}
\xi_t(k,\lambda)
&= e^{- i \abs{k} t} \xi_0 (k,\lambda)
+ i e^{-i \abs{k} t} \int_0^t  e^{i \abs{k} s} \sqrt{ \frac{4 \pi^3}{\abs{k}}} \mathcal{F}[\kappa](k) \vep_{\lambda}(k) \mathcal{F}[\vJ_{\omega_s, \alpha_s}](k)\,ds .
\end{align}
Note that the integrand on the right hand side is a $ C \big( 0,T; \mathfrak{h}_{1/2} \cap \dot{\mathfrak{h}}_{-1/2} \big)$ function because of \eqref{eq:assumption on the cutoff parameter}, \eqref{eq:Charge current L-infty estimate of the difference} and the regularity properties of $(\omega,\alpha)$. If we define the integral in the Riemann sense, we obtain the strong differentiability of $\xi_t$ with respect to time in $\dot{\mathfrak{h}}_{-1/2}$. Because \eqref{eq:MS wellposedness linearized field equation expanded solution} satisfies \eqref{eq:MS well-posedness linearized field equation} this shows the existence of at least one solution of \eqref{eq:MS well-posedness linearized field equation}.
In order to prove uniqueness let $\xi$ and $\xi'$ be two solutions of \eqref{eq:MS well-posedness linearized field equation} with initial datum $\xi_0$. Since
$i \partial_t e^{i \abs{k} t} \left( \xi_t - \xi_t' \right) = 0$
we get $\xi = \xi'$ by Duhamel's formula, showing the claim.
\end{proof}

\subsection{Local solutions}

\begin{lemma}
\label{lemma:MS system local well-posedness result stated by use of negative weighted Lp-spaces}
For all $(\omega_0, \alpha_0) \in \textfrak{S}^{2,1} (L^2(\mathbb{R}^3)) \times  \mathfrak{h}_{1/2} \cap \dot{\mathfrak{h}}_{-1/2}$ there exists $T >0$ and a unique
$C \big( 0,T; \textfrak{S}^{2,1} \big( L^2(\mathbb{R}^3) \big) \big) \cap C^1 \big( 0,T; \textfrak{S}^{1} \big( L^2(\mathbb{R}^3) \big) \big)   \times  C \big( 0,T; \mathfrak{h}_{1/2} \cap \dot{\mathfrak{h}}_{-1/2} \big) \cap C^1 \big( 0,T; \dot{\mathfrak{h}}_{-1/2} \big) $-valued function which satisfies \eqref{eq:Maxwell-Schroedinger equations}  in $\textfrak{S}^1(L^2(\mathbb{R}^3)) \oplus  \dot{\mathfrak{h}}_{-1/2}$ with initial datum $(\omega_0, \alpha_0)$.
\end{lemma}

\begin{proof}
The local existence of the solution will be shown by the means of the contraction mapping principle. To this end, we define for $(\omega_0, \alpha) \in \textfrak{S}^{2,1}(L^2(\mathbb{R}^3)) \times \mathfrak{h}_{1/2} \cap \dot{\mathfrak{h}}_{-1/2}$ and $T, R_1, R_2 > 0$
satisfying $R_2 \geq 1+ 2 \norm{\alpha_0}_{\mathfrak{h}_{1/2} \, \cap \, \dot{\mathfrak{h}}_{-1/2}}$ and $R_1 \geq 1 + 20 K_{\varepsilon}^2  R_2^{8} \norm{\omega_0}_{\textfrak{S}^{2,1}}$ (with $K_\varepsilon$ being the constant of Lemma \ref{lemma:evolution operator for magnetic laplacian}) the $(T,R_1,R_2)$-dependent space
\begin{align}
\label{eq:local solutions MS equations definition Banach space 1}
\mathcal{Z}_{T, R_1, R_2} 
&= \Big\{
(\omega, \alpha) \in \big( C(0,T; \textfrak{S}^{1,1}(L^2(\mathbb{R}^3)) \cap L^{\infty}(0,T; \textfrak{S}^{2,1}(L^2(\mathbb{R}^3)) \big) 
\nonumber \\
&\qquad
\times \big( C(0,T; \mathfrak{h}_{1/2}  \cap  \dot{\mathfrak{h}}_{-1/2} ) \cap W^{1,2}(0,T; \dot{\mathfrak{h}}_{-1/2} ) \big)  : 
(\omega(t),\alpha(t)) \big|_{t= 0} = (\omega_0, \alpha) ,
\nonumber \\
&\qquad
\norm{\omega}_{L_T^{\infty}\textfrak{S}^{2,1}} \leq R_1 ,
\max \{ \norm{\alpha}_{L_T^{\infty} \mathfrak{h}_{1/2} \, \cap \, L_T^{\infty} \dot{\mathfrak{h}}_{-1/2}}, \norm{\partial_t \alpha}_{L^2(0,T; \dot{\mathfrak{h}}_{-1/2} )} \} \leq R_2 
\Big\} .
\end{align}
Equipping $\mathcal{Z}_{T, R_1, R_2}$ with the metric
\begin{align}
d_T \big( (\omega, \alpha) , (\omega', \alpha')  \big)
&= \max \left\{
\norm{\omega - \omega'}_{L_T^{\infty}\textfrak{S}^{1,1}} ,
\norm{\alpha - \alpha'}_{L_T^{\infty} \mathfrak{h}_{1/2} \, \cap \, L_T^{\infty} \dot{\mathfrak{h}}_{-1/2}} \right\},
\end{align}
where $(\omega, \alpha)$, $(\omega', \alpha) \in \mathcal{Z}_{T, R_1, R_2}$, leads to the Banach space $(\mathcal{Z}_{T, R_1, R_2}, d_T)$.

Next, we consider the solutions of the linearized equations from Lemma~\ref{lemma:MS well-posedness linearized Schroedinger equation} and Lemma~\ref{lemma:MS well-posedness linearized field equation}, satisfying
\begin{align}
\label{eq: MS well-posedness linearized MS equations}
\begin{cases}
i \varepsilon \partial_t \Gamma_t &= \left[ \left( - i \varepsilon \nabla - \kappa * \vA_{\alpha_t} \right)^2   + K * \rho_{\omega_t} - X_{\omega_t} , \Gamma_t \right]
\\
i \partial_t \xi_t(k,\lambda) &=  \abs{k}  \xi_t(k,\lambda) -  \sqrt{ \frac{4 \pi^3}{\abs{k}}}  \mathcal{F}[\kappa](k) \vep_{\lambda}(k) \cdot \mathcal{F}[\vJ_{\omega_t, \alpha_t}](k)  
\end{cases} ,
\end{align}
with initial datum $(\omega_0, \alpha_0)$, and define the mapping $\Phi: \mathcal{Z}_{T, R_1, R_2} \rightarrow C^1 \left( 0,T;\textfrak{S}^{1} \left( (L^2(\mathbb{R}^3) \right) \right) \cap C \left( 0,T; \textfrak{S}^{2,1} \left( L^2(\mathbb{R}^3) \right) \right) \times C \big( 0,T; \mathfrak{h}_{1/2} \cap \dot{\mathfrak{h}}_{-1/2} \big) \cap C^1 \big( 0,T;  \dot{\mathfrak{h}}_{-1/2} \big)$ by
\begin{align}
\label{eq:contraction map}
\Phi \begin{pmatrix}
\omega \\ \alpha
\end{pmatrix}
&= 
\begin{pmatrix}
\Gamma \\ \xi
\end{pmatrix} .
\end{align}
With this regard note that the assumptions on $(\omega, \alpha)$ from Lemma~\ref{lemma:MS well-posedness linearized Schroedinger equation} and Lemma \ref{lemma:MS well-posedness linearized field equation} are satisfied because $W^{1,2}(0,T; \dot{\mathfrak{h}}_{-1/2} ) \subset W^{1,1}(0,T; \dot{\mathfrak{h}}_{-1/2} )$. Below, we will show that there exists a $\varepsilon$-dependent constant $\widetilde{C}_{\varepsilon} \geq 0$  such that for all times $T \leq \frac{1}{\widetilde{C}_{\varepsilon} R_1^2 R_2^8}$ the mapping $\Phi: \mathcal{Z}_{T, R_1, R_2}  \rightarrow \mathcal{Z}_{T, R_1, R_2}$ is well defined and is a contraction  on $(\mathcal{Z}_{T, R_1, R_2} , d_{T})$. By the Banach fixed-point theorem we obtain the existence of a unique fixed point $(\omega, \alpha) \in \mathcal{Z}_{T, R_1, R_2}$ such that $d \big( (\omega, \alpha) , \Phi (\omega, \alpha) \big) = 0$. Together with the regularity properties of the linearized equations from Lemma~\ref{lemma:MS well-posedness linearized Schroedinger equation} and Lemma~\ref{lemma:MS well-posedness linearized field equation} this proves Lemma \ref{lemma:MS system local well-posedness result stated by use of negative weighted Lp-spaces}.

\paragraph{Well-definedness of $\Phi$:}
Next, we prove that, for $T \leq \frac{1}{\widetilde{C}_{\varepsilon} R_1^2 R_2^8}$ and $\widetilde{C}_{\varepsilon} >0$ chosen sufficiently large, $\Phi$ maps $\mathcal{Z}_{T, R_1, R_2}$ into itself. To this end we use the expansion \eqref{eq:MS wellposedness linearized field equation expanded solution} with $\xi_0 = \alpha_0$, \eqref{eq:assumption on the cutoff parameter} and \eqref{eq:Charge current L-infty estimate of the difference} to estimate
\begin{align}
\norm{\xi_t}_{\mathfrak{h}_{1/2} \, \cap \, \dot{\mathfrak{h}}_{-1/2}}
&\leq  \norm{\alpha_0}_{\mathfrak{h}_{1/2} \, \cap \, \dot{\mathfrak{h}}_{-1/2}}
+ C \int_0^t \norm{\left< \cdot \right> \abs{\cdot}^{-1}  \mathcal{F}[\kappa]}_{L^2(\mathbb{R}^3)}
\norm{\mathcal{F}[\vJ_{\omega_s, \alpha_s}}_{L^{\infty}(\mathbb{R}^3, \mathbb{C}^3)}ds
\nonumber \\
&\leq  \norm{\alpha_0}_{\mathfrak{h}_{1/2} \, \cap \, \dot{\mathfrak{h}}_{-1/2}}
+ C_{\varepsilon} \int_0^t \left( 1+ \norm{\alpha_s}_{\dot{\mathfrak{h}}_{-1/2}} \right)  \norm{\omega_s}_{\textfrak{S}^{1,1}} ds
\nonumber \\
&\leq  \norm{\alpha_0}_{\mathfrak{h}_{1/2} \, \cap \, \dot{\mathfrak{h}}_{-1/2}}
+ C_{\varepsilon} \left( 1 + \norm{\alpha}_{L_T^{\infty} \dot{\mathfrak{h}}_{-1/2}} \right)  \norm{\omega}_{L_T^{\infty} \textfrak{S}^{1,1}}  t .
\end{align}
Using $R_2 \geq 1+ 2 \norm{\alpha_0}_{\mathfrak{h}_{1/2} \, \cap \, \dot{\mathfrak{h}}_{-1/2}}$ and $(\omega, \alpha) \in \mathcal{Z}_{T, R_1, R_2}$ we obtain
\begin{align}
\label{eq:well-definedness of contraction mapping in fixed-point argument estimate for xi}
\norm{\xi_t}_{ L_T^{\infty} \mathfrak{h}_{1/2} \, \cap \, L_T^{\infty} \dot{\mathfrak{h}}_{-1/2}}
&\leq  \frac{R_2}{2}
+ C_{\varepsilon}  R_1  R_2  T
\leq R_2
\end{align}
if $\widetilde{C}_{\varepsilon}$ in the definition of $T$ is chosen large enough.
By similar means we get
\begin{align}
\norm{\partial_t \xi_t}_{L_T^2 \dot{\mathfrak{h}}_{-1/2}}
&\leq T^{1/2} \norm{\partial_t \xi_t}_{L_T^{\infty} \dot{\mathfrak{h}}_{-1/2}}
\nonumber \\
&\leq  T^{1/2} \left(
\norm{\xi_t}_{L_T^{\infty} \mathfrak{h}_{1/2}}
+ C \norm{\abs{\cdot}^{-1} \mathcal{F}[\kappa]}_{L^2(\mathbb{R}^3)}
\norm{\mathcal{F}[\vJ_{\omega_t, \alpha_t}]}_{L_T^{\infty} L^{\infty}(\mathbb{R}^3, \mathbb{C}^3)} \right)
\nonumber \\
&\leq   T^{1/2} \left(
\norm{\xi_t}_{L_T^{\infty} \mathfrak{h}_{1/2}}
+ C_{\varepsilon}  R_1 R_2 \right) 
\nonumber \\
&\leq R_2 .
\end{align}
To obtain the ultimate inequality we have used \eqref{eq:well-definedness of contraction mapping in fixed-point argument estimate for xi}. 
Using the Duhamel expansion
\begin{align}
\Gamma_t &= U_{\alpha}(t;0) \omega_0 U^*_{\alpha}(t;0)
- i \varepsilon^{-1} \int_0^t  U_{\alpha}(t;s)
\left[ K * \rho_{\omega_s} - X_{\omega_s}  , \Gamma_s \right] U_{\alpha}^*(t;s)ds
\end{align}
with $U_{\alpha}$ being defined as in Lemma~\ref{lemma:evolution operator for magnetic laplacian}, \eqref{eq:evolution operator for magnetic laplacian estimate} and Lemma~\ref{lemma:general estimates for the interaction} we estimate 
\begin{align}
\norm{\Gamma}_{L_t^{\infty} \textfrak{S}^{2,1}}
&\leq K_{ \varepsilon, \alpha, t}^2
\left( \norm{\omega_0}_{\textfrak{S}^{2,1}}
 + C_{\varepsilon}
 \int_0^t  
\norm{\omega_s}_{\textfrak{S}^{2,1}}  \norm{\Gamma_s}_{\textfrak{S}^{2,1}} ds\right)
\nonumber \\
&\leq  K_{\varepsilon, \alpha, T}^2 
\left( \norm{\omega_0}_{\textfrak{S}^{2,1}}
+ C_{\varepsilon}  R_1 \int_0^t   \norm{\Gamma}_{L_s^{\infty} \textfrak{S}^{2,1}}ds \right) 
\quad \text{for all} \; t \in [0,T] .
\end{align}
Gr\"{o}nwall's lemma then leads to
\begin{align}
\norm{\Gamma}_{L_T^{\infty} \textfrak{S}^{2,1}} &\leq K_{\varepsilon, \alpha, T}^2 \norm{\omega_0}_{\textfrak{S}^{2,1}}
\left( 1 + e^{C_{\varepsilon} R_1 K_{\alpha, \varepsilon, T}^2 T } \right) .
\end{align}
Note that
\begin{align}
\label{eq:estimate for the constant of the magnetic laplacian}
K_{\varepsilon, \alpha, T} &\leq 2 K_{\varepsilon} R_2^4 e^{K_{\varepsilon} R_2 T^{1/2} } \leq 3  K_{\varepsilon} R_2^4
\end{align}
and 
\begin{align}
\norm{\Gamma}_{L_T^{\infty} \textfrak{S}^{2,1}} &\leq \frac{R_1}{2} 
\left( 1 + e^{C_{\varepsilon} R_1 K_{\alpha, \varepsilon, T}^2 T } \right) 
\leq R_1
\end{align}
if we choose $\widetilde{C}_{\varepsilon}$ in the definition of $T$ sufficiently large
because $R_1 \geq 1 + 20 K_{\varepsilon}^2  R_2^{8} \norm{\omega_0}_{\textfrak{S}^{2,1}}$.
In total, this shows that  $\Phi$ maps $\mathcal{Z}_{T,R_1, R_2}$ into itself.

\paragraph{Contraction property of $\Phi$:}

Next, we are going to prove that the mapping $\Phi$ is a contraction on $\mathcal{Z}_{T,R_1, R_2}$. Let us consider $(\omega,\alpha)$, $(\omega',\alpha') \in \mathcal{Z}_{T,R_1,R_2}$ and denote their images under the mapping $\Phi$ by $(\Gamma, \xi)$ and $(\Gamma', \xi')$ respectively.
Using
\begin{align}
&i \varepsilon \partial_t U_{\alpha}^*(t;0) \left( \Gamma_t - \Gamma_t' \right) U_{\alpha}(t;0) 
\nonumber \\
&\quad = U_{\alpha}^*(t;0)
\bigg(
\left[ K * \rho_{\omega_t} - X_{\omega_t}   , \Gamma_t \right]
- \left[ K * \rho_{\omega_t'} - X_{\omega_t'}  , \Gamma_t' \right]
\nonumber \\
&\qquad
- \left[ \left( - i \varepsilon \nabla - \kappa * \vA_{\alpha'_t} \right)^2 - \left( - i \varepsilon \nabla - \kappa * \vA_{\alpha_t} \right)^2 , \Gamma_t' \right]
\bigg) U_{\alpha}(t;0) 
\end{align}
and the fact that both solutions have the same initial condition, we get for $t \in [0,T]$
\begin{subequations}
\begin{align}
\label{eq:MS well-posedness contraction agrument reduced density 1}
\norm{\Gamma_t - \Gamma_t'}_{\textfrak{S}^{1,1}}
&\leq \varepsilon^{-1}
\int_0^t 
\norm{ U_{\alpha}(t;s) \left[ K * \rho_{\omega_s - \omega_s'}  - X_{\omega_s - \omega_s'}  , \Gamma_s \right] U_{\alpha}^*(t;s)}_{\textfrak{S}^{1,1}}ds
\\
\label{eq:MS well-posedness contraction agrument reduced density 2}
&\quad +  \varepsilon^{-1}
\int_0^t   \norm{ U_{\alpha}(t;s) \left[ K * \rho_{\omega_s'} - X_{\omega_s'}  , \left( \Gamma_s - \Gamma_s' \right) \right] U_{\alpha}^*(t;s) }_{\textfrak{S}^{1,1}}ds
\\
\label{eq:MS well-posedness contraction argument reduced density 3}
&\quad +  \varepsilon^{-1}
\int_0^t 
\norm{ U_{\alpha}(t;s) \left[ \big( (\kappa * \vA_{\alpha_s} )^2 - (\kappa * \vA_{\alpha_s'} )^2
\big) , \Gamma_s' \right] U_{\alpha}^*(t;s)}_{\textfrak{S}^{1,1}} ds
\\
\label{eq:MS well-posedness contraction agrument reduced density 4}
&\quad + 2 
\int_0^t 
\norm{ U_{\alpha}(t;s) \left[ \kappa * \vA_{\alpha_s - \alpha_s'} \cdot \nabla , \Gamma_s' \right] U_{\alpha}^*(t;s) }_{\textfrak{S}^{1,1}}ds .
\end{align}
\end{subequations}
Next, we estimate each line separately. By means of Lemma \ref{lemma:evolution operator for magnetic laplacian} and Lemma~\ref{lemma:general estimates for the interaction} we obtain
\begin{align}
\eqref{eq:MS well-posedness contraction agrument reduced density 1}
&\leq \varepsilon^{-1} \int_0^t 
\norm{ U_{\alpha}(t;s)}_{\textfrak{S}^{\infty}(H^1(\mathbb{R}^3))}
\norm{ U_{\alpha}^*(t;s)}_{\textfrak{S}^{\infty}(H^1(\mathbb{R}^3))}
\norm{\Gamma_s}_{\textfrak{S}^{1,1}}
\nonumber \\
&\qquad \qquad \times \left(
\norm{ K * \rho_{\omega_s - \omega_s'}}_{W_0^{1, \infty}(\mathbb{R}^3)}
+ C N^{-1} \norm{\omega_s - \omega_s'}_{\textfrak{S}^{1,1}} 
\right)ds
\nonumber \\
&\leq C_{\varepsilon} K_{\varepsilon,\alpha,t} \, t
\norm{\Gamma}_{L_t^{\infty} \textfrak{S}^{1,1}}
\norm{\omega - \omega}_{L_t^{\infty} \textfrak{S}^{1,1}} 
\end{align}
and 
\begin{align}
\eqref{eq:MS well-posedness contraction agrument reduced density 2}
&\leq C_{\varepsilon} K_{\varepsilon,\alpha,t} \, t
\norm{\omega'}_{L_t^{\infty} \textfrak{S}^{1,1}}
\norm{\Gamma - \Gamma'}_{L_t^{\infty} \textfrak{S}^{1,1}}  . 
\end{align}
Using $(\kappa * \vA_{\alpha_s} )^2 - (\kappa * \vA_{\alpha_s'} )^2 = \kappa * \vA_{\alpha_s - \alpha_s'}  \, \kappa * \vA_{\alpha_s + \alpha_s'}$, Lemma \ref{lemma:evolution operator for magnetic laplacian} and \eqref{eq:estimates for the vector potential and interaction 1} we get
\begin{align}
\eqref{eq:MS well-posedness contraction argument reduced density 3}
&\leq  2 \varepsilon^{-1} K_{\varepsilon,\alpha,t} 
\int_0^t 
\norm{\kappa * \vA_{\alpha_s - \alpha_s'}}_{W_0^{1, \infty}(\mathbb{R}^3)}
\norm{\kappa * \vA_{\alpha_s + \alpha_s'}}_{W_0^{1, \infty}(\mathbb{R}^3)}
\norm{\Gamma_s' }_{\textfrak{S}^{1,1}}ds
\nonumber \\
&\leq C_{\varepsilon}  K_{\varepsilon,\alpha,t}  \, t \left( \norm{\alpha}_{L_t^{\infty} \mathfrak{h}}
+ \norm{\alpha'}_{L_t^{\infty} \mathfrak{h}} \right) \norm{\Gamma' }_{L_t^{\infty} \textfrak{S}^{1,1}} 
\norm{\alpha - \alpha'}_{L_t^{\infty} \mathfrak{h}}  
\end{align}
and
\begin{align}
\eqref{eq:MS well-posedness contraction agrument reduced density 4}
&\leq 2 K_{\varepsilon,\alpha,t} 
\int_0^t 
\norm{\kappa * \vA_{\alpha - \alpha'}}_{W_0^{1, \infty}(\mathbb{R}^3)}
\left( 
\norm{\nabla \Gamma_s' }_{\textfrak{S}^{1,1}} 
+  \norm{\Gamma_s' \nabla}_{\textfrak{S}^{1,1}} \right)ds
\nonumber \\
&\leq  C K_{\varepsilon,\alpha,t}  \, t \norm{\Gamma }_{L_t^{\infty}  \textfrak{S}^{2,1}}  \norm{\alpha - \alpha'}_{L_t^{\infty} \mathfrak{h}}   .
\end{align}
Summing all contributions together, taking the supremum in $t$ over the interval $[0,T]$ and using that $(\omega, \alpha)$, $(\omega', \alpha)$, $(\Gamma, \xi)$, $(\Gamma', \xi') \in \mathcal{Z}_{T,R_1,R_2}$ lead to
\begin{align}
\norm{\Gamma - \Gamma'}_{L_T^{\infty} \textfrak{S}^{1,1}}
&\leq C K_{\varepsilon, \alpha, T} R_1 \, T
\left( \norm{\Gamma - \Gamma'}_{L_T^{\infty} \textfrak{S}^{1,1}}
+ R_2 \, d_T \left( (\omega, \alpha) , (\omega', \alpha') \right) \right) .
\end{align}
By means of \eqref{eq:estimate for the constant of the magnetic laplacian}
we have $C K_{\varepsilon, \alpha, T} R_1 \, T \leq 1/2$ and $2 C K_{\varepsilon, \alpha, T} R_1 R_2 \, T < 1$ for $\widetilde{C}_\varepsilon$ in the definition of $T$ chosen sufficiently large. This leads to
\begin{align}
\norm{\Gamma - \Gamma'}_{L_T^{\infty} \textfrak{S}^{1,1}}
&< d_T \left( (\omega, \alpha) , (\omega', \alpha') \right) .
\end{align}
Using Duhamel's formula, \eqref{eq:assumption on the cutoff parameter} and
\eqref{eq:Charge current L-infty estimate of the difference} we get
\begin{align}
\label{eq:MS well-posedness contraction estimate for the field}
\norm{\xi_t - \xi_t'}_{L_T^{\infty} \mathfrak{h}_{1/2} \, \cap \, L_T^{\infty} \dot{\mathfrak{h}}_{-1/2}}
&\leq C \sup_{t \in [0,T]}
 \int_0^t  \norm{\left< \cdot \right> \abs{\cdot}^{-1}  \mathcal{F}[\kappa]}_{L^2(\mathbb{R}^3)} \norm{\mathcal{F}[\vJ_{\omega_s, \alpha_s} - \vJ_{\omega_s', \alpha_s'}]}_{L^{\infty}(\mathbb{R}^3, \mathbb{C}^3)}ds
\nonumber \\
&\leq C_{\varepsilon}  \left( 1 + \norm{\omega'}_{L_T^{\infty} \textfrak{S}^1} + \norm{\alpha}_{L_T^{\infty} \mathfrak{h}} \right) T  d_T \left( (\omega, \alpha) , (\omega', \alpha') \right) 
\nonumber \\
&\leq C_{\varepsilon}  \left( 1 + R_1 + R_2 \right) T  d_T \left( (\omega, \alpha) , (\omega', \alpha') \right) 
\nonumber \\
&<   d_T \left( (\omega, \alpha) , (\omega', \alpha') \right) .
\end{align}
For $\widetilde{C}_{\varepsilon} >0$ chosen sufficiently large and $T \leq \frac{1}{\widetilde{C}_{\varepsilon} R_1^2 R_2^8}$ the estimates above imply
\begin{align}
d_{T}  \big( \Phi (\omega, \alpha), \Phi (\omega' , \alpha' )  \big)
< d_{T} \big( (\omega, \alpha) , (\omega' , \alpha' )  \big)  ,
\end{align}
showing that $\Phi$ is a contraction on $\mathcal{Z}_{T,R_1,R_2}$.

\end{proof}

\subsection{Global solutions}

Here below we conclude the proof of Proposition~\ref{proposition:well-posedness of Maxwell-Schroedinger} concerning the well-posedness and regularity theory for the Vlasov--Maxwell system.

\begin{proof}[Proof of Proposition~\ref{proposition:well-posedness of Maxwell-Schroedinger}]

From Lemma~\ref{lemma:MS system local well-posedness result stated by use of negative weighted Lp-spaces} we infer the existence of a time $T$ and a unique local $C \big( 0,T; \textfrak{S}^{2,1} \big( L^2(\mathbb{R}^3) \big) \big) \cap C^1 \big( 0,T; \textfrak{S}^{1} \big( L^2(\mathbb{R}^3) \big) \big)   \times  C \big( 0,T; \mathfrak{h}_{1/2} \cap \dot{\mathfrak{h}}_{-1/2} \big) \cap C^1 \big( 0,T; \dot{\mathfrak{h}}_{-1/2} \big) $-valued solution of \eqref{eq:Maxwell-Schroedinger equations} with initial datum $(\omega_0, \alpha_0)$, which will be denoted by $(\omega, \alpha)$ in the following. Let $t \in [0,T]$, $H_{\omega,\alpha}(t) = \left( - i \varepsilon \nabla - \kappa * \vA_{\alpha_t} \right)^2  + K * \rho_{\omega_t} - X_{\omega_t}$ and $\psi \in H^2(\mathbb{R}^3)$. Using \eqref{eq:vector potential L-infty estimate}, \eqref{eq:W-0-2-infty estimate for the direct interaction potential}, \eqref{eq:estimate exchange term in operator norm} and the regularity properties of $(\omega,\alpha)$ it is straightforward to show the estimate
\begin{align}
\label{eq:estimate for the Hamiltonian from Hatree-Fock against Laplacian}
\norm{\left( H_{\omega, \alpha}(t) + \varepsilon^2 \Delta \right) \psi}_{L^2(\mathbb{R}^3)}
&\leq \delta \norm{- \varepsilon^2 \Delta \psi}_{L^2(\mathbb{R}^3)}
+ C \Big[ \left( 1 + \delta^{-1} \right)
\norm{\alpha_t}_{\dot{\mathfrak{h}}_{-1/2}}^2 
\nonumber \\
&\quad 
+ N^{-1} \norm{\omega_t}_{\textfrak{S}^1}
\Big] \norm{\psi}_{L^2(\mathbb{R}^3)}  
\end{align} 
and the continuous differentiability of the mapping $t\in [0,T] \mapsto H_{\omega, \alpha}(t) \psi\in L^2(\mathbb{R}^3)$. Inequality \eqref{eq:estimate for the Hamiltonian from Hatree-Fock against Laplacian} and the Kato-Rellich theorem imply that $H_{\omega,\alpha}(t)$ is a self-adjoint operator with domain $\mathcal{D}( H_{\omega, \alpha}(t)) = H^2(\mathbb{R}^3)$. Together with the strong differentiability this (see \cite[Theorem X.70, proof of Theorem X.71 ]{RS1975} and \cite[Theorem 2.2]{GS2014})
gives rise to a two-parameter family $\{ U_{\omega,\alpha}(t,s) \}_{(t,s) \in [0,T]^2}$ of unitary operators on $L^2(\mathbb{R}^3)$ such  that $ U_{\omega,\alpha}(t;s) H^2(\mathbb{R}^3) \subset H^2(\mathbb{R}^3)$ and  $\psi_s(t) = U_{\omega;\alpha}(t,s) \psi$ with $\psi \in H^2(\mathbb{R}^3)$  is strongly continuous differentiable and satisfies $\frac{d}{dt}  \psi_s(t) = - i \varepsilon^{-1} H_{\omega, \alpha}(t) \psi_s(t)$, $\psi_s(s) = \psi$. The operator $U_{\omega, \alpha}(t;0) \omega_0 U_{\omega, \alpha}(t;0)^* \in C \big( 0,T; \textfrak{S}^{2,1} \big( L^2(\mathbb{R}^3) \big) \big) \cap C^1 \big( 0,T; \textfrak{S}^{1} \big( L^2(\mathbb{R}^3) \big) \big)$ satisfies the first equation of \eqref{eq:Maxwell-Schroedinger equations} with initial value $\omega_0$. Uniqueness then implies that $\omega_t = U_{\omega, \alpha}(t;0) \omega_0 U_{\omega, \alpha}(t;0)^*$, hence $\omega_t \in \textfrak{S}_{+}^{1}$ for all $t \in [0,T]$ because $\omega_0 \geq 0$ by assumption. The conservation of mass and energy, i.e. 
\begin{align}
\tr  \left( \omega_t \right) = \tr  \left( \omega_0 \right)
\quad \text{and} \quad 
\mathcal{E}^{\rm{MS}}[\omega_t, \alpha_t] = \mathcal{E}^{\rm{MS}}[\omega_0, \alpha_0]
\quad \text{for all} \; t \in [0,T] ,
\end{align}
are obtained by direct inspection. In the following, we will prove 
\begin{align}
\label{eq:estimate for global existence 1}
\norm{\alpha_t}_{\mathfrak{h}_{1/2} \, \cap \, \dot{\mathfrak{h}}_{-1/2}} 
&\leq \norm{\alpha_0}_{\mathfrak{h}_{1/2} \, \cap \, \dot{\mathfrak{h}}_{-1/2}}  +  t \,  C \big( \varepsilon, N \norm{\omega_0}_{\textfrak{S}^{1,1}} , \norm{\alpha_0}_{\dot{\mathfrak{h}}_{1/2}} \big)  
\end{align}
and
\begin{align}
\label{eq:estimate for global existence 2}
\norm{\omega_t}_{\textfrak{S}^{2,1}}
&\leq \exp \left[ \exp \left[ \left< t \right>^2 \,  C \big( \varepsilon, N \norm{\omega_0}_{\textfrak{S}^{2,1}} , \norm{\alpha_0}_{\mathfrak{h}_{1/2} \, \cap \, \dot{\mathfrak{h}}_{-1/2}} \big) \right] \right]
\end{align}
for all $t \in [0,T]$.
These estimates lead to the global existence of solutions because, by standard methods (see e.g. \cite[Theorem 4.3.4]{CH1998}),
 one can derive from the contraction mapping principle a blow up alternative which states that the maximal time of existence $T_{\text{max}}$ is either infinite or $\lim_{t \nearrow T_{\text{max}}} ( \norm{\omega_t}_{\textfrak{S}^{2,1}} + \norm{\alpha_t}_{\mathfrak{h}_{1/2} \, \cap \, \dot{\mathfrak{h}}_{-1/2}} ) = \infty$.

\paragraph{Inequality \eqref{eq:estimate for global existence 1}.}

Using \eqref{eq:estimate for the MS energy functional 1}--\eqref{eq:estimate for the MS energy functional 3} and the conservation of mass and energy, we get
\begin{align}
\label{eq:propagation of the weighted norms for alpha during Maxwell-Schroedinger}
\norm{\omega_t}_{\textfrak{S}^{1,1}} + \norm{\alpha_t}_{\dot{\mathfrak{h}}_{1/2}}^2  &\leq 
C \big( \varepsilon, N \norm{\omega_0}_{\textfrak{S}^{1,1}} , \norm{\alpha_0}_{\dot{\mathfrak{h}}_{1/2}} \big)  
\quad \text{for all} \; t \in [0,T] .
\end{align}
Together with Duhamel's formula, \eqref{eq:assumption on the cutoff parameter} and \eqref{eq:Charge current L-infty estimate of the difference} this leads to 
\begin{align}
\norm{\alpha_t}_{\mathfrak{h}_{1/2} \, \cap \, \dot{\mathfrak{h}}_{-1/2}}  
&\leq \norm{\alpha_0}_{\mathfrak{h}_{1/2} \, \cap \, \dot{\mathfrak{h}}_{-1/2}} 
+ C \int_0^t 
\norm{\left< \cdot \right> \abs{\cdot}^{-1} \mathcal{F}[\kappa] }_{L^2(\mathbb{R}^3)}
\norm{ \mathcal{F}[\vJ_{\omega_s, \alpha_s}]}_{L^{\infty}(\mathbb{R}^3, \mathbb{C}^3)}ds
\nonumber \\
&\leq \norm{\alpha_0}_{\mathfrak{h}_{1/2} \, \cap \, \dot{\mathfrak{h}}_{-1/2}} 
+ C_{\varepsilon} \int_0^t  \left( 1 + \norm{\alpha_s}_{\dot{\mathfrak{h}}_{1/2}} \right)  \norm{\omega_s}_{\textfrak{S}^{1,1}}ds
\nonumber \\
&\leq \norm{\alpha_0}_{\mathfrak{h}_{1/2} \, \cap \, \dot{\mathfrak{h}}_{-1/2}} + t \,  C \big( \varepsilon, N \norm{\omega_0}_{\textfrak{S}^{1,1}} , \norm{\alpha_0}_{\dot{\mathfrak{h}}_{1/2}} \big)   
\quad \text{for all} \; t \in [0,T] .
\end{align}

\paragraph{Inequality \eqref{eq:estimate for global existence 2}.}

By means of the Duhamel expansion
\begin{align}
\omega_t &= U_{\alpha}(t;0) \omega_0 U_{\alpha}^*(t;0)
- i \varepsilon^{-1} \int_0^t U_{\alpha}(t;s)
\left[ K * \rho_{\omega_s} - X_{\omega_s}  , \omega_s \right] U_{\alpha}^*(t;s)\,ds,
\end{align}
with $U_{\alpha}$ being defined as in Lemma \ref{lemma:evolution operator for magnetic laplacian}, Lemma \ref{lemma:evolution operator for magnetic laplacian} and \eqref{eq:W-0-2-infty estimate for the direct interaction potential},  we get
\begin{align}
\label{eq:intermediate estimate for propagation of L-2-1 norm of omega}
\norm{\omega_t}_{\textfrak{S}^{2,1}}
&\leq 
K_{\varepsilon, \alpha, t}^2
\left( \norm{\omega_0}_{\textfrak{S}^{2,1}}
+ 
C_{\varepsilon} \int_0^t 
\left(  \norm{\omega_s}_{\textfrak{S}^{1}}  \norm{\omega_s}_{\textfrak{S}^{2,1}}
+ \norm{\left[ X_{\omega_s}, \omega_s \right]}_{\textfrak{S}^{2,1}} \right)\,ds
\right) .
\end{align}
Since $\omega_s \in \textfrak{S}_{+}^{2,1}(L^2(\mathbb{R}^3))$, there exists a spectral set $\{ \lambda_j, \psi_j \}_{j \in \mathbb{N}}$ with $\lambda_j \geq 0$ for all $j \in \mathbb{N}$ such that $\omega_s = \sum_{j \in \mathbb{N}} \lambda_j \ket{\psi_j} \bra{\psi_j}$,  $\norm{\omega_s}_{\textfrak{S}^{1}} = \sum_{j \in \mathbb{N}} \lambda_j \norm{\psi_j}_{L^2(\mathbb{R}^3)}$ and
$\norm{\omega_s}_{\textfrak{S}^{2,1}} = \sum_{j \in \mathbb{N}} \lambda_j \norm{\psi_j}_{H^2(\mathbb{R}^3)}$.
This allows us to write out the expression involving the exchange term quite explicitly. More concretely, we bound
\begin{align}
\norm{\left[ X_{\omega_s}, \omega_s \right]}_{\textfrak{S}^{2,1}}
&\leq \norm{X_{\omega_s} \omega_s }_{\textfrak{S}^{2,1}}
+ \norm{\omega_s X_{\omega_s}}_{\textfrak{S}^{2,1}}
\end{align}
and then use the projection property of $\ket{\psi_j} \bra{\psi_j}$  to estimate
\begin{align}
\norm{X_{\omega_s} \omega_s}_{\textfrak{S}^{2,1}}
&\leq \sum_{j \in \mathbb{N}} \lambda_j
\norm{\left( 1 - \Delta \right) X_{\omega_s} \ket{\psi_j} \bra{\psi_j} \left( 1 - \Delta \right)}_{\textfrak{S}^1}
\nonumber \\
&\leq \sum_{j \in \mathbb{N}} \lambda_j
\norm{\left( 1 - \Delta \right) X_{\omega_s} \ket{\psi_j} \bra{\psi_j} }_{\textfrak{S}^2} \,
\norm{\ket{\psi_j} \bra{\psi_j}\left( 1 - \Delta \right)}_{\textfrak{S}^2}
\nonumber \\
&= 
\sum_{j \in \mathbb{N}} \lambda_j \norm{\psi_j}_{H^2}
\norm{(1 - \Delta) X_{\omega_s} \psi_j}_{L^2(\mathbb{R}^3)} \nonumber \\
&\leq N^{-1} \sum_{j,k \in \mathbb{N}} \lambda_j \lambda_k
\norm{\psi_j}_{H^2(\mathbb{R}^3)}
\norm{\psi_k K * \{ \overline{\psi_k} \psi_j \}}_{H^2(\mathbb{R}^3)}
\nonumber \\
&\leq N^{-1} \sum_{j,k \in \mathbb{N}} \lambda_j \lambda_k
\norm{\psi_j}_{H^2(\mathbb{R}^3)}
\norm{\psi_k}_{H^2(\mathbb{R}^3)}
\norm{K * \{ \overline{\psi_k} \psi_j \}}_{W_0^{2,\infty}(\mathbb{R}^3)} .
\end{align}
Using that  $\norm{\left< \cdot \right>^2 \mathcal{F}[K]}_{L^1(\mathbb{R}^3)} \leq C$ because of \eqref{eq:Fourier transform of the potential} and \eqref{eq:assumption on the cutoff parameter},  Young's inequality and the Cauchy--Schwarz inequality, we obtain
\begin{align}
\norm{X_{\omega_s} \omega_s}_{\textfrak{S}^{2,1}}
&\leq N^{-1} \norm{K}_{W_0^{2,\infty}(\mathbb{R}^3)}  \sum_{j,k \in \mathbb{N}} \lambda_j \lambda_k
\norm{\psi_j}_{H^2(\mathbb{R}^3)}
\norm{\psi_k}_{H^2(\mathbb{R}^3)}
\norm{\overline{\psi_k} \psi_j}_{L^1(\mathbb{R}^3)}
\nonumber \\
&\leq N^{-1}
\norm{\left< \cdot \right>^2 \mathcal{F}[K]}_{L^1(\mathbb{R}^3)} 
\sum_{j,k \in \mathbb{N}} \lambda_j \lambda_k
\norm{\psi_j}_{H^2(\mathbb{R}^3)} \norm{\psi_j}_{L^2(\mathbb{R}^3)}
\norm{\psi_k}_{H^2(\mathbb{R}^3)}
\norm{\psi_k}_{L^2(\mathbb{R}^3)}
\nonumber \\
&\leq 
C N^{-1}
\left( \sum_{j \in \mathbb{N}} \lambda_j \norm{\psi_j}_{H^2(\mathbb{R}^3)}^2 \right)
\left( \sum_{k \in \mathbb{N}} \lambda_k \norm{\psi_k}_{L^2(\mathbb{R}^3)}^2 \right)
\nonumber \\
&\leq 
C N^{-1} \norm{\omega_s}_{\textfrak{S}^{1}} \norm{\omega_s}_{\textfrak{S}^{2,1}} .
\end{align}
Plugging this into \eqref{eq:intermediate estimate for propagation of L-2-1 norm of omega}
and  the conservation of mass lead to
\begin{align}
\norm{\omega_t}_{\textfrak{S}^{2,1}}
&\leq K_{\varepsilon, \alpha, t}^2 \left(
\norm{\omega_0}_{\textfrak{S}^{2,1}}
+ C_{\varepsilon}  \norm{\omega_0}_{\textfrak{S}^{1}}  \int_0^t 
\norm{\omega_s}_{\textfrak{S}^{2,1}}ds \right) .
\end{align}
Combining similar estimates as in the proof of \eqref{eq:estimate for global existence 1} with \eqref{eq:estimate for global existence 1} gives
\begin{align}
\label{eq:propagation of h -1-2 norm of derivative of alpha during Maxwell-Schroedinger}
\norm{\partial_t \alpha_t}_{\dot{\mathfrak{h}}_{- 1/2}}
&\leq 
\norm{ \alpha_t}_{\dot{\mathfrak{h}}_{1/2}} 
+ C \int_0^t 
\norm{ \abs{\cdot}^{-1} \mathcal{F}[\kappa] }_{L^2(\mathbb{R}^3)}
\norm{ \mathcal{F}[\vJ_{\omega_s, \alpha_s}]}_{L^{\infty}(\mathbb{R}^3, \mathbb{C}^3)}ds
\nonumber \\
&\leq  \left< t \right> \,  C \big( \varepsilon, N \norm{\omega_0}_{\textfrak{S}^{1,1}} ,  \norm{\alpha_0}_{\mathfrak{h}_{1/2} \, \cap \, \dot{\mathfrak{h}}_{-1/2}}  \big) 
\quad \text{for all} \; t \in [0,T] .
\end{align}
Together with \eqref{eq:estimate for global existence 1} this allows us to estimate
\begin{align}
K_{\varepsilon, \alpha, t} &\leq e^{\left< t \right>^2 \,  C \big( \varepsilon, N \norm{\omega_0}_{\textfrak{S}^{1,1}} , \norm{\alpha_0}_{\mathfrak{h}_{1/2} \, \cap \, \dot{\mathfrak{h}}_{-1/2}} \big)}  ,
\end{align}
leading to
\begin{align}
\norm{\omega_t}_{\textfrak{S}^{2,1}}
&\leq e^{\left< t \right>^2 \,  C \big( \varepsilon, N \norm{\omega_0}_{\textfrak{S}^{2,1}} , \norm{\alpha_0}_{\mathfrak{h}_{1/2} \, \cap \, \dot{\mathfrak{h}}_{-1/2}} \big)}  \left( 1 + \int_0^t ds \,
\norm{\omega_s}_{\textfrak{S}^{2,1}} \right) .
\end{align}
Application of Gr\"{o}nwall's Lemma then gives \eqref{eq:estimate for global existence 2}.

\end{proof}

\section{Solution theory of the regularized Vlasov--Maxwell system}\label{section:Properties of the solutions of the Vlasov-Maxwell equations}

In this section, we prove the existence of  unique global solutions for \eqref{eq: Vlasov-Maxwell different style of writing}. Similarly as in \cite{D1986} we restrict our consideration to initial data with compact support in the velocity variable. 

\subsection{Propagation estimates for the characteristics}

\begin{lemma}
Let  $T>0$, $(x,v) \in \mathbb{R}^6$, $f \in C(0,T ; W_4^{0,2}(\mathbb{R}^6))$ and $\alpha \in C(0,T; \mathfrak{h}_1)$.
Then, 
\begin{align}
\label{eq:characteristics}
\begin{cases}
\dot{X}_t(x,v) &=  2 V_t(x,v) - 2 \kappa * \vA_{\alpha_t}(X_t(x,v))  \\
\dot{V}_t(x,v) &=  - \vF_{f_t, \alpha_t} ( X_t(x,v), V_t(x,v)) .
\end{cases}  
\end{align}
with initial datum $(X_t(x,v), V_t(x,v)) \big|_{t=0} = (x,v)$ has a unique $C^1 \big( 0,T, \mathbb{R}^6 \big)$ solution satisfying
\begin{subequations}
\begin{align}
\label{eq:propagation estimate on V}
\left< V_t(x,v) \right>
&\leq C \left< v \right> e^{C  \int_0^t  \left( \norm{f_s}_{W_4^{0,2}(\mathbb{R}^6)} + \norm{\alpha_s}_{\mathfrak{h}}^2 \right)ds} ,
\\
\label{eq:propagation estimate on X}
\left<X_t(x,v) \right>
&\leq C  \left<(x,v) \right> \left< t \right>
e^{C  \int_0^t  \left( \norm{f_s}_{W_4^{0,2}(\mathbb{R}^6)} + \norm{\alpha_s}_{\mathfrak{h}}^2 \right)ds} .
\end{align}
\end{subequations}
\end{lemma}

\begin{proof}
The existence of a unique local solution can be shown by a standard fixed point argument.  Using  \eqref{eq:estimates for the vector potential and interaction 1}   and \eqref{eq:estimates for the vector potential and interaction 2}  we obtain
\begin{align}
\sup_{x \in \mathbb{R}^3}  \abs{\vF_{f_s, \alpha_s}(x,v)}
&\leq C \left< v \right> \left( \norm{f_s}_{W_4^{0,2}(\mathbb{R}^6)} + \norm{\alpha_s}_{\mathfrak{h}}^2 \right)  .
\end{align}
Together with  \eqref{eq:characteristics} this shows
\begin{align}
\abs{V_t(x,v)} &\leq v + C  \int_0^t  \left( \norm{f_s}_{W_4^{0,2}(\mathbb{R}^6)} + \norm{\alpha_s}_{\mathfrak{h}}^2 \right) \left< V_s(x,v) \right> ds
\end{align}
and \eqref{eq:propagation estimate on V} by Gr\"{o}nwall's lemma. Inequality \eqref{eq:propagation estimate on X} is a consequence of \eqref{eq:characteristics}, \eqref{eq:estimates for the vector potential and interaction 1}  and \eqref{eq:propagation estimate on V}. By means of \eqref{eq:propagation estimate on V} and \eqref{eq:propagation estimate on X} the solution can then be extended to the whole interval $[0,T]$.
\end{proof}

\subsection{Linear equations}
\label{subsection:VM existence theory linearized equations}

\begin{lemma} 
\label{lemma:VM existence theory linear equation for the density}
Let $R>0$, $T>0$, $a, b \in \mathbb{N}$ such that $a \geq 4$ and $b \geq 3$. Moreover, let $(f, \alpha) \in C \big( 0,T; W_a^{b-1,2}(\mathbb{R}^6) \big) \times  C \big( 0,T ; \mathfrak{h}_{b} \cap \dot{\mathfrak{h}}_{-1/2} \big)$ and $g_0 \in W_{a}^{b,2}(\mathbb{R}^6)$ such that $\supp g_0 \subset A_R$ with $A_R = \{ (x,v) \in \mathbb{R}^6 , \abs{v} \leq R \}$. Then, 
\begin{align}
\label{eq:VM existence theory linear equation for the density}
\partial_t g_t  &= - 2 \left( v - \kappa * \vA_{\alpha_t} \right) \cdot \nabla_x g_t + \vF_{f_t, \alpha_t} \cdot \nabla_v g_t
\end{align}
with initial datum $g_0$ and $\vF_{f_t, \alpha_t} $ being defined as in \eqref{eq:F-field of Vlasov equation} has a unique$L^{\infty} \big( 0,T; W_{a}^{b,2}(\mathbb{R}^6) \big) \cap C \big( 0,T; W_{a}^{b-1,2}(\mathbb{R}^6) \big)  \cap C^1 (0,T; W_{a}^{b-2,2}(\mathbb{R}^6))$
solution. In addition, $\supp g_t \subset A_{R(t)}$ with $R(t) 
\leq C \left< R \right> \exp \left( C  \int_0^t \big( \norm{f_s}_{W_4^{0,2}(\mathbb{R}^6)} + \norm{\alpha_s}_{ \mathfrak{h}}^2 \big)ds \right)$ and 
\begin{align}
\label{eq:VM fixed point estimate W-a-b-2 for g-t}
\norm{g_t}_{W_a^{b,2}(\mathbb{R}^6)} &\leq 
C \norm{g_0}_{W_a^{b,2}}
\exp \left( C \int_0^t   R(s) \big( 1 + \norm{f_s}_{W_4^{b-2,2}(\mathbb{R}^6)} + \norm{\alpha_s}_{\mathfrak{h}_b}^2 \big)ds \right) .
\end{align}
\end{lemma}

\begin{proof}
For $T,R_1, R_2 > 0$   we consider the Banach space $(\mathcal{Z}_{T,R_1,R_2}, d_T)$ with
\begin{align}
\mathcal{Z}_{T,R_1,R_2}
&= \Big\{ g \in C(0,T; W_a^{b-1,2}(\mathbb{R}^6)) \cap L^{\infty} (0,T; W_a^{b,2}(\mathbb{R}^6))
: g_t \Big|_{t=0} = g_0,  
\nonumber \\
&\qquad
\norm{\id_{\abs{v} \geq R_1} g}_{L_T^{\infty} L^2(\mathbb{R}^6)} = 0,  
\norm{\left( 1 - \Delta \right)^{b/2}  \left< \cdot \right>^a g}_{L_T^{\infty} L^2(\mathbb{R}^6)} \leq R_2
\Big\}
\end{align}
and metric
$d_T(g,g') = \norm{g - g'}_{L_T^{\infty} W_a^{b-1,2}(\mathbb{R}^6)}$.
We, moreover, define the mapping $\Phi: \mathcal{Z}_{T,R_1,R_2} \rightarrow \mathcal{Z}_{T,R_1,R_2}$ by
\begin{align}
\Phi_t(g) = g_0 + \int_0^t  
\left( - 2 \left( v - \kappa * \vA_{\alpha_s} \right) \cdot \nabla_x g_s + \vF_{f_s, \alpha_s} \cdot \nabla_v g_s \right) ds.
\end{align}
For all 
\begin{align}
\label{eq:VM fixed point estimate linear equation f condition on the time parameter}
R_2 \geq C \norm{g_0}_{W_a^{b,2}(\mathbb{R}^6)}
\quad \text{and} \quad
T^* \leq  \frac{1}{C \left< R_1 \right> \left( 1 + \norm{f}_{L_T^{\infty} W_4^{b-2,2}(\mathbb{R}^6)} + \norm{\alpha}_{L_T^{\infty} \mathfrak{h}_b}^2 \right) } 
\end{align}
with $C$ chosen sufficiently large it is shown below that the mapping $\Phi$ is well-defined and a contraction on $(\mathcal{Z}_{T^*,R_1,R_2},d_{T^*})$.
By the Banach fixed point theorem this proves the existence of a unique $g \in C (0,T^*; W_{a}^{b-1,2}(\mathbb{R}^6)) \cap L^{\infty}(0,T^*; W_a^{b,2}(\mathbb{R}^6))$ satisfying
\begin{align}
g_t = g_0 + \int_0^t  
\left( - 2 \left( v - \kappa * \vA_{\alpha_s} \right) \cdot \nabla_x g_s + \vF_{f_s, \alpha_s} \cdot \nabla_v g_s \right) ds
\end{align}
in $W_{a}^{b-1,2}(\mathbb{R}^6))$. Note that $\kappa * \vA_{\alpha} \in C(0,T; W_0^{b+1,\infty}(\mathbb{R}^3))$ and $K  * \widetilde{\rho}_{f} \in C(0,T; W_0^{b+2,\infty}(\mathbb{R}^3))$ 
because of \eqref{eq:estimates for the vector potential and interaction 1}  and  \eqref{eq:estimates for the vector potential and interaction 2} and the assumed regularity of $(f,\alpha)$. This allows us to conclude that the integrand in the equation above is an element of $C (0,T^*; W_{a}^{b-2,2}(\mathbb{R}^6))$. Defining the integral in the Riemann sense then proves $\Phi(g) \in C^1 (0,T^*; W_{a}^{b-2,2}(\mathbb{R}^6))$. Let $\big( X^{-1}_{t}(x,v) , V^{-1}_{t}(x,v)) \big)$ be the backward flow of the characteristics \eqref{eq:characteristics} with initial datum $(X_t^{-1}(x,v), V_t^{-1}(x,v)) \big|_{t=0} = (x,v)$. Since
\begin{align}
&\partial_t g_0 \big( X_t^{-1}(x,v) , V_t^{-1}(x,v)) \big)
\nonumber \\
&\quad = - \left( 2 V_t^{-1}(x,v) - 2 \kappa * \vA_{\alpha_t}(X_t^{-1}(x,v)) \right) \cdot \left( \nabla_1  g_0 \right) \big( X_t^{-1}(x,v) , V_t^{-1}(x,v) \big)
\nonumber \\
&\qquad +   \vF_{f_t, \alpha_t} ( X_t^{-1}(x,v), V_t^{-1}(x,v)) \cdot \left( \nabla_2  g_0 \right) \big( X_t^{-1}(x,v)  , V_t^{-1}(x,v) \big)
\end{align}
we obtain $g_t(x,v) = g_0 \big( X_t^{-1}(x,v) , V_t^{-1}(x,v)) \big)$ for all $t \in [0, T^*]$ by the uniqueness of \eqref{eq:VM existence theory linear equation for the density}. Using the support assumption on $g_0$, the fact that $(x,v) = \big( X_t \big( X_t^{-1}(x,v) \big) , V_t \big( V_t^{-1}(x,v) \big) \big)$ and the propagation estimate \eqref{eq:propagation estimate on V} we get that $\supp g_t \subset A_{R(t)}$ with
\begin{align}
R(t) 
&\leq C \left< R \right> e^{ C  \int_0^t  \big( \norm{f_s}_{W_4^{0,2}(\mathbb{R}^6)} + \norm{\alpha_s}_{ \mathfrak{h}}^2 \big)ds } .
\end{align}
This allows us to choose $R_1 = 2 C \left< R \right> \exp \big( C T  \big( \norm{f_s}_{L_T^{\infty} W_4^{0,2}(\mathbb{R}^6)} + \norm{\alpha}_{L_T^{\infty} \mathfrak{h}}^2 \big) \big)$ in the contraction argument below. In this case $T>0$ is fixed, $R_1$ is independent of $t$ and $\supp{g_t} \subset A_{R_1/2}$ holds for all $t$ in the time interval of existence. We can therefore patch local solutions together until we obtain a solution on the whole interval $[0,T]$.

\paragraph{Proof of well-definedness.}
Since $g \in  L^{\infty} (0,T; W_a^{b,2}(\mathbb{R}^6))$ we can define the integral in the definition of $\Phi$ on $W_{a-1}^{b-1,2}(\mathbb{R}^6)$ as a Bochner integral, implying that $\Phi(g) \in C (0,T; W_{a-1}^{b-1,2}(\mathbb{R}^6))$.
Using $L^{\infty} (0,T; W_a^{b,2}(\mathbb{R}^6))$ with $b \geq 3$ we have that $g_t \in C^1(\mathbb{R}^6)$ for all $t \in [0,T]$ and therefore that $\id_{\abs{v} \geq R_1} g_t(x,v) =0$ holds pointwise for all $t \in [0,T]$. This implies that $\id_{\abs{v} \geq R_1} \Phi_t(g) =0$ for all $t \in [0,T]$ and therefore $\norm{\id_{\abs{v} \geq R_1} \Phi(g)}_{L_T^{\infty} L^2(\mathbb{R}^6)} = 0$. Using $g \in  L^{\infty} (0,T; W_a^{b,2}(\mathbb{R}^6))$ again it then follows that  $\Phi(g) \in C (0,T; W_{a}^{b-1,2}(\mathbb{R}^6))$. 
Next, we derive an estimate for $\norm{\left( 1 - \Delta_z \right)^{b/2} \left< \cdot \right>^a \Phi_t(g)}_{L^2(\mathbb{R}^6)}$. To this end let us define the regularized Laplacian 
\begin{align}
\left( 1 - \Delta_z \right)_{\leq n} = \frac{\left( 1 - \Delta_z \right)}{\left( 1 - \Delta_z / n^2 \right)} 
\quad \text{with} \quad 
n \in \mathbb{N}
\end{align}
and the transport operator
\begin{align}
\label{eq:definition transport operator}
T_{f, \alpha} =  2 \left( v - \kappa * \vA_{\alpha} \right) \cdot \nabla_x  - \vF_{f, \alpha} \cdot \nabla_v .
\end{align}
By means of
\begin{align}
\Phi_t(g) &= g_0 - \int_0^t  T_{f_s, \alpha_s}[g_s]\,ds
\end{align}
we compute
\begin{align}
&\norm{\left( 1 - \Delta_z \right)_{\leq n}^{b/2}  \left< \cdot \right>^a  \Phi_t(g)}_{L^2(\mathbb{R}^6)}^2
- \norm{ \left( 1 - \Delta_z \right)_{\leq n}^{b/2} \left< \cdot \right>^a g_0}_{L^2(\mathbb{R}^6)}^2
\nonumber \\
&\quad = - 2 \int_0^t 
\scp{\left( 1 - \Delta_z \right)_{\leq n}^{b/2}  \left< \cdot \right>^a  g_s}{\left( 1 - \Delta_z \right)_{\leq n}^{b/2} \left< \cdot \right>^a T_{f_s, \alpha_s} [g_s]} ds
\nonumber \\
&\quad = - 2 \int_0^t  
\scp{ \left( 1 - \Delta_z \right)_{\leq n}^{b/2}  \left< \cdot \right>^a  g_s}{ T_{f_s, \alpha_s} \Big[ \left( 1 - \Delta_z \right)_{\leq n}^{b/2}   \left< \cdot \right>^a g_s  \Big]} ds
\nonumber \\
&\qquad - 2 \int_0^t  
\scp{\left( 1 - \Delta_z \right)_{\leq n}^{b/2}  \left< \cdot \right>^a g_s}{ \Big[ \left( 1 - \Delta_z \right)_{\leq n}^{b/2} \left< \cdot \right>^a , T_{f_s, \alpha_s} \Big] [g_s]} ds.
\end{align}
Note that $\left[ F_{f, \alpha} , \cdot \nabla_v \right] = 0$ because we are working in the Coulomb gauge  and integration by parts leads to
\begin{align}
\label{eq:VM fixed point argument integration by parts of the transport operator}
2 \int_{\mathbb{R}^6}  m(x,v) n(x,v) T_{f , \alpha}[n](x,v) dx\,dv 
&= -  \int_{\mathbb{R}^6}  T_{f , \alpha}[m](x,v) n^2(x,v) dx\,dv
\end{align}
for sufficiently regular functions $m, n: \mathbb{R}^6 \rightarrow \mathbb{R}$. This let us conclude
\begin{align}
\scp{\left( 1 - \Delta_z \right)_{\leq n}^{b/2} \left< \cdot \right>^a  g_s}{ T_{f_s, \alpha_s} \Big[ \left( 1 - \Delta_z \right)_{\leq n}^{b/2}  \left< \cdot \right>^a g_s  \Big]} = 0 ,
\end{align}
leading to
\begin{align}
&\norm{\left( 1 - \Delta_z \right)_{\leq n}^{b/2}  \left< \cdot \right>^a  \Phi_t(g)}_{L^2(\mathbb{R}^6)}^2
- \norm{\left( 1 - \Delta_z \right)_{\leq n}^{b/2}  \left< \cdot \right>^a  g_0}_{L^2(\mathbb{R}^6)}^2 
\nonumber \\
\label{eq:VM fixed point estimate linear equation W-a-b-2 estimate 1}
&\quad = - 2 \int_0^t  
\scp{\left( 1 - \Delta_z \right)_{\leq n}^{b/2}  \left< \cdot \right>^a g_s}{\left( 1 - \Delta_z \right)_{\leq n}^{b/2}  \Big[  \left< \cdot \right>^a , T_{f_s, \alpha_s} \Big] [g_s]}ds
\\
\label{eq:VM fixed point estimate linear equation W-a-b-2 estimate 2}
&\qquad - 2 \int_0^t  
\scp{\left( 1 - \Delta_z \right)_{\leq n}^{b/2} \left< \cdot \right>^a g_s}{  \Big[ \left( 1 - \Delta_z \right)_{\leq n}^{b/2} , T_{f_s, \alpha_s} \Big] [ \left< \cdot \right>^a g_s ]} ds.
\end{align}
Note that
\begin{align}
\label{eq:commutator of transport operator with weight function}
\left[ T_{f_s, \alpha_s} , \left< z \right>^a \right]
&= a \left< z \right>^{a -2 } \left(  2 \left( v - \kappa * \vA_{\alpha_t} \right) \cdot x  - \vF_{f_t, \alpha_t} \cdot v \right) 
\end{align}
because
\begin{align}
\nabla_x \left< z \right>^{a} 
&= a  \left< z \right>^{a -2}  x
\quad \text{and} \quad 
\nabla_v \left< z \right>^{a} = a  \left< z \right>^{a -2}  v.
\end{align}
If we use the Cauchy-Schwarz inequality and 
$\left( 1 - \Delta_z \right)_{\leq n}^{b/2} \leq \left( 1 - \Delta_z \right)^{b/2}$ we get
\begin{align}
\abs{\eqref{eq:VM fixed point estimate linear equation W-a-b-2 estimate 1}}
&\leq C \int_0^t 
\norm{\left( 1 - \Delta_z \right)^{b/2}  \left< \cdot \right>^a g_s}_{L^2(\mathbb{R}^6)}
\nonumber \\
&\quad \times
\norm{\left( 1 - \Delta_z \right)^{b/2}  \left< z \right>^{a -2 } \left(  2 \left( v - \kappa * \vA_{\alpha_t} \right) \cdot x  - \vF_{f_t, \alpha_t} \cdot v \right) g_s}_{L^2(\mathbb{R}^6)}ds
\nonumber \\
&\leq C \int_0^t 
\norm{g_s}_{W_a^{b,2}(\mathbb{R}^6)}
\norm{\left< z \right>^{a -2 } \left(  2 \left( v - \kappa * \vA_{\alpha_t} \right) \cdot x  - \vF_{f_t, \alpha_t} \cdot v \right) g_s}_{W_0^{b,2}(\mathbb{R}^6)}ds
\end{align}
By means of  \eqref{eq:estimates for the vector potential and interaction 1}  and \eqref{eq:estimates for the vector potential and interaction 2} we obtain
\begin{align}
\abs{\eqref{eq:VM fixed point estimate linear equation W-a-b-2 estimate 1}}
&\leq C \int_0^t 
\left( \norm{f_s}_{W_4^{b-2,2}(\mathbb{R}^6)} + \norm{\alpha_s}_{\mathfrak{h}_b}^2 \right)
\norm{ g_s}_{W_a^{b,2}(\mathbb{R}^6)}^2 ds.
\end{align}
Next, we consider 
\begin{align}
\Big[  \left( 1 - \Delta_z \right)_{\leq n}^{b/2}  , T_{f_s, \alpha_s} \Big] 
&= 2 \left[ \left( 1 - \Delta_z \right)_{\leq n}^{b/2} , v \right] \cdot \nabla_x
- 2 \left[ \left( 1 - \Delta_z \right)_{\leq n}^{b/2} , \kappa * \vA_{\alpha_t}  \right] \cdot \nabla_x
\nonumber \\
&\quad +  2 \sum_{i=1}^3 \nabla \kappa * \vA_{\alpha_s}^i  \left[   \left( 1 - \Delta_z \right)_{\leq n}^{b/2} , v^i  \right] \cdot \nabla_v
\nonumber \\
&\quad + 2 \sum_{i=1}^3 \left[   \left( 1 - \Delta_z \right)_{\leq n}^{b/2} , \nabla \kappa * \vA_{\alpha_s}^i \right]  v^i  \cdot \nabla_v
\nonumber \\
&\quad 
- \left[  \left( 1 - \Delta_z \right)_{\leq n}^{b/2} , \left( \nabla  K * \widetilde{\rho}_{f_s} + \nabla (\kappa * \vA_{\alpha_s})^2 \right) \right] \cdot \nabla_v .
\end{align}
Note that
\begin{align}
\norm{\left[ v^i , \left( 1 - \Delta_z \right)_{\leq n}^{b/2} \right] f }_{L^2(\mathbb{R}^6)}
&\leq \norm{ \left( 1 - \Delta_z \right)_{\leq n}^{b-1/2}  f }_{L^2(\mathbb{R}^6)}
\end{align}
because
$\abs{\partial_{u_j} \left( \frac{1 + u^2}{1 + u^2/ n^2} \right)^{b/2}}
\leq b \left( \frac{1 + u^2}{1 + u^2/ n^2} \right)^{(b-1)/2}$ .
For $V: \mathbb{R}^3 \rightarrow \mathbb{R}^3$ we, moreover, have
\begin{align}
&\left[ \left( 1 - \Delta_z \right)_{\leq n}^{b/2}  , V(x) \right]
\nonumber \\
&\quad = \int_{\mathbb{R}^3} dk \, \mathcal{F}[V](k)  e^{-ikx}
\left( \frac{\left( 1 +  (i \nabla_x + k)^2 - \Delta_v \right)^{b/2}}{\left( 1 + [ (i \nabla_x + k)^2 - \Delta_v ] / n^2 \right)^{b/2}}
- \frac{\left( 1 - \Delta_z \right)^{b/2}}{\left( 1 -  \Delta_z / n^2 \right)^{b/2} } \right) .
\end{align}
Using the estimate
\begin{align}
&\pm \left( \frac{\left( 1 +  (i \nabla_x + k)^2 - \Delta_v \right)^{b/2}}{\left( 1 + [ (i \nabla_x +  k)^2 - \Delta_v ] / n^2 \right)^{b/2}}
- \frac{\left( 1 - \Delta_z \right)^{b/2}}{\left( 1 -  \Delta_z / n^2 \right)^{b/2} } \right)
\nonumber \\
&\quad = \pm b \int_0^1 dr \,
\left( \frac{\left( 1 +  (i \nabla_x + r k)^2 - \Delta_v \right)}{\left( 1 + [ (i \nabla_x + r k)^2 - \Delta_v ] / n^2 \right)} \right)^{b/2 - 1} 
\frac{2 k \cdot (i \nabla_x + r k)  (1 - \frac{1}{n^2})}{\left( 1 + [ (i \nabla_x + r k)^2 - \Delta_v ] / n^2 \right)^2}
\nonumber \\
&\leq C  \abs{k} \left(  \left< k \right>^{b-1} + \left( 1 - \Delta_z \right)^{(b-1/2)} \right) 
\end{align}
we get
\begin{align}
\norm{\left[ \left( 1 - \Delta_z \right)_{\leq n}^{b/2}  , V(x) \right] f }_{L^2(\mathbb{R}^6)}
&\leq  C \int_{\mathbb{R}^3}  \left< k \right>^{b} \abs{\mathcal{FT}[V](k)}
\norm{\left( 1 - \Delta_z \right)^{(b-1)/2} f}_{L^2(\mathbb{R}^6)} dk.
\end{align}
Since
\begin{align}
\int_{\mathbb{R}^3}  \left< k \right>^{b} \abs{\mathcal{FT}[\nabla \kappa * \vA_{\alpha}](k)}dk
&\leq C \norm{\alpha}_{\mathfrak{h}_b}
\end{align}
and
\begin{align}
\int_{\mathbb{R}^3}  \left< k \right>^b\abs{\mathcal{FT}[\nabla K * \widetilde{\rho}_{f}](k)}
)dk
&\leq C \norm{f}_{W_4^{b-2,2}(\mathbb{R}^6)}
\end{align}
can obtained by similar bounds as in the proofs of \eqref{eq:estimates for the vector potential and interaction 1} and \eqref{eq:estimates for the vector potential and interaction 2}  and $\left( 1 - \Delta_z \right)_{\leq n}^{b/2} \leq \left( 1 - \Delta_z \right)^{b/2}$ this leads to
\begin{align}
\norm{\Big[  \left( 1 - \Delta_z \right)_{\leq n}^{b/2}   , T_{f_s, \alpha_s} \Big] f}_{L^2(\mathbb{R}^6)}
&\leq C \left( 1 + \norm{f_s}_{W_4^{b-2,2}(\mathbb{R}^6)} + \norm{\alpha_s}_{\mathfrak{h}_b}^2 \right) \norm{\left< v \right> \left( \nabla_x + \nabla_v \right) f}_{W_0^{b-1,2}(\mathbb{R}^6)} .
\end{align}
Using the  Cauchy--Schwarz inequality we obtain
\begin{align}
\abs{\eqref{eq:VM fixed point estimate linear equation W-a-b-2 estimate 2}}
&\leq C \int_0^t   \left( 1 + \norm{f_s}_{W_4^{b-2,2}(\mathbb{R}^6)} + \norm{\alpha_s}_{\mathfrak{h}_b}^2 \right)
\norm{ g_s}_{W_a^{b,2}(\mathbb{R}^6)}
\nonumber \\
&\quad \times
\norm{\left< v \right> \left( \nabla_x + \nabla_v \right)  \left< \cdot \right>^a g_s}_{W_0^{b-1,2}(\mathbb{R}^6)}ds
\end{align}
Because $\id_{\abs{v} \geq R_1} g_t(x,v) =0$ holds pointwise for all $t \in [0,T]$ and differentiation is a local property we get
\begin{align}
\abs{\eqref{eq:VM fixed point estimate linear equation W-a-b-2 estimate 2}}
&\leq C \left< R_1 \right> \int_0^t   \left( 1 + \norm{f_s}_{W_4^{b-2,2}(\mathbb{R}^6)} + \norm{\alpha_s}_{\mathfrak{h}_b}^2 \right)
\norm{ g_s}_{W_a^{b,2}(\mathbb{R}^6)}
\nonumber \\
&\quad \times
\norm{\left( \nabla_x + \nabla_v \right)  \left< \cdot \right>^a g_s}_{W_0^{b-1,2}(\mathbb{R}^6)}ds
\nonumber \\
&\leq C \left< R_1 \right> \int_0^t  \left( 1 + \norm{f_s}_{W_4^{b-2,2}(\mathbb{R}^6)} + \norm{\alpha_s}_{\mathfrak{h}_b}^2 \right)
\norm{g_s}^2_{W_a^{b,2}(\mathbb{R}^6)} ds.
\end{align}
Collecting the estimates and $\left( 1 - \Delta_z \right)_{\leq n}^{b/2} \leq \left( 1 - \Delta_z \right)^{b/2}$ lead to
\begin{align}
&\norm{\left( 1 - \Delta_z \right)_{\leq n}^{b/2} \left< \cdot \right>^a  \Phi_t(g)}_{L^2(\mathbb{R}^6)}^2
\nonumber \\
&\quad \leq C \norm{g_0}_{W_a^{b,2}(\mathbb{R}^6)}^2
 +   C \left< R_1 \right> \int_0^t  
\left( 1 + \norm{f_s}_{W_4^{b-2,2}(\mathbb{R}^6)} + \norm{\alpha_s}_{\mathfrak{h}_b}^2 \right)
\norm{g_s}_{W_a^{b,2}(\mathbb{R}^6)}^2  ds.
\end{align}
Note that
\begin{align}
\norm{\left( 1 - \Delta_z \right)_{\leq n}^{b/2} f}_{L^2(\mathbb{R}^6)}^2 = \int_0^{\infty}  \frac{(1 + \lambda^2)^b}{(1 + \lambda^2/n^2)^b}d \mu_{f}(\lambda)
\end{align} 
by the spectral calculus for $- \Delta_z$. Using monotone convergence let us obtain
\begin{align}
\lim_{n \rightarrow \infty} \norm{\left( 1 - \Delta_z \right)_{\leq n}^{b/2} f}_{L^2(\mathbb{R}^6)}^2 = \norm{\left( 1 - \Delta_z \right)^{b/2} f}_{L^2(\mathbb{R}^6)}^2 
\end{align}
and
\begin{align}
&\norm{(1 - \Delta_z)^{b/2} \left< \cdot \right>^a  \Phi_t(g)}_{L^2(\mathbb{R}^6)}^2
\nonumber \\
&\quad \leq 
C \norm{g_0}_{W_a^{b,2}(\mathbb{R}^6)}^2
+  C \left< R_1 \right> \int_0^t 
\left( 1 + \norm{f_s}_{W_4^{b-2,2}(\mathbb{R}^6)} + \norm{\alpha_s}_{\mathfrak{h}_b}^2 \right)
\norm{g_s}_{W_a^{b,2}(\mathbb{R}^6)}^2 ds.
\end{align}
If we choose $C$ in \eqref{eq:VM fixed point estimate linear equation f condition on the time parameter} sufficiently large we get
\begin{align}
&\norm{(1 - \Delta)^{b/2}  \left< \cdot \right>^a  \Phi_t(g)}_{L^2(\mathbb{R}^6)}^2
\nonumber \\
&\quad \leq
C \norm{g_0}_{W_a^{b,2}(\mathbb{R}^6)}^2
+  C \left< R_1 \right> \int_0^t 
\left( 1 + \norm{f_s}_{W_4^{b-2,2}(\mathbb{R}^6)} + \norm{\alpha_s}_{\mathfrak{h}_b}^2 \right)
\norm{g_s}_{W_a^{b,2}(\mathbb{R}^6)}^2 ds
\nonumber \\
&\quad \leq
\frac{R_2}{2}
+ C t \left< R_1 \right> \left( 1 + \norm{f}_{L_t^{\infty} W_4^{b-2,2}(\mathbb{R}^6)} + \norm{\alpha}_{L_t^{\infty} \mathfrak{h}_b}^2 \right)  \norm{g}_{L_t^{\infty} W_a^{b,2}(\mathbb{R}^6)}^2
\end{align}
and 
\begin{align}
\label{eq:VM fixed point estimate linear equation final estimate for W-a-b-2}
\norm{\Phi_t(g)}_{W_a^{b,2}(\mathbb{R}^6)}^2
&\leq   C \left( \norm{g_0}_{W_a^{b,2}(\mathbb{R}^6)}^2
+ \left< R_1 \right> \int_0^t  
\left( 1 + \norm{f_s}_{W_4^{b-2,2}(\mathbb{R}^6)} + \norm{\alpha_s}_{\mathfrak{h}_b}^2 \right)
\norm{g_s}_{W_a^{b,2}(\mathbb{R}^6)}^2  ds
\right).
\end{align}
Together with $\norm{g}_{L_T^{\infty} W_a^{b,2}(\mathbb{R}^6)}^2 \leq C R_2$ this shows $\Phi(g) \in L^{\infty} (0,T; W_a^{b,2}(\mathbb{R}^6))$ and
\begin{align}
\norm{(1 - \Delta)^{b/2}  \left< \cdot \right>^a  \Phi_t(g)}_{L_T^{\infty} L^2(\mathbb{R}^6)}^2 \leq R_2
\end{align}
for all $(R_2, T^*)$ satisfying \eqref{eq:VM fixed point estimate linear equation f condition on the time parameter} with $C$ chosen sufficiently large.

\paragraph{Contraction property and proof of \eqref{eq:VM fixed point estimate W-a-b-2 for g-t}:}
Next, we will show that for all $T^*$ satisfying \eqref{eq:VM fixed point estimate linear equation f condition on the time parameter} $\Phi$
is a contraction on $(\mathcal{Z}_{T^*,R_1,R_2}, d_{T^*})$. To this end, consider $g, g' \in \mathcal{Z}_{T^*,R_1,R_2}$ with $g_0 - g'_{0} = g_0 - g_0 = 0$ and note that $\Phi$ is a linear mapping. Performing the substitutions $b \mapsto b-1$ and $g \mapsto (g - g')$ in \eqref{eq:VM fixed point estimate linear equation final estimate for W-a-b-2} then leads to
\begin{align}
\norm{\Phi_t(g) - \Phi_t(g')}_{W_a^{b-1,2}(\mathbb{R}^6)}^2
&\leq   C \left< R_1 \right> t 
\left( 1 + \norm{f}_{L_t^{\infty} W_4^{b-3,2}(\mathbb{R}^6)} + \norm{\alpha}_{L_t^{\infty} \mathfrak{h}_{b-1}}^2 \right)
d_t(g,g')^2  .
\end{align}
For all $T^*$ satisfying \eqref{eq:VM fixed point estimate linear equation f condition on the time parameter} we consequently get
\begin{align}
d_{T^*} \big(\Phi(g), \Phi(g') \big) < d_{T^*} (g,g') ,
\end{align} 
proving the contraction property of $\Phi$.
If we apply the same estimates that lead to \eqref{eq:VM fixed point estimate linear equation final estimate for W-a-b-2} to the integral equations of  \eqref{eq:VM existence theory linear equation for the density} and replace $R_1$ in the estimates by $R(s)$ as defined in Lemma \eqref{lemma:VM existence theory linear equation for the density} we get
\begin{align}
\norm{g_t}_{W_a^{b,2}(\mathbb{R}^6)}^2
&\leq  C \norm{g_0}_{W_a^{b,2}(\mathbb{R}^6)}^2
+ C  \int_0^t  R(s)
\left( 1 + \norm{f_s}_{W_4^{b-2,2}(\mathbb{R}^6)} + \norm{\alpha_s}_{\mathfrak{h}_b}^2 \right)
\norm{g_s}_{W_a^{b,2}(\mathbb{R}^6)}^2 ds.
\end{align}
Inequality \eqref{eq:VM fixed point estimate W-a-b-2 for g-t} the follows by Gr\"{o}nwall's lemma.
\end{proof}

Next, we are looking at the linearized equation for the electromagnetic field.
\begin{lemma} 
\label{lemma:MS well-posedness linearized field equation}
Let $T>0$, $a, b \in \mathbb{N}$ such that $a \geq 5$, $(f, \alpha) \in C \big( 0,T; W_a^{b-1,2}(\mathbb{R}^6) \big) \times  C \big( 0,T ; \mathfrak{h}_{b} \cap \dot{\mathfrak{h}}_{-1/2} \big)$ and $\widetilde{\vJ}_{f_t,  \alpha_t}$ be defined as in \eqref{eq:current density of Vlasov equation}. The linear equation
\begin{align}
\label{eq:MS well-posedness linearized field equation}
i \partial_t \xi_t(k,\lambda) &=  \abs{k}  \xi_t(k,\lambda) -  \sqrt{ \frac{4 \pi^3}{\abs{k}}}  \mathcal{F}[\kappa](k) \vep_{\lambda}(k) \cdot \mathcal{F}[\widetilde{\vJ}_{f_t,  \alpha_t} ](k)  
\end{align}
with initial datum $\xi_0 \in \mathfrak{h}_{b} \cap \dot{\mathfrak{h}}_{-1/2}$ has a unique $C \big( 0,T ; \mathfrak{h}_{b} \cap \dot{\mathfrak{h}}_{-1/2} \big) \cap C^1 \big( 0,T ; \mathfrak{h}_{b-1} \cap \dot{\mathfrak{h}}_{-1/2} \big)$ solution. 
\end{lemma}

\begin{proof}
Note that
\begin{align}
\label{eq:VM fixed point estimate difference of the currents in L-1 norm}
&\norm{\widetilde{\vJ}_{f_t,  \alpha_t} - \widetilde{\vJ}_{f_s,  \alpha_s}}_{W_0^{b-1, 1}(\mathbb{R}^3, \mathbb{C}^3)}
\nonumber \\
&\quad \leq C \left( 1 + \norm{\alpha_t}_{\mathfrak{h}_{b-2}} \right) \norm{f_t - f_s}_{W_5^{b-1,2}(\mathbb{R}^6)}
+ C \norm{f_s}_{W_4^{b-1,2}(\mathbb{R}^6)} \norm{\alpha_t - \alpha_s}_{\mathfrak{h}_{b-2}}
\nonumber \\
&\quad \leq C \left( 1 + \norm{\alpha_t}_{\mathfrak{h}_{b-2}} + \norm{f_s}_{W_4^{b-1,2}(\mathbb{R}^6)} \right) 
\left( \norm{f_t - f_s}_{W_5^{b-1,2}(\mathbb{R}^6)} 
+ \norm{\alpha_t - \alpha_s}_{\mathfrak{h}_{b-2}} \right)
\end{align}
because of \eqref{eq:estimate current density Vlasov Maxwell W-0-sigma-1 norm}, \eqref{eq:estimates for the vector potential and interaction 1} and $\norm{f}_{W_0^{b-1,1}(\mathbb{R}^6)} \leq C \norm{f}_{W_4^{b-1,2}(\mathbb{R}^6)}$.
By means of \eqref{eq:assumption on the cutoff parameter} we have
\begin{align}
\norm{\abs{\cdot}^{-1/2} \mathcal{F}[\kappa] \, \vep \, \mathcal{F}[\widetilde{\vJ}_{f_t,  \alpha_t} - \widetilde{\vJ}_{f_s,  \alpha_s} ]}_{\mathfrak{h}_b \cap \, \dot{\mathfrak{h}}_{-1/2}}
&\leq C
\norm{ \left< \cdot \right>^{b-1} \mathcal{F}[\widetilde{\vJ}_{f_t,  \alpha_t} - \widetilde{\vJ}_{f_s,  \alpha_s} ]}_{L^{\infty}(\mathbb{R}^3)}
\nonumber \\
&\leq C \norm{\widetilde{\vJ}_{f_t,  \alpha_t} - \widetilde{\vJ}_{f_s,  \alpha_s}}_{W_0^{b-1, 1}(\mathbb{R}^3, \mathbb{C}^3)} .
\end{align}
This and the continuity of $(f, \alpha)$ let us conclude that $e^{i \abs{k} s} \sqrt{ \frac{4 \pi^3}{\abs{k}}} \mathcal{F}[\kappa](k) \vep_{\lambda}(k) \mathcal{F}[\widetilde{\vJ}_{f_s, \alpha_s}](k)$ is a $C \big( 0,T ; \mathfrak{h}_{b} \cap \dot{\mathfrak{h}}_{-1/2} \big)$ function. The Lemma can then be proven similarly as Lemma \ref{lemma:MS well-posedness linearized field equation}.
\end{proof}

\subsection{Local solutions}

\begin{lemma}
\label{lemma:VM local well-posedness}
Let $R>0$ and $a, b \in \mathbb{N}$ such that $a \geq 5$ and $b \geq 3$. For all $(f_0, \alpha_0) \in W_a^{b,2}(\mathbb{R}^6) \times \mathfrak{h}_b \cap \dot{\mathfrak{h}}_{-1/2}$ such that $\supp f_0 \subset A_R$ with $A_R = \{ (x,v) \in \mathbb{R}^6 , \abs{v} \leq R \}$ there exists $T>0$ and a unique
$L^{\infty} \big( 0,T; W_{a}^{b,2}(\mathbb{R}^6) \big) \cap C \big( 0,T; W_{a}^{b-1,2}(\mathbb{R}^6) \big)  \cap C^1 (0,T; W_{a}^{b-2,2}(\mathbb{R}^6)) \times  C \big( 0,T ; \mathfrak{h}_{b} \cap \dot{\mathfrak{h}}_{-1/2} \big) \cap C^1 \big( 0,T ; \mathfrak{h}_{b-1} \cap \dot{\mathfrak{h}}_{-1/2} \big)$--valued function which satisfies \eqref{eq: Vlasov-Maxwell different style of writing} in $W_a^{b-2,2}(\mathbb{R}^6) \oplus \mathfrak{h}_{b-1} \cap \dot{\mathfrak{h}}_{-1/2}$ with initial datum $(f_0, \alpha_0)$.
\end{lemma}

\begin{proof}
Note that
\begin{align}
\mathcal{Z}_{T, R^*} &= \Big\{ (f, \alpha) \in C \big( 0,T; W_a^{b-1,2}(\mathbb{R}^6) \big)  \times  C \big( 0,T ; \mathfrak{h}_{b} \cap \dot{\mathfrak{h}}_{-1/2} \big) : (f(t),\alpha(t)) \big|_{t=0} = (f_0, \alpha_0) ,
\nonumber \\
&\qquad  
\max \Big\{ \norm{f}_{L_T^{\infty} W_a^{b-1,2}(\mathbb{R}^6)} , \norm{\alpha}_{L_T^{\infty} \mathfrak{h}_b \, \cap \, L_T^{\infty} \dot{\mathfrak{h}}_{-1/2}}
\Big\} \leq R^*
\Big\}
\end{align}
with metric
\begin{align}
d_T ((f, \alpha), (f', \alpha')) 
&= \max \left\{ \norm{f - f'}_{L_T^{\infty} W_a^{b-1,2}(\mathbb{R}^6)} , \norm{\alpha - \alpha'}_{L_T^{\infty} \mathfrak{h}_b \, \cap \, L_T^{\infty} \dot{\mathfrak{h}}_{-1/2}}
\right\}
\end{align}
is a Banach space. Next, we consider the solutions of Section \ref{subsection:VM existence theory linearized equations} satisfying
\begin{align}
\label{eq:VM fixed point estimates linearized equations for contraction mapping principle}
\begin{cases}
\partial_t g_t  &= - 2 \left( v - \kappa * \vA_{\alpha_t} \right) \cdot \nabla_x g_t + \vF_{f_t, \alpha_t} \cdot \nabla_v g_t , \\
i \partial_t \xi_t(k,\lambda) &=  \abs{k}  \xi_t(k,\lambda) -  \sqrt{ \frac{4 \pi^3}{\abs{k}}}  \mathcal{F}[\kappa](k) \vep_{\lambda}(k) \cdot \mathcal{F}[\widetilde{\vJ}_{f_t,  \alpha_t} ](k)  
\end{cases}
\end{align}
with initial datum $(f_0, \alpha_0)$ and define the mapping $\Phi: \mathcal{Z}_{R^*,T} \rightarrow L^{\infty} \big( 0,T; W_{a}^{b,2}(\mathbb{R}^6) \big) \cap C \big( 0,T; W_{a}^{b-1,2}(\mathbb{R}^6) \big)  \cap C^1 (0,T; W_{a}^{b-2,2}(\mathbb{R}^6)) \times  C \big( 0,T ; \mathfrak{h}_{b} \cap \dot{\mathfrak{h}}_{-1/2} \big) \cap C^1 \big( 0,T ; \mathfrak{h}_{b-1} \cap \dot{\mathfrak{h}}_{-1/2} \big)$ by

\begin{align}
\label{eq:VM contraction map}
\Phi \begin{pmatrix}
f \\ \alpha
\end{pmatrix}
&= 
\begin{pmatrix}
g \\ \xi
\end{pmatrix} .
\end{align}
Below, we will show that for 
\begin{align}
\label{eq:VM fixed point estimate restriction on the time interval}
R^* \geq C \left( \norm{f_0}_{W_a^{b-1,2}} + \norm{\alpha_0}_{\mathfrak{h}_b \, \cap \,\dot{\mathfrak{h}}_{-1/2} } \right)
\quad \text{and} \quad
T \leq \exp \left( - C \left( \left< R \right> + \left< R^* \right>^2 + \norm{f_0}_{W_a^{b,2}(\mathbb{R}^6)} \right) \right)
\end{align}
with $C>0$ chosen sufficiently large the mapping $\Phi$ is a contraction on $\mathcal{Z}_{T, R^*}$. This leads to a unique fixed point satisfying $d_T ((f, \alpha), (g, \xi))  = 0$. Replacing $(g, \xi)$ by $(f, \alpha)$ in the integral version of \eqref{eq:VM fixed point estimates linearized equations for contraction mapping principle} and then proves that $(f, \alpha)$ satisfies \eqref{eq: Vlasov-Maxwell different style of writing}.

\paragraph{Well-definedness of $\Phi$:}
By means of \eqref{eq:VM fixed point estimate W-a-b-2 for g-t} and $(f,\alpha) \in \mathcal{Z}_{T,R^*}$ we get
\begin{align}
\norm{g_t}_{W_a^{b-1,2}(\mathbb{R}^6)}
&\leq C  \norm{f_0}_{W_a^{b-1,2}(\mathbb{R}^6)}
e^{C t \left< R^* \right>^2 \left< R \right> e^{C t \left< R^* \right>^2}}.
\end{align}
Using Duhamel's formula, \eqref{eq:assumption on the cutoff parameter}  and \eqref{eq:estimate current density Vlasov Maxwell W-0-sigma-1 norm} we get by similar estimates as in the proof of Lemma \eqref{lemma:MS well-posedness linearized field equation}
\begin{align}
\norm{\xi_t}_{\mathfrak{h}_{b} \, \cap \, \dot{\mathfrak{h}}_{-1/2}}
&\leq C
\norm{\alpha_0}_{\mathfrak{h}_{b} \, \cap \, \dot{\mathfrak{h}}_{-1/2}}
+ C  \int_0^t   \norm{\widetilde{\vJ}_{f_s, \alpha_s} }_{W_0^{b-1, 1}(\mathbb{R}^3, \mathbb{C}^3)}ds
\nonumber \\
&\leq C
\norm{\alpha_0}_{\mathfrak{h}_{b} \, \cap \, \dot{\mathfrak{h}}_{-1/2}}
+ C t  \left( 1 + \norm{\alpha}_{L_t^{\infty} \mathfrak{h}_{b -2}} \right) \norm{f}_{L_t^{\infty} W_5^{b-1,2}(\mathbb{R}^6)}
\nonumber \\
&\leq C
\norm{\alpha_0}_{\mathfrak{h}_{b} \, \cap \, \dot{\mathfrak{h}}_{-1/2}}
+ C t  \left< R^* \right>^2 .
\end{align}
For $C$ in \eqref{eq:VM fixed point estimate restriction on the time interval} chosen sufficiently large this shows $\norm{g}_{L_T^{\infty} W_a^{b-1,2}(\mathbb{R}^6)} \leq R^*$ and
$\norm{\xi}_{L_T^{\infty} \mathfrak{h}_b \, \cap \, L_T^{\infty} \dot{\mathfrak{h}}_{-1/2}}
\leq R^*$, proving that $\Phi$ maps $\mathcal{Z}_{T,R^*}$ into itself.

\paragraph{Contraction property of $\Phi$:}

Let $(f, \alpha)$, $(f', \alpha') \in \mathcal{Z}_{T, R^*}$ and denote their images under the mapping $\Phi$ by $(g, \xi)$ and $(g', \xi')$. In the following, we prove $d_T ((g, \xi), (g', \xi'))  < d_T ((f, \alpha), (f', \alpha'))$ for all $T \geq 0$ satisfying \eqref{eq:VM fixed point estimate restriction on the time interval} with sufficiently large constant.
Using Duhamel's formula, \eqref{eq:assumption on the cutoff parameter} and \eqref{eq:VM fixed point estimate difference of the currents in L-1 norm}  we get
\begin{align}
\label{eq:VM well-posedness contraction estimate for the field}
&\norm{\xi - \xi'}_{L_T^{\infty} \mathfrak{h}_{b} \, \cap \, L_T^{\infty} \dot{\mathfrak{h}}_{-1/2}}
\nonumber \\
&\quad \leq C \sup_{t \in [0,T]}
 \int_0^t  \norm{ \abs{\cdot}^{-1} \left< \cdot \right>^{3/2}   \mathcal{F}[\kappa]}_{L^2(\mathbb{R}^3)} \norm{\left< \cdot \right>^{b-1} \mathcal{F}[\widetilde{\vJ}_{f_s, \alpha_s} - \widetilde{\vJ}_{f_s', \alpha_s'}]}_{L^{\infty}(\mathbb{R}^3, \mathbb{C}^3)}ds
\nonumber \\
&\quad \leq C \sup_{t \in [0,T]}
 \int_0^t   \norm{\widetilde{\vJ}_{f_s, \alpha_s} - \widetilde{\vJ}_{f_s', \alpha_s'}}_{W_0^{b-1, 1}(\mathbb{R}^3, \mathbb{C}^3)}ds
\nonumber \\
&\leq  C T \left( 1 + \norm{\alpha}_{L_T^{\infty} \mathfrak{h}_{b-2}} + \norm{f'}_{L_T^{\infty} W_4^{b-1,2}(\mathbb{R}^6)} \right) 
d_T ((f, \alpha), (f', \alpha'))
\nonumber \\
&\leq C T  \left< R^* \right>
\left( \norm{f - f'}_{L_T^{\infty} W_5^{b-1,2}(\mathbb{R}^6)} 
+ \norm{\alpha - \alpha'}_{L_T^{\infty} \mathfrak{h}_{b-2}} \right)
\nonumber \\
&< d_T ((f, \alpha), (f', \alpha'))  .
\end{align}
By means of 
\begin{align}
 D_z^{\sigma} \left( g_t - g'_t \right)
&= 
- \int_0^t   T_{f_s' , \alpha_s'}[D_z^{\sigma} (g_s - g_s' )]\,ds
- \int_0^t  \big[ D_z^{\sigma} , T_{f_s' , \alpha_s'} \big] [g_s - g_s']\,ds
\nonumber \\
&\quad - \int_0^t  D_z^{\sigma} \big(  T_{f_s , \alpha_s} - T_{f_s' , \alpha_s'} \big) [g_s]\,ds
\end{align}
with $T_{f,\alpha}$ being defined as in \eqref{eq:definition transport operator} we get for $0 \leq \abs{\sigma} \leq b-1$ and $t \in [0,T]$
\begin{subequations}
\begin{align}
&\norm{ D_z^{\sigma} (g_t - g'_t)}_{W_a^{0,2}}^2
\nonumber \\
\label{eq:VM fixed point estimate W-a-(b-1)-2 estimate difference g and g' 1}
&\quad = - 2 \int_0^t  \int_{\mathbb{R}^6}  \left< z \right>^{2a} (D_z^{\sigma} (g_s - g'_s))(x,v)  \big( T_{f_s' , \alpha_s'}[D_z^{\sigma} (g_s - g_s' )] \big)(x,v)\,dx\,dv\,ds
\\
\label{eq:VM fixed point estimate W-a-(b-1)-2 estimate difference g and g' 2}
&\qquad - 2 \int_0^t  \int_{\mathbb{R}^6}  \left< z \right>^{2a} (D_z^{\sigma} (g_s - g'_s))(x,v)  \big( \big[ D_z^{\sigma} , T_{f_s' , \alpha_s'} \big] [g_s - g_s'] \big)(x,v)\,dx\,dv\,ds
\\
\label{eq:VM fixed point estimate W-a-(b-1)-2 estimate difference g and g' 3}
&\qquad - 2 \int_0^t  \int_{\mathbb{R}^6}  \left< z \right>^{2a} (D_z^{\sigma} (g_s - g'_s))(x,v)  \big(  D_z^{\sigma} \big(  T_{f_s , \alpha_s} - T_{f_s' , \alpha_s'} \big) [g_s] \big)(x,v)\,dx\,dv\,ds.
\end{align}
\end{subequations}
Using \eqref{eq:VM fixed point argument integration by parts of the transport operator}, \eqref{eq:commutator of transport operator with weight function}, \eqref{eq:estimates for the vector potential and interaction 1}  and \eqref{eq:estimates for the vector potential and interaction 2}  we estimate 
\begin{align}
\abs{\eqref{eq:VM fixed point estimate W-a-(b-1)-2 estimate difference g and g' 1}}
&= \abs{  \int_0^t  \int_{\mathbb{R}^6}  T_{f'_s, \alpha'_s}\big[ \left< z \right>^{2a} \big] \abs{(D_z^{\sigma} (g_s - g'_s))(x,v)}^2 }dx\,dv\,ds
\nonumber \\
&\leq   C  \int_0^t   \left( 1 + \norm{f'_s}_{W_4^{0,2}(\mathbb{R}^6)} + \norm{\alpha'_s}_{\mathfrak{h}}^2 \right)
\int_{\mathbb{R}^6}   \left< z \right>^{2a} \abs{(D_z^{\sigma} (g_s - g'_s))(x,v)}^2 dx\,dv\,ds
\nonumber \\
&\leq   C t \left( 1 + \norm{f'}_{L_t^{\infty} W_4^{0,2}(\mathbb{R}^6)} + \norm{\alpha'}_{L_t^{\infty} \mathfrak{h}}^2 \right)
\norm{g - g'}_{ L_t^{\infty} W_a^{b-1,2}(\mathbb{R}^6)}^2 
\nonumber \\
&\leq C t \left< R^* \right>^2  \norm{g - g'}_{ L_t^{\infty} W_a^{b-1,2}(\mathbb{R}^6)}^2 .
\end{align}
Due to \eqref{eq:estimates for the vector potential and interaction 1}  and \eqref{eq:estimates for the vector potential and interaction 2} and the chain rule of differentiation we have
\begin{align}
&\norm{ \left< v \right> ^{-1} \left< z \right> ^a \big[ D_z^{\sigma} , T_{f_s' , \alpha_s'} \big] [g_s - g_s'] }_{L^2(\mathbb{R}^6)}
\nonumber \\
&\quad \leq C \left( 1 + \norm{f'_s}_{W_4^{b-3,2}(\mathbb{R}^6)} + \norm{\alpha'_s}_{\mathfrak{h}_{b-1}}^2 \right) \norm{g_s - g'_s}_{ W_a^{b-1,2}(\mathbb{R}^6)}  
\nonumber \\
&\quad \leq C \left< R^* \right>^2  \norm{g_s - g'_s}_{ W_a^{b-1,2}(\mathbb{R}^6)}  .
\end{align}
Note that
\begin{align}
\supp (g_s - g'_s) \subset A_{R(s)}
\quad \text{with} \quad
R(s) &\leq C \left< R \right> e^{C s \left< R^* \right>^2}
\end{align}
because of Lemma \ref{lemma:VM existence theory linear equation for the density} and $(f,\alpha), (f',\alpha') \in \mathcal{Z}_{T,R^*}$.
Since differentiation with respect to $D_z^{\sigma}$ is a local operation we get
\begin{align}
\label{eq:VM fixed point estimate W-a-(b-1)-2 estimate auxilliary estimate for compact support}
\norm{ \left< v \right> \left< z \right> ^a  D_z^{\sigma} (g_s - g_s')}_{L^2(\mathbb{R}^6)}
&\leq R(s)
\norm{\left< z \right> ^a  D_z^{\sigma} (g_s - g_s')}_{L^2(\mathbb{R}^6)} .
\end{align}
Together with the Cauchy--Schwarz inequality this leads to 
\begin{align}
\abs{\eqref{eq:VM fixed point estimate W-a-(b-1)-2 estimate difference g and g' 2}}
&\leq C t  \left< R \right> \left< R^* \right>^2
e^{C t \left< R^* \right>^2}
\norm{g - g'}_{ L_t^{\infty} W_a^{b-1,2}(\mathbb{R}^6)}^2 .
\end{align}
Recall that
\begin{align}
T_{f, \alpha} - T_{f', \alpha'}
&= - 2 \kappa * \vA_{\alpha - \alpha'} \cdot \nabla_x - \big( \vF_{f, \alpha} - \vF_{f' , \alpha'} \big) \cdot \nabla_v
\end{align}
and
\begin{align}
\vF_{f, \alpha} - \vF_{f' , \alpha'}
&= 
\nabla K  * \widetilde{\rho}_{f - f'} (x)
 - 2 \sum_{i=1}^3 ( \nabla  \kappa * \vA^i_{\alpha - \alpha'} ) (x)
\left( v^i - \kappa * \vA^i_{\alpha}(x) \right)
\nonumber \\
&\quad +
 2 \sum_{i=1}^3 ( \nabla  \kappa * \vA^i_{\alpha'} ) (x) \kappa * \vA^i_{\alpha - \alpha'}(x) .
\end{align}
Using \eqref{eq:estimates for the vector potential and interaction 1}  and \eqref{eq:estimates for the vector potential and interaction 2} and the chain rule we get
\begin{align}
\abs{\big( D_z^{\sigma} \big(  T_{f_s , \alpha_s} - T_{f_s' , \alpha_s'} \big) [g_s] \big) (x,v)}
&\leq C  \left< v \right>
\left( \norm{\alpha_s - \alpha'_s}_{\mathfrak{h}_{b-1}}
+ \norm{f_s - f'_s}_{W_a^{b-3,2}}
\right)
\nonumber \\
&\quad \times 
\left( 1 + \norm{\alpha_s}_{\mathfrak{h}_{b-1}} + \norm{\alpha'_s}_{\mathfrak{h}_{b-1}} \right)
\sum_{\abs{\nu} \leq b}
\abs{D_z^{\nu} g_s(x,v)} 
\end{align}
and
\begin{align}
&\norm{\left< v \right>^{-1} \left< z \right>^a D_z^{\sigma} \big(  T_{f_s , \alpha_s} - T_{f_s' , \alpha_s'} \big) [ g_s]}_{L^2(\mathbb{R}^6)} 
\nonumber \\
&\quad \leq  C 
\left( 1 + \norm{\alpha_s}_{\mathfrak{h}_{b-1}} + \norm{\alpha'_s}_{\mathfrak{h}_{b-1}} \right) \norm{g_t}_{W_a^{b,2}(\mathbb{R}^6)} 
\left( \norm{\alpha_s - \alpha'_s}_{\mathfrak{h}_{b-1}}
+ \norm{f_s - f'_s}_{W_a^{b-3,2}(\mathbb{R}^6)}
\right) .
\end{align}
By means of the Cauchy--Schwarz inequality and \eqref{eq:VM fixed point estimate W-a-(b-1)-2 estimate auxilliary estimate for compact support}
we obtain
\begin{align}
\abs{\eqref{eq:VM fixed point estimate W-a-(b-1)-2 estimate difference g and g' 3}}
&\leq 2 \int_0^t  \norm{ \left< v \right> \left< z \right> ^a  D_z^{\sigma} (g_s - g_s')}_{L^2(\mathbb{R}^6)}
\norm{\left< v \right>^{-1} \left< z \right>^a D_z^{\sigma} \big(  T_{f_s , \alpha_s} - T_{f_s' , \alpha_s'} \big) [g_s]}_{L^2(\mathbb{R}^6)} ds
\nonumber \\
&\leq C  t \, R(t)
\left( 1 + \norm{\alpha}^2_{L_t^{\infty} \mathfrak{h}_{b-1}} + \norm{\alpha'}^2_{L_t^{\infty} \mathfrak{h}_{b-1}} \right) \norm{g}_{L_t^{\infty} W_a^{b,2}(\mathbb{R}^6)} 
\nonumber \\
&\quad \times 
\norm{g - g'}_{ L_t^{\infty} W_a^{b-1,2}(\mathbb{R}^6)}
\left( \norm{\alpha - \alpha'}_{L_t^{\infty} \mathfrak{h}_{b-1}}
+ \norm{f - f'}_{L_t^{\infty} W_a^{b-3,2}(\mathbb{R}^6)}
\right) .
\end{align}
In combination with \eqref{eq:VM fixed point estimate W-a-b-2 for g-t} this leads to
\begin{align}
\abs{\eqref{eq:VM fixed point estimate W-a-(b-1)-2 estimate difference g and g' 3}}
&\leq C  t \, R(t) 
\left( 1 + \norm{\alpha}^2_{L_t^{\infty} \mathfrak{h}_{b-1}} + \norm{\alpha'}^2_{L_t^{\infty} \mathfrak{h}_{b-1}} \right) 
\norm{g_0}_{W_a^{b,2}}
\nonumber \\
&\quad \times 
e^{C \int_0^t    R(s) \left( 1 + \norm{f_s}_{W_4^{b-2,2}(\mathbb{R}^6)} + \norm{\alpha_s}_{\mathfrak{h}_b}^2 \right)ds }
\norm{g - g'}_{ L_t^{\infty} W_a^{b-1,2}(\mathbb{R}^6)}
d_t((f,\alpha), (f', \alpha')) 
\nonumber \\
&\leq C t \left< R^* \right>^2 e^{2 \left< R \right> e^{C t \left< R^* \right>^2}}
\norm{g_0}_{W_a^{b,2}}
\norm{g - g'}_{ L_t^{\infty} W_a^{b-1,2}(\mathbb{R}^6)}
d_t((f,\alpha), (f', \alpha')) .
\end{align}
Collecting the estimates, taking the sum $\sum_{\abs{\sigma} \leq b-1}$ and dividing by $\norm{g - g'}_{ L_t^{\infty} W_a^{b-1,2}(\mathbb{R}^6)}$ gives
\begin{align}
\norm{g - g'}_{ L_t^{\infty} W_a^{b-1,2}(\mathbb{R}^6)}
&\leq C t  \left< R \right> \left< R^* \right>^2
e^{C t \left< R^* \right>^2}
\norm{g - g'}_{ L_t^{\infty} W_a^{b-1,2}(\mathbb{R}^6)} 
\nonumber \\
&\quad + C t \left< R^* \right>^2 e^{2 \left< R \right> e^{C t \left< R^* \right>^2}}
\norm{g_0}_{W_a^{b,2}}
d_t((f,\alpha), (f', \alpha'))
\nonumber \\
&\leq  t e^{C \left( \left< R \right> + \left< R^* \right>^2 \right)  e^{C t \left< R^* \right>^2}}
\left( 1 +  \norm{g_0}_{W_a^{b,2}}  \right)
\nonumber \\
&\quad \times
\left( \norm{g - g'}_{ L_t^{\infty} W_a^{b-1,2}(\mathbb{R}^6)} + d_t((f,\alpha), (f', \alpha')) \right) .
\end{align}
For $T \geq 0$ satisfying
\eqref{eq:VM fixed point estimate restriction on the time interval} with $C> 0$ chosen large enough
we then obtain
\begin{align}
\norm{g - g'}_{L_T^{\infty} W_a^{b-1,2}}
&\leq  \frac{1}{2}  \norm{g - g'}_{L_T^{\infty} W_a^{b-1,2}}
+ \frac{1}{4} d_T((f, \alpha), (f', \alpha'))
\end{align}
and therefore
$\norm{g - g'}_{L_T^{\infty} W_a^{b-1,2}}
< d_T((f, \alpha), (f', \alpha'))$.

\end{proof}

\subsection{Global solutions}


\begin{proof}[Proof of Proposition~\ref{lemma:global solutions VM}]
The existence of a unique local solution is ensured by Lemma \ref{lemma:VM local well-posedness}. By similar means as in the proof of \cite[Theorem 4.3.4]{CH1998} it is straightforward to prove that $T_{\text{max}}$ is either infinite or $\lim_{t \nearrow T_{\text{max}}} \left( \norm{f_t}_{W_a^{b,2}} + \norm{\alpha_t}_{\mathfrak{h}_b \, \cap \,\dot{\mathfrak{h}}_{-1/2} } \right) = \infty$.
In the following we assume that the local solution of Lemma \ref{lemma:VM local well-posedness} exists until time $T$ and show $\norm{f_t}_{W_a^{b,2}} + \norm{\alpha_t}_{\mathfrak{h}_b \, \cap \,\dot{\mathfrak{h}}_{-1/2} } < + \infty$ for all $t \in [0,T]$. On the interval of existence the conservation of the energy follows from a straightforward calculation. If we define the characteristics $(X_t(x,v), V_t(x,v))$ by \eqref{eq:characteristics} with $(f_t, \alpha_t)$ satisfying \eqref{eq: Vlasov-Maxwell different style of writing} and initial datum $(X_t(x,v), V_t(x,v)) \big|_{t=0} = (x,v)$ we can write the particle distribution in terms of the corresponding backward characteristics as $f_t(x,v) = f_0 \big( X_t^{-1}(x,v) , V_t^{-1}(x,v)) \big)$. The conservation of the $L^p$--norms then follows from a change of coordinates and the fact that the flow is measure preserving. Next, we will show the finiteness of the $W_a^{b,2}(\mathbb{R}^6)$ and $\mathfrak{h}_b \, \cap \,\dot{\mathfrak{h}}_{-1/2}$ norms. By means of \eqref{eq:characteristics} and \eqref{eq:estimate F function supremum in x} and the conservation of mass and energy we get
\begin{align}
\left< V_t(x,v) \right> &\leq C \left< v \right> 
+ C \left( 1 + \mathcal{E}^{\rm{VM}}[f_0, \alpha_0] 
+  C \norm{f}_{L^1(\mathbb{R}^6)}^2  \right)
 \int_0^{t}  \left< V_s(x,v) \right> ds,
\end{align}
leading to 
\begin{align}
\label{eq:characteristics L-infinity bound for the velocity}
\left< V_t(x,v) \right>
&\leq C \left<v \right> e^{C \left( 1 + \mathcal{E}^{\rm{VM}}[f_0 , \alpha_0] 
+ C \norm{f_0}_{L^1(\mathbb{R}^6)}^2  \right) t}
\end{align}
by Gr\"{o}nwall's lemma. The first equation of \eqref{eq:characteristics}, \eqref{eq:estimates for the vector potential and interaction 0}, \eqref{eq:estimate for the VM energy functional 1}, the conservation of mass and energy as well as \eqref{eq:characteristics L-infinity bound for the velocity} let us obtain
\begin{align}
\abs{X_t(x,v)}
&\leq \abs{x} + C \left( 1 + \mathcal{E}^{\rm{VM}}[f_0, \alpha_0] 
+ C \norm{f_0}_{L^1(\mathbb{R}^6)}^2  \right)
\left( \abs{t} +  \int_0^{t}  \abs{V_s(x,v)} ds\right)
\nonumber \\
&\leq \left< (x,v) \right> 
e^{C \left( 1 + \mathcal{E}^{\rm{VM}}[f_0, \alpha_0] 
+ C \norm{f_0}_{L^1(\mathbb{R}^6)}^2  \right) \left< t \right>} .
\end{align}
If we use the previous estimates and the fact that the flow is measure preserving we get 
\begin{align}
\norm{f_t}_{W_a^{0,2}(\mathbb{R}^6)}^2
&= \int_{\mathbb{R}^6}  \left< (x,v) \right>^{2a} \abs{f_0 \big( X_t^{-1}(x,v) , V_t^{-1}(x,v)) \big)}^2 dx\,dv
\nonumber \\
&= 
\int_{\mathbb{R}^6} \left< \big( X_t(x,v) , V_t(x,v) \big) \right>^{2a} \abs{f_0 (x,v)}^2 dx\,dv
\nonumber \\
&\leq e^{C \left( 1 + \mathcal{E}^{\rm{VM}}[f_0, \alpha_0] 
+ C \norm{f_0}_{L^1(\mathbb{R}^6)}^2  \right) \left< t \right>}
\norm{f_0}_{W_a^{0,2}(\mathbb{R}^6)}^2
\end{align}
by a change of coordinates. Due to \eqref{eq: Vlasov-Maxwell different style of writing}, Duhamel's formula and \eqref{eq:assumption on the cutoff parameter} we have
\begin{align}
\norm{\alpha_t}_{\mathfrak{h}_{b} \, \cap \,  \dot{\mathfrak{h}}_{-1/2}}
&\leq \norm{\alpha_0}_{\mathfrak{h}_{b} \, \cap \,  \dot{\mathfrak{h}}_{-1/2}}
+ C \int_0^t \norm{\widetilde{\vJ}_{f_s, \alpha_s} }_{W_0^{b-1, 1}(\mathbb{R}^3, \mathbb{C}^3)} ds.
\end{align}
By \eqref{eq:estimate current density Vlasov Maxwell L-1 norm}, 
\eqref{eq:estimate current density Vlasov Maxwell W-0-sigma-1 norm} and the conservation of the $L^p$--norms and energy we obtain
\begin{align}
\label{eq:propagation of gradient moments estimate for alpha-t in h-1}
\norm{\alpha_t}_{\mathfrak{h}_1 \, \cap \,  \dot{\mathfrak{h}}_{-1/2}}
&\leq 
\norm{\alpha_0}_{\mathfrak{h}_1 \, \cap \,  \dot{\mathfrak{h}}_{-1/2}} + C t \left( 1 + \mathcal{E}^{\rm{VM}}[f_0, \alpha_0] 
+  C \norm{f_0}_{L^1(\mathbb{R}^6)}^2  \right)
\quad \text{for} \;  b = 1
\end{align}
and 
\begin{align}
\label{eq:propagation of gradient moments estimate for alpha-t in h-b}
\norm{\alpha_t}_{\mathfrak{h}_{b} \, \cap \,  \dot{\mathfrak{h}}_{-1/2}}
&\leq \norm{\alpha_0}_{\mathfrak{h}_{b} \, \cap \,  \dot{\mathfrak{h}}_{-1/2}}
+ C t  \left( 1 + \norm{\alpha}_{L_t^{\infty}\mathfrak{h}_{b-2}} \right) \norm{f}_{L_t^{\infty} W_5^{b-1,2}(\mathbb{R}^6)}
\quad \text{for} \;  b \geq 2.
\end{align}
We consequently have 
\begin{align}
\norm{f}_{L_T^{\infty} W_a^{0,2}(\mathbb{R}^6)}
+ \norm{\alpha}_{L_T^{\infty} \mathfrak{h}_1 \, \cap \, L_T^{\infty} \dot{\mathfrak{h}}_{-1/2}}
< + \infty .
\end{align}
The second inequality of \eqref{eq:propagation of gradient moments estimate for alpha-t in h-b} and (see \eqref{eq:VM fixed point estimate W-a-b-2 for g-t})
\begin{align}
\norm{f_t}_{W_a^{b,2}(\mathbb{R}^6)} &\leq 
C \norm{f_0}_{W_a^{b,2}}
e^{C \int_0^t    R(s) \left( 1 + \norm{f_s}_{W_4^{b-2,2}(\mathbb{R}^6)} + \norm{\alpha_s}_{\mathfrak{h}_b}^2 \right) ds}  
\end{align}
with $R(t)  \leq C \left< R \right> e^{C  \int_0^t  \left( \norm{f_s}_{W_4^{0,2}(\mathbb{R}^6)} + \norm{\alpha_s}_{ \mathfrak{h}}^2 \right)ds}$ let us then iteratively increase the values of $\beta$ until we get $\sup_{t \in [0,T]} \left\{ \norm{f_t}_{W_a^{b,2}} + \norm{\alpha_t}_{\mathfrak{h}_b \, \cap \,\dot{\mathfrak{h}}_{-1/2} } \right\} < + \infty$.

\end{proof}

\appendix

\section{Equivalent forms of the Vlasov--Maxwell equations}\label{section:comparison of the different Vlasov-Maxwell equations}


Let $a \in \mathbb{R}$ and $(f, \alpha)$ satisfy \eqref{eq: Vlasov-Maxwell different style of writing} with $\kappa (x) = a \delta(x)$. We define 
\begin{align}
\vE^{\perp}_{\alpha_t}(x) 
&=  \frac{i}{(2 \pi)^{3/2}} \sum_{\lambda = 1,2} \int  \sqrt{\frac{\abs{k}}{2}} \vep_{\lambda}(k) \left( e^{i k x} \alpha(k,\lambda) - e^{- i k x} \overline{\alpha(k, \lambda)} \right)dk
\end{align}
and $g_t(x,v) = f_t(x, v + a \vA_{\alpha_t}(x) )$ satisfying $\widetilde{\rho}_{f_t} =  \widetilde{\rho}_{g_t}$ and $\widetilde{\vJ}_{f_t, \alpha_t}(x) = 2 \int_{\mathbb{R}^3}  v g_t(x,v)dv$. Using
\begin{align}
&\vF_{f_t, \alpha_t}(x,  v + a \vA_{\alpha_t}(x)) + 2 \sum_{j=1}^3 v^j  \nabla^j \vA_{\alpha_t}(x)
\nonumber \\
&\quad = c \nabla | \cdot |^{-1} * \widetilde{\rho}_{f_t}(x) + 2 a \sum_{j=1}^3 \left( v^j \nabla^j \vA_{\alpha_t}(x) - v^j \nabla \vA^i_{\alpha_t}(x) \right)
\nonumber \\
&\quad =  c \nabla | \cdot |^{-1} * \widetilde{\rho}_{f_t}(x) - 2 a v \times \left( \nabla \times \vA_{\alpha_t} \right)(x)
\end{align}
and \eqref{eq: Vlasov-Maxwell different style of writing} we obtain
\begin{align}
\begin{cases}
\partial_t g_t
&= - 2 v \cdot \nabla_x g_t 
+ \left(  - a   \vE^{\perp}_{\alpha_t} + \frac{a^2}{8 \pi}   \nabla | \cdot |^{-1} * \widetilde{\rho}_{g_t} - 2 a v \times \left( \nabla \times \vA_{\alpha_t} \right) (x) \right) \cdot \nabla_v g_t
\\
\partial_t \vA_{\alpha_t} 
&= - \vE^{\perp}_{\alpha_t}
\\
\partial_t \vE^{\perp}_{\alpha_t} 
&=  - \Delta \vA_{\alpha_t} - 2 a \left( 1 - \nabla \text{div} \Delta^{-1} \right)  \int_{\mathbb{R}^3}  v g_t(\cdot,v) dv
\end{cases}
\end{align}
Choosing $a = - \frac{1}{4}$, defining $h_t(x,v) = g_t(x, v/2) / (256 \pi)$  leads to
\begin{align}
\begin{cases}
\partial_t h_t
&= -  v \cdot \nabla_x h_t 
+ \frac{1}{2} \left(  \vE^{\perp}_{\alpha_t} + \nabla | \cdot |^{-1} * \widetilde{\rho}_{h_t}  + v \times \left( \nabla \times \vA_{\alpha_t} \right) (x) \right) \cdot \nabla_v h_t
\\
\partial_t \vA_{\alpha_t} 
&= - \vE^{\perp}_{\alpha_t}
\\
\partial_t \vE^{\perp}_{\alpha_t} 
&=  - \Delta \vA_{\alpha_t} +  \left( 1 - \nabla \text{div}  \Delta^{-1} \right) 
4 \pi  \int_{\mathbb{R}^3}  v h_t(\cdot,v)dv
\end{cases}
\end{align}
If we define the electric and magnetic field in the usual way by 
\begin{align}
\label{eq:definition electromagnetic potentials}
\vE_t =  \vE^{\perp}_{\alpha_t} + \nabla   | \cdot |^{-1} * \widetilde{\rho}_{h_t}
\quad \text{and} \quad
\vB_t = \nabla \times \vA_t 
\end{align}
we have that the equations
\begin{align}
\nabla \cdot \vB_t = 0 
\quad \text{and} \quad
\partial_t \vB_t + \nabla \times \vE_t = 0  
\end{align}
are automatically satisfied. Using that we are working in the Coulomb gauge and $- \Delta \frac{1}{4 \pi \abs{x- y}} = \delta (x-y)$ let us obtain
\begin{align}
\nabla \cdot \vE_t = - 4 \pi \widetilde{\rho}_{h_t} .
\end{align}
Since $\nabla \times \vB = \nabla \text{div} \vA - \Delta \vA = - \Delta \vA $ and $\partial_t \nabla   | \cdot |^{-1} * \widetilde{\rho}_{h_t} =    \nabla \text{div} \Delta^{-1}  4 \pi  \int  v h_t(\cdot,v)dv$ we, moreover, get
\begin{align}
\partial_t \vE_{t} - \nabla \times \vB_t &=  4 \pi \int  v h_t(\cdot,v)dv.
\end{align}
In total, this shows that $(h_t, \vE, \vB)$ satisfies \eqref{eq: Vlasov-Maxwell literature form} with $c=1$, $e=1$ and $m=2$.


\section{Auxiliary estimates}\label{section:auxiliary estimates}


\begin{lemma}
\label{lemma:trace norm estimates with gradients and commutators}
Let $\beta \in \mathbb{N}_0^3$,  $\omega_N$ regular enough and $W_N$ denote its Wigner transform. Then,
\begin{align}
\sup \left\{ \norm{D \omega_N D}_{\textfrak{S}^1} ,  \norm{D \omega_N i \varepsilon \nabla}_{\textfrak{S}^1}, \norm{D \omega_N}_{\textfrak{S}^1} \right\} 
&\leq C N \sum_{j=0}^6 \varepsilon^j \norm{W_N}_{H_6^j} ,
\\
\norm{\left[ \hat{x} , D \omega_N D \right]}_{\textfrak{S}^1}
&\leq C N \sum_{j=1}^7 \varepsilon^j \norm{W_N}_{H_6^{j}}  , 
\\
\norm{\left[ i \nabla , D \omega_N D \right]}_{\textfrak{S}^1}
&\leq C N \sum_{j=0}^6 \varepsilon^j \norm{W_N}_{H_6^{j+1}} ,
\\
\norm{\left[ \hat{x} , D \omega_N D i \nabla \right]}_{\textfrak{S}^1}
&\leq C N \sum_{j=0}^7 \varepsilon^j \norm{W_N}_{H_7^{j+1}} ,
\\
\varepsilon^{\abs{\beta}}
\norm{ \left( 1 + x^2 \right) \nabla^{\beta} \left[ \hat{x} , \left[ \hat{x} , \omega_N \right] \right]}_{\textfrak{S}^{2}}
&\leq C N^{1/2} \sum_{j=2}^{4 + \abs{\beta}} \varepsilon^j \norm{W_N}_{H_{2 + \abs{\beta}}^{j}} ,
\\
\varepsilon^{\abs{\beta}}
\norm{ \left( 1 + x^2 \right) \nabla^{\beta} \left[ \hat{x} , \left[ \hat{x} , \omega_N \right] \right] i \varepsilon \nabla}_{\textfrak{S}^{2}}
&\leq C N^{1/2} \sum_{j=2}^{5 + \abs{\beta}} \varepsilon^j \norm{W_N}_{H_{3 + \abs{\beta}}^{j}} ,
\\
\varepsilon^{\abs{\beta}}
\norm{ \left( 1 + x^2 \right) \nabla^{\beta} \left[ \hat{x} , \left[ \hat{x} , \left\{ i \varepsilon \nabla , \omega_N \right\} \right] \right]}_{\textfrak{S}^{2}}
&\leq C N^{1/2} \sum_{j=2}^{4 + \abs{\beta}} \varepsilon^j \norm{W_N}_{H_{3 + \abs{\beta}}^{j}} ,
\\
\varepsilon^{\abs{\beta}}
\norm{ \left( 1 + x^2 \right) \nabla^{\beta} \left[ \hat{x} , \left[ \hat{x} , \left\{ i \varepsilon \nabla , \omega_N \right\} \right] \right] i \varepsilon \nabla}_{\textfrak{S}^{2}}
&\leq C N^{1/2} \sum_{j=2}^{5 + \abs{\beta}} \varepsilon^j \norm{W_N}_{H_{4 + \abs{\beta}}^{j}} ,
\\
\varepsilon^{\abs{\beta}}
\norm{ \left( 1 + x^2 \right) \nabla^{\beta} \left[ \hat{x} , \left[ i \varepsilon \nabla , \omega_N \right] \right]}_{\textfrak{S}^{2}}
&\leq C N^{1/2} \sum_{j=2}^{4 + \abs{\beta}} \varepsilon^j \norm{W_N}_{H_{2 + \abs{\beta}}^{j}} ,
\\
\varepsilon^{\abs{\beta}}
\norm{ \left( 1 + x^2 \right) \nabla^{\beta} \left[ \hat{x} , \left[ i \varepsilon \nabla , \omega_N \right] \right] i \varepsilon \nabla}_{\textfrak{S}^{2}}
&\leq C N^{1/2} \sum_{j=2}^{5 + \abs{\beta}} \varepsilon^j \norm{W_N}_{H_{3 + \abs{\beta}}^{j}} .
\end{align}
\end{lemma}

\begin{proof}
In the following, we use $\mathcal{W}[\omega_N]$ to denote the Wigner transform of an operator $\omega_N$.
Note that (see e.g. \cite[Chapter 3]{BPSS2016})
\begin{align}
\norm{\left( 1 - \varepsilon^2 \Delta \right)^{-1} \left( 1 + x^2 \right)^{-1}}_{\textfrak{S}^2}
&\leq C \sqrt{N} 
\end{align}
and 
\begin{align}
\norm{\omega_{N,s}}_{\textfrak{S}^2} &=  \left( \frac{2 \pi}{\varepsilon} \right)^{3/2} \norm{\mathcal{W}[\omega_{N,s}]}_{L^2(\mathbb{R}^6)} .
\end{align}
To obtain the inequalities of the Lemma it consequently suffices to estimate the $L^2(\mathbb{R}^6)$-norm of the respective Wigner transforms. Concerning the first inequality also note that
\begin{align}
&\norm{\left( 1 + x^2 \right) \left( 1 - \varepsilon^2 \Delta \right) D \omega_N i \varepsilon \nabla}_{\textfrak{S}^2}^2
\nonumber \\
&\quad = \tr \left( \left( 1 + x^2 \right) \left( 1 - \varepsilon^2 \Delta \right) D \omega_N \left( - \varepsilon^2 \Delta \right) \omega_N D \left( 1 - \varepsilon^2 \Delta \right)  \left( 1 + x^2 \right) 
\right)
\nonumber \\
&\quad \leq  \norm{\left( 1 + x^2 \right) \left( 1 - \varepsilon^2 \Delta \right) D \omega_N D}_{\textfrak{S}^2}^2
\end{align}
because $- \varepsilon^2 \Delta \leq D^2$ and
\begin{align}
\norm{\left( 1 + x^2 \right) \left( 1 - \varepsilon^2 \Delta \right) D \omega_N}_{\textfrak{S}^2}^2
&\leq  \norm{\left( 1 + x^2 \right) \left( 1 - \varepsilon^2 \Delta \right) D \omega_N D}_{\textfrak{S}^2}^2
\end{align}
by similar means.
Since 
\begin{align}
\mathcal{W} \left[ \left( 1 + \hat{x}^2 \right) \left( 1 - \varepsilon^2 \Delta \right)  \widetilde{\omega}_{N,s}\right]
&= \left( 1 + \frac{1}{4} \left( i \varepsilon \nabla_v + 2 x \right)^2 \right)
\left( 1 + \frac{1}{4} \left( i \varepsilon \nabla_x + 2 v \right)^2 \right) \widetilde{W}_{N,s}(x,v)
\nonumber \\
\mathcal{W}[D \widetilde{\omega}_{N,s} D](x,v)
&= \left( 1 + \left( - i \varepsilon \nabla_x + 2 v \right)^2 \right)^{1/2}
\left( 1 + \left( - i \varepsilon \nabla_x - 2 v \right)^2 \right)^{1/2} \widetilde{W}_{N,s}(x,v) ,
\nonumber \\
\mathcal{W} \left[ \left[  \nabla, \widetilde{\omega}_{N,s} \right] \right](x,v)
&=  \nabla_x \widetilde{W}_{N,s}(x,v) ,
\nonumber \\
\mathcal{W} \left[ \left[ \hat{x} , \widetilde{\omega}_{N,s} \right] \right](x,v)
&=  i \varepsilon \nabla_v \widetilde{W}_{N,s}(x,v) ,
\nonumber \\
\mathcal{W} \left[ \widetilde{\omega}_{N,s} i \nabla \right](x,v)
&=  - \frac{1}{2 \varepsilon} \left( 2 v + i \varepsilon \nabla_x \right) \widetilde{W}_{N,s}(x,v) 
\end{align}
we have
\begin{align}
&\mathcal{W} \left[\left( 1 + x^2 \right) \left( 1 - \varepsilon^2 \Delta \right) D \widetilde{\omega}_{N,s}  D \right](x,v)
\nonumber \\
&\quad =
\left( 1 + \frac{1}{4} \left( i \varepsilon \nabla_v + 2 x \right)^2 \right)
\left( 1 + \frac{1}{4} \left( i \varepsilon \nabla_x + 2 v \right)^2 \right)
\left( 1 + \left( - i \varepsilon \nabla_x + 2 v \right)^2 \right)^{1/2}
\nonumber \\
&\qquad \qquad \qquad \times
\left( 1 + \left( - i \varepsilon \nabla_x - 2 v \right)^2 \right)^{1/2} \widetilde{W}_{N,s}(x,v) ,
\end{align}
\begin{align}
&\mathcal{W} \left[ \left( 1 + x^2 \right) \left( 1 - \varepsilon^2 \Delta \right) \left[ \hat{x} , D \widetilde{\omega}_{N,s} D  \right] \right](x,v)
\nonumber \\
&\quad =
\left( 1 + \frac{1}{4} \left( i \varepsilon \nabla_v + 2 x \right)^2 \right)
\left( 1 + \frac{1}{4} \left( i \varepsilon \nabla_x + 2 v \right)^2 \right)
i \varepsilon \nabla_v 
\left( 1 + \left( - i \varepsilon \nabla_x + 2 v \right)^2 \right)^{1/2}
\nonumber \\
&\qquad \qquad \qquad \times
\left( 1 + \left( - i \varepsilon \nabla_x - 2 v \right)^2 \right)^{1/2} \widetilde{W}_{N,s}(x,v) ,
\end{align}
\begin{align}
&\mathcal{W} \left[ \left( 1 + x^2 \right) \left( 1 - \varepsilon^2 \Delta \right) \left[  \nabla  , D \widetilde{\omega}_{N,s} D \right]  \right](x,v)
\nonumber \\
&\quad =
\left( 1 + \frac{1}{4} \left( i \varepsilon \nabla_v + 2 x \right)^2 \right)
\left( 1 + \frac{1}{4} \left( i \varepsilon \nabla_x + 2 v \right)^2 \right) \nabla_x
\left( 1 + \left( - i \varepsilon \nabla_x + 2 v \right)^2 \right)^{1/2}
\nonumber \\
&\qquad \qquad \qquad \times
\left( 1 + \left( - i \varepsilon \nabla_x - 2 v \right)^2 \right)^{1/2} \widetilde{W}_{N,s}(x,v)
\end{align}
and
\begin{align}
&\mathcal{W} \left[ \left( 1 + x^2 \right) \left( 1 - \varepsilon^2 \Delta \right) \left[ \hat{x} , D \widetilde{\omega}_{N,s} D  i \nabla \right] \right](x,v)
\nonumber \\
&\quad = - \frac{1}{2}
\left( 1 + \frac{1}{4} \left( i \varepsilon \nabla_v + 2 x \right)^2 \right)
\left( 1 + \frac{1}{4} \left( i \varepsilon \nabla_x + 2 v \right)^2 \right)
i \nabla_v \left( 2 v + i \varepsilon \nabla_x \right)
\left( 1 + \left( - i \varepsilon \nabla_x + 2 v \right)^2 \right)^{1/2}
\nonumber \\
&\qquad \qquad \qquad \times
\left( 1 + \left( - i \varepsilon \nabla_x - 2 v \right)^2 \right)^{1/2} \widetilde{W}_{N,s}(x,v) .
\end{align}
This leads to 
\begin{align}
\norm{\mathcal{W} \left[\left( 1 + x^2 \right) \left( 1 - \varepsilon^2 \Delta \right) D \widetilde{\omega}_{N,s}  D \right]}_{L^2(\mathbb{R}^6)}
&\leq C \sum_{j=0}^6 \varepsilon^{j} \norm{\widetilde{W}_{N,s}}_{H_6^j} ,
\end{align}
\begin{align}
\norm{\mathcal{W} \left[ \left( 1 + x^2 \right) \left( 1 - \varepsilon^2 \Delta \right) \left[ \hat{x} , D \widetilde{\omega}_{N,s} D   \right] \right]}_{L^2(\mathbb{R}^6)}
&\leq C \sum_{j=1}^7  \varepsilon^{j} \norm{\widetilde{W}_{N,s}}_{H_6^{j}}  ,
\end{align}
\begin{align}
\norm{\mathcal{W} \left[ \left( 1 + x^2 \right) \left( 1 - \varepsilon^2 \Delta \right) \left[ i \nabla  , D \widetilde{\omega}_{N,s} D \right]  \right]}_{L^2(\mathbb{R}^6)}
&\leq C \sum_{j=0}^6 \varepsilon^{j} \norm{\widetilde{W}_{N,s}}_{H_6^{j+1}} ,
\end{align}
\begin{align}
\norm{\mathcal{W} \left[ \left( 1 + x^2 \right) \left( 1 - \varepsilon^2 \Delta \right) \left[ \hat{x} , D \widetilde{\omega}_{N,s} D   \nabla \right] \right]}_{L^2(\mathbb{R}^6)}
&\leq C \sum_{j=0}^7 \varepsilon^{j} \norm{\widetilde{W}_{N,s}}_{H_7^{j+1}} .
\end{align}
The remaining relations follow by similar means.
\end{proof}

\noindent{\it Acknowledgments.} 
The authors would like to thank Fran\c{c}ois Golse for interesting discussions about the Vlasov--Maxwell equations. N.L., moreover, would like to thank Marco Falconi for fruitful discussions about the properties of the charge distribution within project \cite{FL2022}. Support from the Swiss National Science Foundation through the NCCR SwissMAP (N.L. and C.S.) and funding from the European Union's Horizon 2020 research and innovation programme (N.L) through the Marie Sk{\l}odowska-Curie Action EFFECT (grant agreement No. 101024712) as well as the Swiss National Science Foundation (C.S.) through the SNSF Eccellenza project PCEFP2 181153 is gratefully acknowledged.


{}

\end{document}